\documentclass[reqno]{amsart}
\usepackage[a4paper]{geometry}
\usepackage[T1]{fontenc} 
\usepackage[utf8]{inputenc}
\usepackage{amssymb,bm}
\usepackage[mathscr]{eucal}
\usepackage{mathrsfs}
\usepackage[english]{babel}
\usepackage{mathtools}
\usepackage{microtype}
\usepackage{enumerate,paralist}
\usepackage{xcolor}
\usepackage[colorlinks=true]{hyperref}
\hypersetup{urlcolor=blue,citecolor=red,linkcolor=blue}
\usepackage[initials]{amsrefs}
\numberwithin{equation}{section}
\theoremstyle{plain}
\newtheorem{theorem}{Theorem}[section]
\newtheorem{corollary}[theorem]{Corollary}
\newtheorem{proposition}[theorem]{Proposition}
\newtheorem{lemma}[theorem]{Lemma}
\theoremstyle{definition}
\newtheorem{example}[theorem]{Example}

\newtheorem*{assumptions*}{Assumptions}
\newtheorem*{acknowledgements*}{Acknowledgements}
\theoremstyle{remark}
\newtheorem{remark}[theorem]{Remark}
\newtheorem*{remark*}{Remark}
\newtheorem*{note*}{Note}
\newtheorem*{stepone*}{Step 1}
\newtheorem*{steptwo*}{Step 2}
\newtheorem*{stepthree*}{Step 3}
\newtheorem*{mpoweroftwo*}{The case $m=2^\ell$}
\newtheorem*{marbitrary*}{The general case}
\providecommand{\B}[1]{\mathbf{#1}}

\providecommand{\C}[1]{\mathcal{#1}}
\providecommand{\CS}[1]{\mathscr{#1}}
\providecommand{\D}[1]{\mathbb{#1}}
\providecommand{\F}[1]{\mathfrak{#1}}
\providecommand{\R}[1]{\mathrm{#1}}
\newcommand{\dd}{\mathrm{d}}
\newcommand{\eul}{\mathrm{e}}
\newcommand{\ii}{\mathrm{i}}
\newcommand{\id}{\mathrm{Id}}
\newcommand{\expect}{\mathbb{E}}
\newcommand{\four}{\mathscr{F}}
\newcommand{\prob}{\mathbb{P}}
\newcommand{\ac}{\mathrm{ac}}
\newcommand{\cont}{\mathrm{c}}
\newcommand{\fb}{\mathrm{fb}}
\newcommand{\pp}{\mathrm{p}}
\newcommand{\ord}{\mathrm{O}}
\newcommand{\osmall}{\mathrm{o}}
\newcommand{\spec}{\sigma}
\providecommand{\abs}[1]{\lvert#1\rvert}
\providecommand{\accol}[1]{\lbrace#1\rbrace}
\providecommand{\croch}[1]{\lbrack#1\rbrack}
\providecommand{\norm}[1]{\lVert#1\rVert}
\providecommand{\scal}[1]{\langle#1\rangle}
\renewcommand{\vec}[1]{\mathbf{#1}}
\renewcommand{\Im}{\mathop{\rm Im}\nolimits}

\begin{document}
%---------------------------------------------------------%
\title[Ballistic transport in periodic and random media]{Ballistic transport in periodic and random media}
%---------------------------------------------------------%
\author[A. Boutet de Monvel]{Anne Boutet de Monvel}
\address{AB: Institut de Math\'ematiques de Jussieu-Paris Rive Gauche, Universit\'e de Paris, 8 place Aur\'elie Nemours, case 7012, 75205 Paris Cedex 13, France.}
\email{anne.boutet-de-monvel@imj-prg.fr}
%---------------------------------------------------------%
\author[M. Sabri]{Mostafa Sabri}
\address{MS: Department of Mathematics, Faculty of Science, Cairo University, Giza 12613, Egypt.}
\email{mmsabri@sci.cu.edu.eg}
%---------------------------------------------------------%
\subjclass[2020]{Primary: 81Q10; Secondary: 46N50, 47B80}
\keywords{Ballistic transport, delocalization, periodic Schr\"odinger operators, periodic graphs, trees}
%---------------------------------------------------------%
\date{}
%---------------------------------------------------------%
\begin{abstract}
We prove ballistic transport of all orders, that is, $\norm{x^m\eul^{-\ii tH}\psi}\asymp t^m$, for the following models: the adjacency matrix on $\D{Z}^d$, the Laplace operator on $\D{R}^d$, periodic Schr\"odinger operators on $\D{R}^d$, and discrete periodic Schr\"odinger operators on periodic graphs. In all cases we give the exact expression of the limit of $\norm{x^m\eul^{-\ii tH}\psi}/t^m$ as $t\to+\infty$. We then move to universal covers of finite graphs (these are infinite trees) and prove ballistic transport in mean when the potential is lifted naturally, giving a periodic model, and when the tree is endowed with random i.i.d.\ potential, giving an Anderson model. The limiting distributions are then discussed, enriching the transport theory. Some general upper bounds are detailed in the appendix.
\end{abstract}
%---------------------------------------------------------%
\maketitle
%---------------------------------------------------------%
%:s.1
%---------------------------------------------------------%
\section{Introduction}  \label{sec:introduction}
%---------------------------------------------------------%
%:s.1.1
%---------------------------------------------------------%
\subsection{Background}  \label{subsec:background}
Ballistic motion is a dynamical property found in certain delocalized Schr\"o\-dinger operators. 

Broadly speaking, in considering a Schr\"odinger operator, say $H=-\Delta +V$ on $L^2(\D{R}^d)$, we can understand \emph{localization} in some spectral interval $I$ in three senses: (a) Spectral localization means that $\spec(H)\cap I$ is pure point. (b) Exponential localization means that moreover the eigenfunctions decay exponentially. (c) Dynamical localization means that states initially localized in a bounded domain will not leave this domain much as time goes on. 

The opposite regime of \emph{delocalization} is similarly understood in three senses: (a') spectrally, one expects absolutely continuous (AC) spectrum, (b') spatially, the (generalized) eigenfunctions may ideally be equidistributed in some sense, (c') dynamically, one expects the wave packets to spread on the space as time goes on.

The RAGE theorem establishes some links between \emph{spectral} and \emph{dynamical} (de)localization. Recall that if $\psi_0\in L^2(\D{R}^d)$ is some initial state of the system, then $\eul^{-\ii tH}\psi_0$ describes the state of the system at time $t$. Now define the following subspaces of $\CS{H}\coloneqq L^2(\D{R}^d)$:
\[
\CS{H}_{\pp}\coloneqq\overline{\text{span}\accol{\text{eigenvectors of }H}}\quad \text{and} \quad \CS{H}_\cont\coloneqq\CS{H}_{\pp}^{\bot}.
\]
The RAGE theorem (see \cite{Te14}*{Section 5.2}) asserts that
\begin{align}\label{eq:ragepp}
f\in\CS{H}_{\pp}
&\iff \forall\varepsilon>0\ \exists K\subset\D{R}^d \text{ compact s.t.}\ \sup_{t\in\D{R}}\,\norm{\chi_{K^c}\eul^{-\ii tH}f}\leq\varepsilon,\\
\label{eq:ragec}
f\in\CS{H}_{\cont}
&\iff \forall K\subset\D{R}^d\text{ compact}\colon \lim_{T\to+\infty}\frac{1}{T}\int_{-T}^T\norm{\chi_K\eul^{-\ii tH}f}^2\,\dd t=0.
\end{align}
Hence: (i) A state is in the pure point space $\CS{H}_{\pp}$ iff at all times, most of its mass lies within a fixed compact set. (ii) A state is in the continuous space $\CS{H}_{\cont}$ iff its time evolution escapes (in average) from any compact set, after sufficient time has passed.

Exploring further the dynamical aspects of delocalization, a common object of study is the \emph{mean square displacement} $r_f^2(t)\coloneqq\norm{x\eul^{-\ii tH}f}^2$. In presence of AC spectrum, i.e., if the absolutely continuous subspace $\CS{H}_\ac$ is nontrivial, one can deduce from the RAGE theorem\footnote{It is known \cite{Te14} that for $f\in\CS{H}_{\ac}$, \eqref{eq:ragec} strengthens to the fact that $\lim_{t\to+\infty}\norm{\chi_K\eul^{-\ii tH}f}=0$. This implies that $\lim_{t\to+\infty}\norm{\chi_{K^c}\eul^{-\ii tH}f}=\lim_{t\to+\infty}\norm{\eul^{-\ii tH}f}=\norm{f}$. Taking $K=\Lambda_r\coloneqq\accol{\abs{x}\leq r}$, we thus get $\liminf_{t\to+\infty}\norm{x^m\eul^{-\ii tH}f}\geq \liminf_{t\to+\infty}r^m\norm{\chi_{\Lambda_r^c}\eul^{-\ii tH}f}= r^m\norm{f}$. As $r$ is arbitrary, this shows $\lim_{t\to\infty} \norm{x^m\eul^{-\ii tH}f}=\infty$.} that $\lim\limits_{t\to\infty}\norm{x^m\eul^{-\ii tH}f}=\infty$ for any $f\in\CS{H}_\ac$, $f\neq 0$.

In this article we are interested in the rate of this divergence. It is known that $r_f^2(t)\lesssim t^2$, see Theorem~\ref{thm:momconts} for a more general result. The operator exhibits \emph{ballistic transport} if we also have $r_f^2(t)\gtrsim t^2$, at least in a time-averaged sense.

Here we discuss several models that exhibit the exact ballistic rate, that is, $\norm{x^m\eul^{-\ii tH}f}\asymp t^m$ for $f\in\CS{H}_\ac$, $f\neq 0$ and all $m\geq 1$.
%---------------------------------------------------------%
%:s.1.2
%---------------------------------------------------------%
\subsection{Contents}\label{subsec:discussion}
Our aim is to provide a comprehensive outlook of various techniques to prove ballistic motion for different models.

In Section \ref{sec:laplacian} we establish ballistic transport for the adjacency matrix on $\D{Z}^d$ (Section~\ref{subsec:Zd}), followed by the Laplace operator on $\D{R}^d$ (Section~\ref{subsec:Rd}). This is of course folklore, however we are not aware of references computing the exact limit of $\norm{x^m\eul^{-\ii tH}\psi}^2/t^{2m}$ as $t\to+\infty$. 

Then, in Section~\ref{subsec:ak}, we move on to Schr\"odinger operators with periodic potential on $\D{R}^d$, extending the result of \cite{AsKn98} to all $m$. After that, in Section~\ref{subsec:perdiscrete}, we proceed to the less studied discrete periodic Schr\"odinger operators on periodic graphs. This includes the hexagonal lattice, face-centered and body-centered cubic lattices, and a lot more. For these models we establish ballistic motion, assuming $\psi$ has a nontrivial AC component, i.e., its spectral projection onto $\C{H}_\ac$ is nontrivial.

In Section~\ref{sec:uc} we move to infinite trees that enjoy some form of periodicity. More specifically, we consider a finite graph $G$ endowed with some potential $W$, then consider the universal cover $\C{T}=\widetilde{G}$ endowed with the naturally lifted potential in Section~\ref{sec:ucper}. We proceed afterwards to replace such lifted potentials by random i.i.d.\ potentials in Section~\ref{sec:ucand}, which gives an Anderson model. For both models (periodic and random) we establish ballistic transport \emph{in mean}, that is, we consider $2\eta\int_0^\infty\eul^{-2\eta t}\norm{\abs{x}^{\beta/2}\eul^{-\ii tH}\psi}^2\,\dd t$ instead, with $\eta\downarrow 0$. The periodic result seems new, the random result appeared before in the special case of the $d$-regular tree in \cites{AiWa12,Kl96}.

In Section~\ref{sec:epi} we refine our results for the periodic models of Section~\ref{sec:periodic} by studying the limiting distribution of $X_t/t$, when $X_t$ has distribution $\abs{\eul^{-\ii tH}\psi(x)}^2$ and $\psi$ is normalized. Note that if $\nu_t^\psi$ is the distribution of $X_t/t$, then ballistic motion is merely the statement that the second moment $\lim\limits_{t\to+\infty}\expect_{\nu_t^\psi}(x^2)$ is $>0$. Here we actually compute the nontrivial limit $\nu_\infty^\psi$ of the measures $\nu_t^\psi$ as $t\to+\infty$. This gives a richer understanding of the transport theory; in particular, this gives an expression for the limit $\lim\limits_{t\to+\infty}\expect_{\nu_t^\psi}(f)=\expect_{\nu_\infty^\psi}(f)$ for any continuous $f$.

Appendix~\ref{sec:app} contains complete proofs of the ballistic upper bounds, which is well-known \cite{RaSi78} for $m=1,2$. We also prove some differentiability results for the moments in the Heisenberg picture, i.e., for the maps $t\mapsto\eul^{\ii tH}x^m\eul^{-\ii tH}$, which are needed in the text.

To conclude this introduction we refer to the beautiful paper \cite{Last96} for finer interplay between the spectral measures and quantum dynamics. We also mention the recent papers \cites{DGS21,KPS21} which prove ballistic transport for quasiperiodic Schr\"odinger operators, under certain conditions. 
%---------------------------------------------------------%
%:s.2
%---------------------------------------------------------%
\section{The free Laplacian}  \label{sec:laplacian}
%---------------------------------------------------------%
%:s.2.1
%---------------------------------------------------------%
\subsection{Discrete case}  \label{subsec:Zd}

We consider the integer lattice $\D{Z}^d$ and the space $\ell^2(\D{Z}^d)$ of square summable sequences $\psi\colon\D{Z}^d\to\D{C}$ equipped with the scalar product $\scal{\phi,\psi}\coloneqq\sum_{n\in\D{Z}^d}\overline{\phi(n)}\psi(n)$. On the other hand we consider the torus $\D{T}^d=\D{R}^d/(2\pi\D{Z})^d=[0,2\pi)^d$ and the space $L^2(\D{T}^d)$ of square integrable functions $f\colon\D{T}^d\to\D{C}$ equipped with the scalar product $\scal{f,g}\coloneqq\int_{\D{T}^d}\overline{f(\theta)}g(\theta)\dd\theta$. The Fourier transform $\four\colon L^2(\D{T}^d)\to\ell^2(\D{Z}^d)$ is defined by
\[
(\four f)(n)\equiv\hat f(n)\coloneqq\frac{1}{(\sqrt{2\pi})^d}\int_{\D{T}^d}f(\theta)\eul^{-\ii n\cdot\theta}\dd\theta=\scal{e_n,f},
\]
where $e_n(\theta)\coloneqq\frac{\eul^{\ii n\cdot\theta}}{(\sqrt{2\pi})^d}$. In particular, $\four e_n=\delta_n$ where $\accol{\delta_n}_{n\in\D{Z}^d}$ is the standard Hilbert basis of $\ell^2(\D{Z}^d)$. The inverse Fourier transform $\four^{-1}\colon\ell^2(\D{Z}^d)\to L^2(\D{T}^d)$ is given by
\[
\croch{\four^{-1}(\psi)}(\theta)=\hat\psi(\theta)=\frac{1}{(\sqrt{2\pi})^d}\sum_{n\in\D{Z}^d}\psi(n)\eul^{\ii n\cdot\theta}=\sum_{n\in\D{Z}^d}\psi(n)e_n(\theta).
\]

In this section we consider the adjacency matrix $\C{A}$ acting in $\ell^2(\D{Z}^d)$, that is,
\[
(\C{A}\psi)(n)=\sum_{w\sim n}\psi(w)=\sum_{j=1}^d\bigl(\psi(n-\F{e}_j)+\psi(n+\F{e}_j)\bigr)
\]
where $\accol{\F{e}_j}_{j=1}^d$ is the standard basis of $\D{Z}^d$ and $w\sim v$ means that the vertices $v,w\in\D{Z}^d$ are nearest neighbors. By Fourier transform $\psi(n-\F{e}_j)+\psi(n+\F{e}_j)$ becomes $(\eul^{\ii\theta_j}+\eul^{-\ii\theta_j})\hat\psi(\theta)$. So $\hat{\C{A}}\coloneqq\four^{-1}\C{A}\four$ is the operator of multiplication by $\phi(\vec{\theta})\coloneqq\sum_{j=1}^d2\cos\theta_j$ in $L^2(\D{T}^d)$ and the spectrum of $\C{A}$ is absolutely continuous: $\spec(\C{A})=\spec_\ac(\C{A})=\croch{-2d,2d}$.

Given $n\in\D{Z}^d$, $\psi\colon\D{Z}^d\to\D{C}$, and $m\in\D{N}$, let $n^m\psi\colon\D{Z}^d\to\D{C}^d$ denote the map defined by $n^m\psi(n)\coloneqq(n_1^m\psi(n),\dots,n_d^m\psi(n))$.

%-------------------%
%:thm 2.1
%-------------------%
\begin{theorem}   \label{thm:freediscrete}
For any $\psi\in\ell^2(\D{Z}^d)$ with $\norm{n^m\psi}<\infty$, $m\in\D{N}$ we have
\begin{equation}  \label{eq:ballizd}
\lim_{t\to+\infty}\frac{\norm{n^m\eul^{\ii t\C{A}}\psi}^2}{t^{2m}}=d\binom{2m}{m}\norm{\psi}^2+\sum_{q=0}^{m-1}\binom{2m}{q}(-1)^{q+m}\scal{\psi,S_{2m-2q}\psi},
\end{equation} 
where $(S_k\psi)(n)\coloneqq\sum_{j=1}^d\bigl(\psi(n-k\F{e}_j)+\psi(n+k\F{e}_j)\bigr)$ is a $k$-step adjacency matrix. 

Moreover, this limit is $>0$ if $\psi\neq 0$ and we have the upper bound
\begin{equation}  \label{eq:Zd-upperbound}
\lim_{t\to+\infty}\frac{\norm{n^m\eul^{\ii t\C{A}}\psi}^2}{t^{2m}}\leq 4^md\norm{\psi}^2. 
\end{equation} 
\end{theorem}
%-------------------%

This somehow suggests that $\frac{\eul^{-\ii t\C{A}}n^{2m}\eul^{\ii t\C{A}}}{t^{2m}}$ converges weakly to the bounded operator $d\binom{2m}{m}\id+\sum_{q=0}^{m-1}\binom{2m}{q}(-1)^{q+m}S_{2m-2q}$.

%-------------------%
\begin{proof}
Denote $\psi_t\coloneqq\eul^{\ii t\C{A}}\psi$. We first prove that $\norm{n^m\psi}<\infty$ implies $\norm{n^m\psi_t}<\infty$ for any $t\geq 0$. By Fourier transform $\widehat{n^m\psi}=D^m\hat\psi$ where $D^m\coloneqq(-\ii)^m(\partial_{\theta_1}^m,\dots,\partial_{\theta_d}^m)$. So $\norm{n^m\psi}<\infty$ means $\hat\psi\in H^m(\D{T}^d)$. Moreover, since $\hat{\C{A}}$ is the operator of multiplication by $\phi(\vec{\theta})=\sum_{j=1}^d2\cos\theta_j$ it leaves invariant $H^m(\D{T}^d)$ and similarly for $\widehat{\eul^{\ii t\C{A}}}=\eul^{\ii t\hat{\C{A}}}$. So, $\hat\psi_t=\widehat{\eul^{\ii t\C{A}}}\hat\psi\in H^m(\D{T}^d)$ and $\widehat{n^m\psi_t}=D^m\hat\psi_t\in L^2(\D{T}^d)$ which means that $n^m\psi_t\in\ell^2(\D{Z}^d)$.

The limit in \eqref{eq:ballizd} is strictly positive if $\psi\neq 0$ as a consequence of \eqref{eq:recall} below. We then show how formula \eqref{eq:ballizd} implies the upper bound \eqref{eq:Zd-upperbound} which is a slight improvement over the general bound \eqref{eq:discretbound} of Appendix \ref{sec:app}. By using $\norm{S_k}\leq 2d$ to estimate the right-hand side of \eqref{eq:ballizd} we get
\[
\lim_{t\to+\infty}\frac{\norm{n^m\eul^{\ii t\C{A}}\psi}^2}{t^{2m}}\leq\biggl(d\binom{2m}{m}+2d\sum_{q=0}^{m-1}\binom{2m}{q}\biggr)\norm{\psi}^2=\biggl(d\sum_{q=0}^{2m}\binom{2m}{q}\biggr)\norm{\psi}^2=(1+1)^{2m}d\norm{\psi}^2.
\]

Now we give the proof of formula \eqref{eq:ballizd}. Let $\psi_t\coloneqq\eul^{\ii t\C{A}}\psi$ be as above. Since $\hat{\C{A}}$ is the operator of multiplication by $\phi(\theta)\coloneqq 2\sum_{j=1}^d\cos\theta_j$, we have
\[
\widehat{n^m\psi_t}=D^m\hat\psi_t=D^m(\eul^{\ii t\hat{\C{A}}}\hat\psi)=D^m(\eul^{\ii t\phi}\hat\psi).
\]
In particular,
\begin{equation}\label{eq:proof2nor}
\norm{n^m\eul^{\ii t\C{A}}\psi}_{\ell^2(\D{Z}^d)}=\norm{D^m(\eul^{\ii t\phi}\hat\psi)}_{L^2(\D{T}^d)}.
\end{equation}
By Leibniz,
\begin{equation}\label{eq:paref-1}
\partial_{\theta_j}^m\bigl(\eul^{\ii t\phi(\theta)}\hat\psi(\theta)\bigr)=\sum_{r=0}^m\binom{m}{r}\partial_{\theta_j}^r\eul^{\ii t\phi(\theta)}\partial_{\theta_j}^{m-r}\hat\psi(\theta).
\end{equation}
Clearly the leading term in $t$ is for $r=m$ and we have
\begin{equation}\label{eq:proof2main}
\partial_{\theta_j}^m\eul^{\ii t\phi(\theta)}=(-2\ii t\sin\theta_j)^m\eul^{\ii t\phi(\theta)}+\ord(t^{m-1})
\end{equation}
for some $\ord(t^{m-1})$ which can be made explicit. Since $\hat\psi\in H^m(\D{T}^d)$, we deduce from \eqref{eq:proof2nor}-\eqref{eq:proof2main} that
\begin{equation}\label{eq:recall}
\norm{n^m\eul^{\ii t\C{A}}\psi}^2=4^mt^{2m}\sum_{j=1}^d\norm{(\sin\theta_j)^m\hat\psi}^2+\ord(t^{2m-1}).
\end{equation}
Now we observe that 
\[
(\sin\theta)^{2m}=\frac{(\eul^{\ii\theta}-\eul^{-\ii\theta})^{2m}}{(2\ii)^{2m}}=\frac{(-1)^m}{4^m}\sum_{q=0}^{2m}\binom{2m}{q}\eul^{\ii q\theta}\eul^{-\ii\theta(2m-q)}(-1)^{2m-q}. 
\]
Using $(-1)^{2m-q}=(-1)^q$ and $\binom{2m}{2m-q}=\binom{2m}{q}$, we thus have
\[
4^m(\sin\theta)^{2m}=\binom{2m}{m}+\sum_{q=0}^{m-1}\binom{2m}{q}(-1)^{q+m}(2\cos(2m-2q)\theta)\,.
\]
Recall that $S_k\colon\ell^2(\D{Z}^d)\to\ell^2(\D{Z}^d)$ is defined by
\[
(S_k\psi)(n)\coloneqq\sum_{j=1}^d\bigl(\psi(n-k\F{e}_j)+\psi(n+k\F{e}_j)\bigr).
\]
By Fourier transform $\psi(n-k\F{e}_j)+\psi(n+k\F{e}_j)$ becomes $(\eul^{\ii k\theta_j}+\eul^{-\ii k\theta_j})\hat\psi$. So $\hat S_k\coloneqq\four^{-1}S_k\four$ is the operator of multiplication by the function $\phi_k$ in $L^2(\D{T}^d)$, with $\phi_k(\theta)\coloneqq\sum_{j=1}^d 2\cos k\theta_j$. It follows that
\begin{align*}
4^m\sum_{j=1}^d\norm{(\sin\theta_j)^m\hat\psi}^2 
&=\Bigl\langle\hat\psi,4^m\sum_{j=1}^d(\sin\theta_j)^{2m}\hat\psi\Bigr\rangle\\
&=d\binom{2m}{m}\norm{\hat\psi}^2+\sum_{q=0}^{m-1}\binom{2m}{q}(-1)^{q+m}\scal{\hat\psi,\phi_{2m-2q}\hat\psi}\\
&=d\binom{2m}{m}\norm{\hat\psi}^2+\sum_{q=0}^{m-1}\binom{2m}{q}(-1)^{q+m}\scal{\hat\psi,\widehat{S_{2m-2q}}\hat\psi}\\
&=d\binom{2m}{m}\norm{\psi}^2+\sum_{q=0}^{m-1}\binom{2m}{q}(-1)^{q+m}\scal{\psi,S_{2m-2q}\psi}.
\end{align*}
Recalling \eqref{eq:recall}, dividing by $t^{2m}$ and taking $t\to+\infty$, we obtain the statement.
\end{proof}
%-------------------%

%-------------------%
%:rem 2.2
%-------------------%
\begin{remark}  \label{rem:Zd}
For $m=1$ we get $\lim_{t\to+\infty}\frac{\norm{n\eul^{\ii t\C{A}}\psi}^2}{t^2}=2d\norm{\psi}^2-\scal{\psi,S_2\psi}$. In particular a bound like $\norm{n\eul^{\ii t\C{A}}\psi}\geq Ct\norm{\psi}$ cannot be valid for a positive constant $C$ independent of $\psi$. In fact, let $d=1$. Consider $\psi=(\psi_j)$ given by $\psi_1=\dots=\psi_n=\frac{1}{\sqrt{n}}$ and $\psi_j=0$ otherwise. Then $\norm{\psi}^2=1$ and $\sum_j\psi_j\overline{\psi_{j+2}}=\psi_1\overline{\psi_3}+\dots+\psi_{n-2}\overline{\psi_n}+\psi_{n-1}\overline{\psi_{n+1}}+\psi_n\overline{\psi_{n+2}}=\frac{1}{n}+\dots+\frac{1}{n}+0+0=\frac{n-2}{n}$. So in this case the limit is $2(1-\frac{n-2}{n})=\frac{4}{n}$, which can get arbitrarily small while conserving $\norm{\psi}=1$.
\end{remark}
%-------------------%

%---------------------------------------------------------%
%:s.2.2
%---------------------------------------------------------%
\subsection{Continuous case}  \label{subsec:Rd}
Consider now the case of the Laplacian $\Delta$ on Euclidean space $\D{R}^d$. The space $L^2(\D{R}^d)$ is equipped with the scalar product $\scal{\phi,\psi}\coloneqq\int_{\D{R}^d}\overline{\phi(x)}\psi(x)\dd x$ and the Fourier transform $\four\colon L^2(\D{R}^d)\to L^2(\D{R}^d)$ is defined by
\[
\four(\psi)(y)\equiv\hat\psi(y)\coloneqq\frac{1}{(\sqrt{2\pi})^d}\int_{\D{R}^d}\psi(x)\eul^{-\ii y\cdot x}\dd x.
\]
For $x=(x_1,\dots,x_d)\in\D{R}^d$ we denote $x^k\coloneqq(x_1^k,\dots,x_d^k)$ and $x^k\psi(x)\coloneqq(x_1^k\psi(x),\dots,x_d^k\psi(x))$. Let also $D^k\coloneqq(-\ii)^k(\partial_{x_1}^k,\dots,\partial_{x_d}^k)$ and $\Delta\coloneqq\partial_{x_1}^2+\dots+\partial_{x_d}^2$.

%-------------------%
%:thm 2.3
%-------------------%
\begin{theorem}   \label{thm:cts}
Assume $\psi\in L^2(\D{R}^d)$ satisfies $\sum_{k=0}^m\norm{x^{m-k}D^k\psi}^2<\infty$. Then
\[
\lim_{t\to+\infty}\frac{\norm{x^m\eul^{\ii t\Delta}\psi}^2}{t^{2m}}=4^m\norm{D^m\psi}^2.
\] 
\end{theorem}
%-------------------%

Let $\vec{m}=(m_1,\dots,m_d)$, $m\coloneqq\max\accol{m_1,\dots,m_d}$, and $x^{\vec{m}}\coloneqq(x_1^{m_1},\dots,x_d^{m_d})$. The proof shows more generally that
\[
\lim_{t\to+\infty}\frac{\norm{x^{\vec{m}}\eul^{\ii t\Delta}\psi}^2}{t^{2m}}=4^m\sum_{\substack{1\leq j\leq d\\m_j=m}}\norm{\partial_{x_j}^m\psi}^2.
\]

%-------------------%
\begin{proof}
Using the Fourier transform one sees that
\[
\eul^{\ii t\Delta}\psi(x)=\frac{\eul^{\ii x\cdot x/4t}}{(4\pi\ii t)^{d/2}}\int_{\D{R}^d}\eul^{\ii y\cdot y/4t}\psi(y)\eul^{-\ii x\cdot y/2t}\dd y,
\]
see, e.g., \cite{Te14}*{Section 7.4, (7.43)}. Denoting $\phi_t(y)\coloneqq\eul^{\ii y\cdot y/4t}\psi(y)$, this can be written as
\begin{equation}\label{eq:semicts}
\eul^{\ii t\Delta}\psi(x)=\frac{\eul^{\ii x\cdot x/4t}}{(2\ii t)^{d/2}}\,\hat\phi_t\Bigl(\frac{x}{2t}\Bigr).
\end{equation}
Thus,
\begin{align*}
\norm{x^m\eul^{\ii t\Delta}\psi}
&=\frac{1}{(2t)^{d/2}}\Bigl\lVert x^m\hat\phi_t\Bigl(\frac{\cdot}{2t}\Bigr)\Bigr\rVert
=(2t)^m\frac{1}{(2t)^{d/2}}\Bigl\lVert\Bigl(\frac{x}{2t}\Bigr)^m\hat\phi_t\Bigl(\frac{\cdot}{2t}\Bigr)\Bigr\rVert\\
&=(2t)^m\norm{x^m\hat\phi_t}
=(2t)^m\norm{\widehat{D^m\phi_t}}=(2t)^m\norm{D^m\phi_t}.
\end{align*}
By Leibniz formula,
\begin{equation}  \label{eq:leibniz}
\partial_{y_j}^m\phi_t(y)=\sum_{k=0}^m \binom{m}{k} \partial_{y_j}^{m-k}\eul^{\ii y\cdot y/4t}\partial_{y_j}^k\psi(y)=\eul^{\ii y\cdot y/4t}\partial_{y_j}^m\psi(y)+\ord_{\psi}(t^{-1}).
\end{equation}
The error term $\ord_{\psi}(t^{-1})$ depends on $\psi$ and its derivatives. It can be made explicit as follows. The Fa\`a di Bruno formula implies that
\[
\partial_{y_j}^p\eul^{\ii y\cdot y/4t}=\eul^{\ii y\cdot y/4t}\sum_{r=0}^{\lfloor p/2\rfloor}\frac{p!}{(p-2r)!r!}\frac{\ii^{p-r}}{2^pt^{p-r}}y_j^{p-2r}. 
\]
Thus,
\[
\ord_\psi(t^{-1})=\eul^{\ii y\cdot y/4t}\sum_{k=0}^{m-1}\binom{m}{k}\sum_{r=0}^{\lfloor(m-k)/2\rfloor}\frac{(m-k)!}{(m-k-2r)!r!}\frac{\ii^{m-k-r}}{2^{m-k}t^{m-k-r}}y_j^{m-k-2r}\partial_{y_j}^k\psi(y).
\]
By \eqref{eq:leibniz} and for any $x$, $\abs{\partial_{x_j}^m\phi_t(x)}\to\abs{\partial_{x_j}^m\psi(x)}$ as $t\to+\infty$. The explicit form of $\ord_\psi(t^{-1})$ and our assumption $\sum_{k=0}^m\norm{x^{m-k}D^k\psi}^2<\infty$ allow to use dominated convergence to conclude that $\norm{D^m\phi_t}^2=\sum_{j=1}^d\norm{\partial_{x_j}^m\phi_t}^2\to\sum_{j=1}^d\norm{\partial_{y_j}^m\psi}^2=\norm{D^m\psi}^2$.
\end{proof}
%-------------------%

%-------------------%
%:rem 2.4
%-------------------%
\begin{remark}
The regularity of $\psi$ above is important. In fact, as shown in \cite{RaSi78}*{p.~294} for $d=1$, the function $\psi(x)=\sqrt{2}\abs{x}\eul^{-\abs{x}}$ satisfies $\norm{x^2\eul^{\ii\Delta t}\psi}^2=\scal{\eul^{\ii\Delta t}\psi,x^4\eul^{\ii\Delta t}\psi}=+\infty$ for any $t>0$. Note that here $D^2\psi$ does not exist.
\end{remark}
%-------------------%

%---------------------------------------------------------%
%:s.3
%---------------------------------------------------------%
\section{Periodic operators in Euclidean space} \label{sec:periodic}
%---------------------------------------------------------%
%:s.3.1
%---------------------------------------------------------%
\subsection{Continuous case}  \label{subsec:ak}

Consider $H=H_0+V=-\Delta+V$ on $L^2(\D{R}^d)$ with $V$ a periodic potential. We assume $V\in C^{m-1}(\D{R}^d)$ and its partial derivatives of order $<m$ are bounded. Denote as above 
\[
D\coloneqq -\ii\nabla_{x}=-\ii(\partial_{x_1},\dots,\partial_{x_d}).
\]
Recall the direct integral notations. If $(M,\dd m)$ is some measure space and $\F{h}$ a Hilbert space, its direct integral over $M$ is the Hilbert space $\CS{H}$ defined by
\begin{equation}\label{eq:dirinth}
\CS{H}=\int_M^\oplus\F{h}\,\dd m\coloneqq L^2(M,\dd m;\F{h})\,,\qquad\scal{f,g}_{\CS{H}}\coloneqq\int_M\scal{f_m,g_m}_{\F{h}}\,\dd m\,,
\end{equation}
where $f=(f_m)_{m\in M}$ and $g=(g_m)_{m\in M}$ with $f_m,g_m\in\F{h}$. Moreover, a measurable family $(A(m))_{m\in M}$ of operators in $\F{h}$ defines an operator $A$ in $\CS{H}$ by
\begin{equation}\label{eq:dirinta}
A=\int_M^\oplus A(m)\dd m,\qquad (Af)_m=A(m)f_m.
\end{equation}

For definiteness we assume the potential is $(2\pi\D{Z})^d$-periodic. Let $\D{T}^d\coloneqq\D{R}^d/(2\pi\D{Z})^d=[0,2\pi)^d$ and $\D{T}_\ast^d\coloneqq\D{R}^d/\D{Z}^d=[0,1)^d$. We define an operator 
\begin{align*}
U\colon L^2(\D{R}^d)
&\to\int_{\D{T}_\ast^d}^{\oplus}L^2(\D{T}^d)\dd\theta,\\
\psi&\mapsto\bigl((U\psi)_{\theta}\bigr)_{\theta\in\D{T}_\ast^d}
\end{align*}
by setting
\begin{equation}  \label{eq:upsi}
(U\psi)_{\theta}(x)\coloneqq\sum_{n\in\D{Z}^d}\eul^{-\ii\theta\cdot(x+2\pi n)}\psi(x+2\pi n)
\end{equation}
for each $\theta\in\D{T}_\ast^d$ and $x\in\D{T}^d\subset\D{R}^d$. So we indeed have $(U\psi)_{\theta}\in L^2(\D{T}^d)$. This operator $U$ is unitary and we have (see proof below)
\begin{equation}       \label{eq:(1)}
UHU^{-1}=\int_{\D{T}_\ast^d}^{\oplus}H(\theta)\dd\theta
\end{equation}
where $H(\theta)=(D+\theta)\cdot(D+\theta)+V$ on $L^2(\D{T}^d)$ with form domain $H^1(\D{T}^d)$. This means $H(\theta)=\sum_{j=1}^d(-\ii\partial_{x_j}+\theta_j)^2+V$. We prefer to denote $D\cdot D$ rather than $D^2$ to avoid confusing with the notation $D^m=(-\ii)^m(\partial_{x_1}^m,\dots,\partial_{x_d}^m)$ from the previous section.

Here $H(\theta)$ has compact resolvent and
\[
H(\theta)=\sum_{n=1}^{\infty}E_n(\theta)P_n(\theta),
\]
where $E_n(\theta)$ are its eigenvalues in non-decreasing order and $P_n(\theta)$ are the corresponding eigenprojections. Moreover, for every $n$, the following set has full Lebesgue measure in $\D{T}_\ast^d$:
\begin{equation}\label{eq:(4)}
S_n\coloneqq\accol{\theta\in\D{T}_\ast^d: P_n(\theta)\text{ and }E_n(\theta) \text{ are smooth at }\theta,\text{ and }\nabla_{\theta}E_n(\theta)\neq 0}.
\end{equation}
These facts are standard, see \cite{Wi78}*{Theorems 1 and 2} or \cite{Be82}*{Proposition 10.11}. For the last assertion, note that $E_n(\theta)$ is non-constant and analytic outset a nullset \cite{Wi78}, so all its partial derivatives are analytic, and the zero set of each $\partial_{\theta_j}E_n(\theta)$ has measure zero \cite{Ku16}*{Lemma 5.22}. 

%-------------------%
\begin{proof}[Proof of \eqref{eq:(1)}]
For future use we prove \eqref{eq:(1)}. We have for $\phi\in\CS{S}(\D{R}^d)$,
\begin{align} \label{eq:encoreuneeq}
\biggl(\biggl(\int_{\D{T}_\ast^d}^{\oplus}D\dd\theta\biggr)U\phi\biggr)_{\theta}(x)
&\coloneqq D(U\phi)_{\theta}(x)\notag\\
&=\sum_{n\in\D{Z}^d}\accol{D\eul^{-\ii\theta\cdot(x+2\pi n)}\phi(x+2\pi n)+\eul^{-\ii\theta\cdot(x+2\pi n)}D\phi(x+2\pi n)}\notag\\
&=-\theta(U\phi)_{\theta}(x)+(UD\phi)_{\theta}(x),
\end{align}
Relation \eqref{eq:encoreuneeq} means
\begin{equation}\label{eq:ud}
(UD\phi)_{\theta}=(D+\theta)(U\phi)_{\theta},
\end{equation}
that is, $-\ii(U\partial_{x_j}\phi)_{\theta}=(-\ii\partial_{x_j}+\theta_j)(U\phi)_\theta$ for $j=1,\dots,d$. So by iterating this relation we get $-(U\partial_{x_j}^2\phi)_\theta=(-\ii\partial_{x_j}+\theta_j)^2(U\phi)_{\theta}$. We thus find for $H=-\Delta+V$, using periodicity of $V$, that  
\begin{equation}\label{eq:uh}
[UH\phi]_{\theta}=[U(-\Delta\phi)+U(V\phi)]_{\theta}=\sum_{j=1}^d(-\ii\partial_{x_j}+\theta_j)^2(U\phi)_{\theta}+V(U\phi)_{\theta}=H(\theta)(U\phi)_{\theta},
\end{equation}
with $H(\theta)=\sum_{j=1}^d(-\ii\partial_{x_j}+\theta_j)^2+V$, which is \eqref{eq:(1)}.
\end{proof}
%-------------------%

In the following, for $\psi\in L^2(\D{R}^d)$, $x\in\D{R}^d$ and $f\in L^2(\D{T}_\ast^d)$, $\theta\in\D{T}_\ast^d$,
\[
x^m\psi(x)\coloneqq(x_1^m\psi(x),\dots,x_d^m\psi(x))\text{ and }(\nabla_{\theta}f)^m\coloneqq((\partial_{\theta_1}f)^m,\dots,(\partial_{\theta_d}f)^m).
\]

%-------------------%
%:thm 3.1
%-------------------%
\begin{theorem}  \label{thm:contper}
Assume $\psi\in H^{2m}(\D{R}^d)$, $\psi\neq 0$ satisfies $\norm{x^m\psi}<\infty$. Then
\[
\lim_{t\to+\infty}\frac{\norm{x^m\,\eul^{-\ii tH}\psi}^2}{t^{2m}}= \int_{\D{T}_\ast^d}\sum_{n=1}^\infty\abs{(\nabla_{\theta}E_n(\theta))^m}^2\norm{P_n(\theta)(U\psi)_{\theta}}^2_{L^2(\D{T}^d)}\dd\theta>0,
\]
where $\abs{(\nabla_{\theta}E_n(\theta))^m}^2=\sum_{j=1}^d (\partial_{\theta_j}E_n(\theta))^{2m}$.
\end{theorem}
%-------------------%

This theorem was previously obtained for $m=1$ by J.~Asch and A.~Knauf in \cite{AsKn98}. In contrast to \cite{AsKn98} we will not consider temporal means of moment derivatives in the proof, as their rigorous manipulation becomes overly complicated for higher $m$. 

The extension to potentials periodic with respect to a different lattice than $(2\pi\D{Z})^d$ is quite immediate once the lattice vocabulary has been settled, we prefer to postpone this to the next subsection where it is a less standard material.

%-------------------%
\begin{proof}
The proof is divided into three steps. Step 2 is in some sense the main argument, which is quite simple, while Steps 1 and 3 are technical justifications. In Step 1 we show it suffices to prove relation \eqref{eq:mainthingu} below. In Step 2 we prove this relation for $\psi$ in some spectral subspace, then we consider general $\psi$ in Step 3.

%-------------------%
%:step 1
%-------------------%
\begin{stepone*}
For a linear operator $O$, define the propagator sandwich
\begin{equation}\label{eq:heisen}
O(t)\coloneqq\eul^{\ii tH}O\eul^{-\ii tH}
\end{equation}
suggested by the Heisenberg picture. Below we take for $O$ the operator of multiplication by $x^m$. Since $\norm{x^m\psi}<\infty$, we have $x^m\eul^{-\ii tH}\psi\in L^2(\D{R}^d)$ for all $t\geq 0$ by Theorem~\ref{thm:momconts}, so we may apply $U$ to it. We claim Theorem~\ref{thm:contper} will follow if we prove the strong convergence
\begin{equation} \label{eq:mainthingu}
\lim_{t\to+\infty}\frac{Ux^m(t)\psi}{t^m}=\Bigl(\int_{\D{T}_\ast^d}^\oplus\sum_{n=1}^\infty(\nabla_{\theta}E_n(\theta))^mP_n(\theta)\dd\theta\Bigr)U\psi,
\end{equation}
for since $\eul^{\ii tH}$ and $U$ are unitary, we then get by definitions \eqref{eq:dirinth} and \eqref{eq:dirinta} that
\begin{align*}
\frac{\norm{x^m\eul^{-\ii tH}\psi}^2}{t^{2m}}=\frac{\norm{Ux^m(t)\psi}^2}{t^{2m}}
&\longrightarrow\biggl\lVert\Bigl(\int_{\D{T}_\ast^d}^\oplus\sum_{n=1}^\infty(\nabla_{\theta}E_n(\theta))^mP_n(\theta)\dd\theta\Bigr)U\psi\biggr\rVert^2\\
&\qquad=\int_{\D{T}_\ast^d}\Bigl\lVert\sum_{n=1}^\infty (\nabla_{\theta}E_n(\theta))^mP_n(\theta)(U\psi)_{\theta}\Bigl\lVert^2_{L^2(\D{T}^d)}\dd\theta\\
&\qquad=\int_{\D{T}_\ast^d}\sum_{n=1}^\infty\abs{(\nabla_{\theta}E_n(\theta))^m}^2\norm{P_n(\theta)(U\psi)_{\theta}}^2_{L^2(\D{T}^d)}\dd\theta
\end{align*}
since the $P_n(\theta)$ are orthogonal projections. 

This is indeed $>0$: if $S\subset\D{T}_\ast^d$ is the intersection over $n$ of all sets $S_n$ considered in \eqref{eq:(4)}, each of them being of full measure, $S$ also has full measure, and restricting the integral over $S$, we see that it is zero iff
\[
\sum_{n=1}^\infty|(\nabla_{\theta}E_n(\theta))^m|^2\norm{P_n(\theta)(U\psi)_{\theta}}_{L^2(\D{T}^d)}^2=0\quad\text{for a.e.\ }\theta\in\D{T}_\ast^d,
\]
which occurs on $S$ iff $\norm{P_n(\theta)(U\psi)_{\theta}}_{L^2(\D{T}^d)}^2=0$ for all $n$, a.e. $\theta\in\D{T}_\ast^d$, and so $U\psi=0$ and $\psi=0$, in contradiction with our assumption on $\psi$.
\end{stepone*}

%-------------------%
%:step 2
%-------------------%
\begin{steptwo*}
Now by \eqref{eq:upsi},
\[
(Ux^m\phi)_{\theta}(x)=\sum_{n\in\D{Z}^d}\eul^{-\ii\theta\cdot(x+2\pi n)}(x+2\pi n)^m\phi(x+2\pi n)=\ii^m\nabla_{\theta}^m(U\phi)_{\theta}(x),
\]
where $\nabla_{\theta}^mf=(\partial_{\theta_1}^mf,\dots,\partial_{\theta_d}^m)f$. Recalling \eqref{eq:uh}, we thus have
\[
(Ux^m(t)\psi)_{\theta}=\ii^m\eul^{\ii tH(\theta)}\nabla_{\theta}^m\eul^{-\ii tH(\theta)}(U\psi)_{\theta},
\]
\emph{Assume} for now that $\psi$ satisfies $(U\psi)_{\theta}=\sum_{n=1}^NP_n(\theta)(U\psi)_{\theta}$ for some $N$ independent of $\theta$, i.e.,
\begin{equation} \label{eq:assump}
P_n(\theta)(U\psi)_{\theta}=0\text{ for all }n>N.
\end{equation}
Then, denoting $\nabla_\theta^pf\nabla_\theta^qg=(\partial_{\theta_1}^pf\partial_{\theta_1}^qg,\dots,\partial_{\theta_d}^pf\partial_{\theta_d}^qg)$, we have
\begin{align*}
\nabla_{\theta}^m \eul^{-\ii tH(\theta)}(U\psi)_{\theta} 
&= \sum_{n=1}^N\nabla_{\theta}^m[\eul^{-\ii tE_n(\theta)}P_n(\theta)(U\psi)_{\theta}] \\
&= \sum_{n=1}^N\sum_{p=0}^m\binom{m}{p} (\nabla_{\theta}^p\eul^{-\ii tE_n(\theta)})\nabla_{\theta}^{m-p}[P_n(\theta)(U\psi)_{\theta}] \\
&= \sum_{n=1}^N \eul^{-\ii tE_n(\theta)}(-\ii t \nabla_{\theta}E_n(\theta))^mP_n(\theta)(U\psi)_{\theta}+\ord_{n,\theta}(t^{m-1})
\end{align*}
for some error term $\ord_{n,\theta}(t^{m-1})$ which can be made explicit. By dominated convergence, 
\begin{equation}\label{eq:detail1}
\lim_{t\to+\infty}\int_{\D{T}_\ast^d}\Bigl\lVert\frac{1}{t^m}\ii^m\eul^{\ii tH(\theta)}\sum_{n=1}^N\ord_{n,\theta}(t^{m-1})\Bigr\rVert^2\dd\theta=0.
\end{equation}
Indeed, the error term consists of derivatives of $E_n(\theta)$, $P_n(\theta)$ and $(U\psi)_{\theta}$ of order $\leq m$, which are all bounded. Recall that $\nabla_{\theta}^k(U\psi)_{\theta}(x)=(-\ii)^k(Ux^k\psi)_{\theta}$ and the $E_n(\theta)$ and $P_n(\theta)$ are analytic in each coordinate $\theta_j$ by \cite{Kato}*{Theorem VII.3.9}. See also \cite{Be82}*{Lemmas 10.12 and 10.13}.

On the other hand,
\begin{align*}
&\ii^m\eul^{\ii tH(\theta)}\sum_{n=1}^N (-\ii t\nabla_{\theta}E_n(\theta))^m\eul^{-\ii tE_n(\theta)}P_n(\theta)(U\psi)_{\theta}\\
&\qquad\qquad=t^m\sum_{k=1}^\infty \eul^{\ii tE_k(\theta)}P_k(\theta)\sum_{n=1}^N(\nabla_{\theta}E_n(\theta))^m\eul^{-\ii tE_n(\theta)}P_n(\theta)(U\psi)_{\theta}\\
&\qquad\qquad =t^m\sum_{k=1}^\infty(\nabla_{\theta}E_k(\theta))^mP_k(\theta)(U\psi)_{\theta}
\end{align*}
since $P_k(\theta)P_n(\theta)=\delta_{k,n}P_k(\theta)$ and $P_k(\theta)(U\psi)_{\theta}=0$ for $k>N$. We thus showed that 
\begin{equation}\label{eq:detail2}
\lim_{t\to+\infty}\int_{\D{T}_\ast^d}\Bigl\lVert\frac{(Ux^m(t)\psi)_{\theta}}{t^m}-\sum_{k=1}^\infty(\nabla_{\theta}E_k(\theta))^mP_k(\theta)(U\psi)_{\theta}\Bigr\rVert^2\dd\theta=0\,,
\end{equation}
completing the proof in the case where $\psi$ satisfies \eqref{eq:assump} for some $N$.
\end{steptwo*}
%-------------------%
%:step 3
%-------------------%
\begin{stepthree*}
It remains to justify why we can reduce the general case to the particular case where $\psi$ satisfies \eqref{eq:assump}, i.e., $P_k(\theta)(U\psi)_{\theta}=0$ for all $k>N$. Let 
\[
\psi_N\coloneqq U^{-1}\biggl(\int_{\D{T}_\ast^d}^\oplus\sum_{n=1}^N P_n(\theta)\dd\theta\biggr)U\psi\quad\text{and}\quad A\psi\coloneqq  U^{-1}\biggl(\int_{\D{T}_\ast^d}^\oplus\sum_{n=1}^\infty(\nabla_{\theta} E_n(\theta))^m P_n(\theta)\dd\theta\biggr)U\psi.
\]
Then
\[
\Bigl\lVert\frac{x^m(t)\psi}{t^m}-A\psi\Bigr\rVert\leq \Bigl\lVert\frac{x^m(t)\psi_N}{t^m}-A\psi_N\Bigr\rVert+ \Bigl\lVert\frac{x^m(t)\psi}{t^m}-\frac{x^m(t)\psi_N}{t^m}\Bigr\rVert+\norm{A(\psi_N-\psi)}.
\]
We showed in Step 2 that, $N$ being fixed, the first norm vanishes as $t\to+\infty$. If we show that the two other norms vanish as $t\to+\infty$ followed by $N\to\infty$, then the proof will be complete.

By Theorem~\ref{thm:momconts} we know that
\[
\frac{\norm{x^m(t)(\psi_N-\psi)}}{t^m}\leq c_m\frac{\norm{x^m(\psi_N-\psi)}}{t^m}+C_m\sum_{k=0}^m\norm{H^k(\psi_N-\psi)}.
\]
$N$ being fixed, the first term in the right-hand side vanishes as $t\to+\infty$. For $\psi$, by assumption, we indeed have $x^m\psi\in L^2(\D{R}^d)$, and for $\psi_N$:
\begin{align*}
\norm{x^m\psi_N}^2 
&=\int_{\D{T}_\ast^d}\Bigl\lVert\nabla_{\theta}^m\sum_{n=1}^NP_n(\theta)(U\psi)_{\theta}\Bigr\rVert^2\dd\theta\\
&=\int_{\D{T}_\ast^d}\Bigl\lVert\sum_{n=1}^N\sum_{p=0}^m\binom{m}{p}(\nabla_{\theta}^{m-p}P_n(\theta))(Ux^p\psi)_{\theta}\Bigr\rVert^2\dd\theta,
\end{align*}
which is finite for any $N$. Next, by definition,
\[
(U\psi_N)_{\theta}-(U\psi)_{\theta}=\sum_{n=1}^NP_n(\theta)(U\psi)_{\theta}-(U\psi)_{\theta}=-\sum_{n=N+1}^\infty P_n(\theta)(U\psi)_{\theta}.
\]
Using \eqref{eq:uh}, i.e., $(UH^k\phi)_\theta=H^k(\theta)(U\phi)_\theta$, we thus have
\[
\norm{H^k(\psi_N-\psi)}^2=\norm{UH^k(\psi_N-\psi)}^2=\int_{\D{T}_\ast^d}\Bigl\lVert H(\theta)^k\sum_{n=N+1}^\infty P_n(\theta)(U\psi)_{\theta}\Bigr\rVert_{L^2(\D{T}^d)}^2\,\dd\theta. 
\]
But, for any $k\leq m$, 
\[
\sum_{n=1}^\infty\norm{P_n(\theta)H(\theta)^k(U\psi)_{\theta}}^2_{L^2(\D{T}^d)}=\norm{H(\theta)^k(U\psi)_{\theta}}_{L^2(\D{T}^d)}^2\leq c''\norm{(U\psi)_{\theta}}_{H^{2r}(\D{T}^d)}^2<\infty.
\]
Hence, $\norm{H^k(\psi_N-\psi)}^2$ vanishes as $N\to\infty$ by dominated convergence. This completes the proof that the second norm $\bigl\lVert\frac{x^m(t)(\psi-\psi_N)}{t^m}\bigr\rVert$ vanishes as $t\to+\infty$ followed by $N\to\infty$.

Finally 
\begin{align*}
\norm{A\phi}^2
&= \int_{\D{T}_\ast^d}\sum_{n=1}^\infty\big\lVert(\nabla_{\theta} E_n(\theta))^m P_n(\theta)(U\phi)_{\theta}\bigr\rVert^2_{L^2(\D{T}^d)}\dd\theta\\
&= \int_{\D{T}_\ast^d}\sum_{n=1}^\infty\bigl\lVert(\nabla_{\theta}E_n(\theta))^mP_n(\theta)(H(\theta)+\ii)^{-m/2}P_n(\theta)(U(H+\ii)^{m/2}\phi)_{\theta}\bigr\rVert^2_{L^2(\D{T}^d)}\dd\theta\,,
\end{align*}
where we used \eqref{eq:uh} and the fact that $H(\theta)$ commutes with $P_n(\theta)$. This can now be bounded by
\begin{equation}\label{eq:backnow}
\int_{\D{T}_\ast^d}\sum_{n=1}^\infty\norm{(\nabla_{\theta}E_n(\theta))^mP_n(\theta)(H(\theta)+\ii)^{-m/2}}_{L^2\to L^2}^2\norm{P_n(\theta)(U(H+\ii)^{m/2}\phi)_{\theta}}^2_{L^2(\D{T}^d)}\dd\theta\,.
\end{equation}
Now
\begin{align*}
(\nabla_{\theta}E_n(\theta))^mP_n(\theta)(H(\theta)+\ii)^{-m/2} 
&= (\nabla_{\theta}E_n(\theta))^mP_n(\theta)\sum_{k=1}^{\infty}(E_k(\theta)+\ii)^{-m/2}P_k(\theta)\\
&= \frac{(\nabla_{\theta}E_n(\theta))^m}{(E_n(\theta)+\ii)^{m/2}}P_n(\theta).
\end{align*}
So for $g\in L^2(\D{T}^d)$ we have
\begin{equation}\label{eq:magiceid}
\norm{(\nabla_{\theta}E_n(\theta))^mP_n(\theta)(H(\theta)+\ii)^{-m/2}g}=\Bigl\lvert\frac{(\nabla_{\theta}E_n(\theta))}{(E_n(\theta)+\ii)^{1/2}}\Bigr\rvert^m\norm{P_n(\theta)g}.
\end{equation}
It thus suffices to estimate this norm for $m=1$. For this, we observe that
\begin{equation}\label{eq:nabd}
(\nabla_{\theta}E_n(\theta))P_n(\theta)=2P_n(\theta)(D+\theta)P_n(\theta).
\end{equation}
This is based on $E_n(\theta)P_n(\theta)= H(\theta)P_n(\theta)$ and $\nabla_{\theta}H(\theta)=\nabla_{\theta}((D+\theta)^2+V)=2(D+\theta)$. More precisely, on the one hand, we have
\begin{align*}
\nabla_{\theta}(P_nH(\theta)P_n) 
&= (\nabla_{\theta}P_n)H(\theta)P_n + P_n(\nabla_{\theta}H(\theta))P_n + P_nH(\theta)\nabla_{\theta}P_n \\
&= P_n(\nabla_{\theta}H(\theta))P_n + E_n\nabla_{\theta}(P_nP_n)=P_n(\nabla_{\theta}H(\theta))P_n + E_n\nabla_{\theta}P_n.
\end{align*}
On the other hand,
\[
\nabla_{\theta}(P_nH(\theta)P_n)=\nabla_{\theta}(E_nP_n)=(\nabla_{\theta}E_n)P_n + E_n\nabla_{\theta}P_n\,.
\]
Thus, $(\nabla_{\theta}E_n)P_n=P_n(\nabla_{\theta}H(\theta))P_n=2P_n(D+\theta)P_n$ as stated\footnote{The fact that $\nabla_{\theta}H(\theta)g=2(D+\theta)g$ for any $g\in L^2(\D{T}^d)$ is clear by definition of the derivative. Computing $\nabla_{\theta}\eul^{-\ii tH(\theta)}$ however is less clear. This is why we used the spectral decomposition of $H(\theta)$ in Step 2 to estimate it.} in \eqref{eq:nabd}.

We next note that, since $V\geq -b$ for $b=\norm{V}_\infty$, then
\[
\norm{(D+\theta)f}^2=\scal{f,(D+\theta)^2f}=\scal{f,H(\theta)f}-\scal{f, Vf}\leq\norm{(H(\theta)+b)^{1/2}f}^2.
\]
We thus have, using \eqref{eq:nabd},
\begin{align*}
\norm{(\nabla_{\theta}E_n(\theta))P_n(\theta)(H(\theta)+\ii)^{-1/2}g}
&= 2\norm{P_n(\theta)(D+\theta)P_n(\theta)(H(\theta)+\ii)^{-1/2}g} \\
&\leq 2\norm{(H(\theta)+b)^{1/2}(H(\theta)+\ii)^{-1/2}P_n(\theta)g}\\  &\leq c\norm{P_n(\theta)g}_{L^2}
\end{align*}
with $c$ independent of $\theta$. Recalling \eqref{eq:magiceid}, this says that $\big|\frac{\nabla_{\theta}E_n(\theta)}{(E_n(\theta)+\ii)^{1/2}}\big|\leq c$. By \eqref{eq:magiceid}, this implies $\norm{(\nabla_{\theta}E_n(\theta))^mP_n(\theta)(H(\theta)+\ii)^{-m/2}}_{L^2\to L^2}\leq c^m$.

Back to \eqref{eq:backnow}, we thus showed that
\begin{align*}
\norm{A\phi}^2
&\leq c^{2m}\int_{\D{T}_\ast^d}\sum_{n=1}^\infty\norm{P_n(\theta)(U(H+\ii)^{m/2}\phi)_{\theta}}^2_{L^2(\D{T}^d)}\dd\theta\\
&=c^{2m}\int_{\D{T}_\ast^d}\norm{(U(H+\ii)^{m/2}\phi)_{\theta}}^2_{L^2(\D{T}^d)}\dd\theta\\
&= c^{2m}\norm{(H+\ii)^{m/2}\phi}^2_{L^2(\D{R}^d)}.
\end{align*}

Thus, $\norm{A(\psi_N-\psi)}\leq c^m\scal{\psi-\psi_N,(H+\ii)^m(\psi-\psi_N)}^{1/2}$. Since we previously showed that $H^k(\psi_N-\psi)\to 0$ for any $k\leq m$, we get that the third norm $\norm{A(\psi_N-\psi)}$ vanishes as $N\to\infty$. This completes the proof.\qedhere
\end{stepthree*}
\end{proof}
%-------------------%

%-------------------%
%:exa 3.2
%-------------------%
\begin{example}[the laplacian]\label{exa:Vnull}
For $V=0$, i.e., $H=-\Delta$ we have $E_{k}(\theta)=\sum_{j=1}^d(k_j+\theta_j)^2$, corresponding to the eigenvectors $e_{k}(x)=\frac{1}{(2\pi)^{d/2}}\eul^{\ii k\cdot x}$, $k\in\D{Z}^d$, and $P_{k}(\theta)f=\scal{e_{k},f} e_{k}$. So the limit reduces to
\begin{align*}
\int_{\D{T}_\ast^d}\sum_{k\in\D{Z}^d}\abs{2^m(k+\theta)^m}^2\abs{\scal{e_{k},(U\psi)_{\theta}}}^2\dd\theta
&= 4^m\int_{\D{T}_\ast^d}\sum_{k\in\D{Z}^d}|\scal{ e_{k},(D+\theta)^m(U\psi)_{\theta}}|^2\dd\theta\\
&= 4^m\int_{\D{T}_\ast^d}\norm{(UD^m\psi)_{\theta}}^2_{L^2(\D{T}^d)}\dd\theta\\
&= 4^m\norm{D^m\psi}^2\,.
\end{align*}
This agrees with Theorem~\ref{thm:cts} where this special case was already considered.
\end{example}
%-------------------%

%---------------------------------------------------------%
%:s.3.2
%---------------------------------------------------------%
\subsection{Periodic discrete graphs}\label{subsec:perdiscrete}
We now consider periodic Schr\"odinger operators on discrete graphs $\Gamma\subset\D{R}^d$. 
%---------------------------------------------------------%
%:s.3.2.1
%---------------------------------------------------------%
\subsubsection{Preliminaries}
The spectral theory has been considered in great generality in \cite{KorSa14} (and earlier in \cite{Kru11} in the special case of $\D{Z}^d$). The infinite graph $\Gamma= (V,E)$ is assumed to be locally finite and invariant under translations by some fixed basis $\accol{\F{a}_1,\dots,\F{a}_d}$ in $\D{R}^d$. This includes the hexagonal lattice, the triangular lattice, the body-centered and face-centered cubic lattices, and much more. The authors in \cite{KorSa14} take the convention of writing the vertices in this coordinate system. This shortens the notation compared to say \cite{ReSi78}*{pp.~303--309} but it can also be quite confusing and nonstandard. As a compromise we introduce notations as follows.
\begin{itemize}
\item 
Given $x=(x_1,\dots,x_d)\in\D{R}^d$, we let $x_{\F{a}}=x_1\F{a}_1+ \dots+x_d\F{a}_d$.
\item 
Let $\D{Z}_{\F{a}}^d \coloneqq\accol{n_1\F{a}_1+\dots+n_d\F{a}_d:n_j\in\D{Z}}=\accol{n_{\F{a}}:n\in \D{Z}^d}$.
\item 
Let $\F{b}_1,\dots,\F{b}_d$ be the \emph{dual basis}, that is, $\F{a}_i\cdot\F{b}_j=2\pi\delta_{i,j}$. We similarly denote $x_{\F{b}}=x_1\F{b}_1 + \dots + x_d\F{b}_d$.
\item 
Let $\C{C}_{\F{a}}=\accol{x_1\F{a}_1+\dots+x_d\F{a}_d:x_j\in [0,1)}$ and $\C{C}_{\F{b}}=\{x_1\F{b}_1+\dots+x_d\F{b}_d:x_j\in [0,1)\}$ be the unit cells in the respective bases.
\item 
Let $\D{T}_\ast^d\coloneqq[0,1)^d$. Consider $\rho\colon\D{T}_\ast^d \to \C{C}_{\F{b}}$, $\rho(\theta)=\theta_{\F{b}}=\sum_{i=1}^d \theta_i \F{b}_i$. We endow $\C{C}_{\F{b}}$ with the image measure $\dd\rho_\star$ of the Lebesgue measure. In particular, $\int_{\C{C}_{\F{b}}} f\,\dd\rho_\star=\int_{\D{T}_\ast^d} f(\theta_{\F{b}})\,\dd\theta$. 
\end{itemize}

For example, if $\F{a}_i=\F{e}_i$ is the standard basis, then $\F{b}_i =2\pi\F{e}_i$, $x_{\F{a}}=x$, $\theta_{\F{b}}=2\pi\theta$, $\C{C}_{\F{a}}=[0,1)^d$, $\C{C}_{\F{b}}=[0,2\pi)^d$ and $\dd\rho_\star=\frac{\dd\theta}{(2\pi)^d}$. In general, if $\F{b}_i=\sum_{j=1}^d\beta_j(i)\F{e}_j$, then $\rho(\theta)=B \theta$, where $B$ is the matrix $(\beta_i(j))_{i,j=1}^d$ and we get $\dd\rho_\star=\frac{\dd\theta}{|\det B|}$, see e.g. \cite{Te14}*{Appendix A.7}.

We let $V_f=V\cap\C{C}_{\F{a}}$ be the unit crystal, $V_f=\accol{v_1,\dots,v_\nu}$. Since $\Gamma + \F{a}_k=\Gamma$, we have
\begin{equation}\label{eq:V=Z+Vf}
V=\D{Z}_{\F{a}}^d + V_f.
\end{equation}

At the level of edges, such additive invariance implies that $\C{N}_{n_{\F{a}}+v_n}=n_{\F{a}} + \C{N}_{v_n}$, i.e., $\accol{u: u\sim n_{\F{a}} + v_n}=\accol{n_{\F{a}}+v: v\sim v_n}$.

Recall \eqref{eq:V=Z+Vf}. If $v\in V$, we shall denote
\[
v=\lfloor v\rfloor_{\F{a}} + \accol{v}_{\F{a}}\,,
\]
where $\lfloor v\rfloor_{\F{a}}\in\D{Z}_{\F{a}}^d$ and $\accol{v}_{\F{a}}\in V_f$ are the integer and fractional parts of $v$, respectively.
%---------------------------------------------------------%
%:s.3.2.2
%---------------------------------------------------------%
\subsubsection{The Schr\"odinger operator}
We now consider the Schr\"odinger operator 
\[
H=\C{A}+Q
\]
on $\Gamma$, where $Q$ is a periodic potential satisfying $Q(v+\F{a}_i)=Q(v)$. In particular, $Q$ has at most $\nu$ distinct values; those on $V_f$. One could also consider $\Delta + Q$, where $\Delta=d - \C{A}$ and $d$ is the degree matrix; this is in fact the framework in \cite{KorSa14} and $d$ can just be regarded as a periodic potential.

We define $U\colon\ell^2(\Gamma)\to \int_{\C{C}_{\F{b}}}^\oplus \ell^2(V_f)\,\dd\rho_\star$ by
\begin{equation}\label{eq:utile}
(U\psi)_{\theta_{\F{b}}}(v_n)=\sum_{k_{\F{a}}\in \D{Z}_{\F{a}}^d} \eul^{-\ii\theta_{\F{b}} \cdot (v_n+k_{\F{a}})} \psi(v_n+k_{\F{a}})\,.
\end{equation}
Then $U$ is unitary:
\begin{align*}
\norm{U\psi}^2 
&=\int_{\C{C}_{\F{b}}} \norm{(U\psi)_{\theta_{\F{b}}}}^2_{\ell^2(V_f)}\,\dd\rho_\star\\
&=\sum_{v_n\in V_f}\sum_{k_{\F{a}},n_{\F{a}}\in\D{Z}_{\F{a}}^d} \psi(v_n+k_{\F{a}})\overline{\psi(v_n+n_{\F{a}})}\int_{\D{T}_\ast^d} \eul^{\ii\theta_{\F{b}}\cdot (n_{\F{a}}-k_{\F{a}})}\,\dd\theta\\
&=\sum_{v_n\in V_f}\sum_{k_{\F{a}}\in\D{Z}_{\F{a}}^d}|\psi(k_{\F{a}}+v_n)|^2=\norm{\psi}^2,
\end{align*}
where we used that $\theta_{\F{b}}\cdot x_{\F{a}} =\sum_{i,j} \theta_i \F{b}_i\cdot x_j \F{a}_j=2\pi\theta\cdot x$. This shows isometry. Given $g\in \int_{\C{C}_{\F{b}}}^\oplus \ell^2(V_f)\,\dd\rho_\star$, the function $\psi(v_n+k_{\F{a}})=\int_{\D{T}_\ast^d} \eul^{\ii\phi_{\F{b}}\cdot v_n}g_{\phi_{\F{b}}}(v_n)\eul^{\ii k_{\F{a}}\cdot\phi_{\F{b}}}\,\dd\phi$ satisfies $(U\psi)_{\theta_{\F{b}}}(v_n)=g_{\F{\theta_b}}(v_n)$. This can be seen by expanding $\theta\mapsto \eul^{\ii\theta_{\F{b}}\cdot v_n}g_{\theta_{\F{b}}}(v_n)$ in the orthonormal basis $\{\eul^{-\ii\theta_{\F{b}}\cdot k_{\F{a}}}\}_{k_{\F{a}}\in\D{Z}_{\F{a}}^d}$ of $L^2(\D{T}_\ast^d)$. Thus, $U$ is unitary as asserted.

We claim that $UHU^{-1}=\int_{\C{C}_{\F{b}}}^\oplus H(\theta_{\F{b}})\dd \rho_\star$, where $H(\theta_{\F{b}})$ acts on $\ell^2(V_f)$ by\footnote{Our operators $U$ and $H(\theta_{\F{b}})$ differ slightly from those of \cite{KorSa14}. Namely, they consider $(\widetilde{U}\psi)_{\theta_{\F{b}}}(v_n)=\eul^{\ii\theta_{\F{b}}\cdot v_n}(U\psi)_{\theta_{\F{b}}}(v_n)$, so they obtain instead the fiber operator $\widetilde{H}(\theta_{\F{b}})=\eul^{\ii\theta_{\F{b}}\cdot}H(\theta_{\F{b}})\eul^{-\ii\theta_{\F{b}}\cdot}$, where $(\eul^{\pm\ii\theta_{\F{b}}\cdot}f)(v_k)=\eul^{\pm\ii\theta_{\F{b}}\cdot v_k}f(v_k)$. Since $H(\theta_{\F{b}})$ and $\widetilde{H}(\theta_{\F{b}})$ are unitarily equivalent, they share the same eigenvalues $E_n(\theta_{\F{b}})$, moreover $\widetilde{P}_n(\theta_{\F{b}})=\eul^{\ii\theta_{\F{b}}\cdot}P_n(\theta_{\F{b}})\eul^{-\ii\theta_{\F{b}}\cdot}$. We have avoided the introduction of ``bridges'' and quotient graphs and used fractional parts instead, which we think is more transparent for our purposes.}
\begin{equation}\label{eq:hfib}
[H(\theta_{\F{b}})f](v_n)=\sum_{u\sim v_n}\eul^{\ii\theta_{\F{b}}\cdot (u-v_n)} f(\{u\}_{\F{a}})+Q(v_n)f(v_n),
\end{equation}
where the sum runs over the neighbors $u$ of $v_n$ \emph{in the full graph} $\Gamma$ (not just in $V_f$).

In fact, given $\psi\in \ell^2(\Gamma)$,
\begin{align*}
(U\C{A} \psi)_{\theta_{\F{b}}}(v_n) &= \sum_{k_{\F{a}}\in\D{Z}_{\F{a}}^d} \eul^{-\ii\theta_{\F{b}}\cdot(v_n+k_{\F{a}})} \sum_{u\sim k_{\F{a}}+v_n}\psi(u)=\sum_{k_{\F{a}}\in \D{Z}_{\F{a}}^d}\eul^{-\ii\theta_{\F{b}}\cdot(v_n+k_{\F{a}})}\sum_{u\sim v_n} \psi(k_{\F{a}}+u)\\
&= \sum_{u\sim v_n} \eul^{\ii\theta_{\F{b}}\cdot(u-v_n)} \sum_{k_{\F{a}}\in\D{Z}_{\F{a}}^d}\eul^{-\ii\theta_{\F{b}}\cdot (\{u\}_{\F{a}}+\lfloor u\rfloor_{\F{a}}+k_{\F{a}})}\psi(\{u\}_{\F{a}} + \lfloor u\rfloor_{\F{a}} + k_{\F{a}}) \\
&= \sum_{u\sim v_n} \eul^{\ii\theta_{\F{b}} \cdot(u-v_n)}  (U\psi)_{\theta_{\F{b}}}(\{u\}_{\F{a}}).
\end{align*}

Since $Q$ is $\D{Z}_{\F{a}}^d$-periodic, if follows that $(UH\psi)_{\theta_{\F{b}}}(v_n) =H(\theta_{\F{b}})(U\psi)_{\theta_{\F{b}}}(v_n)$ as asserted.

For the purposes of ballistic transport, we regard $H(\theta_{\F{b}})$ as a complex $\nu\times\nu$ matrix (indeed $\ell^2(V_f)\equiv\D{C}^\nu$) and observe that $\theta_j\mapsto H(\theta_{\F{b}})$ is analytic for any fixed $(\theta_k)_{k\neq j}$. We may write $H(\theta_{\F{b}})=\sum_{n=1}^{\nu} E_n(\theta_{\F{b}})P_n(\theta_{\F{b}})$, where $E_n(\theta_{\F{b}})$ are the eigenvalues (in increasing order, counting multiplicities) and $P_n(\theta_{\F{b}})$ the eigenprojections.

The authors in \cite{KorSa14} then use Floquet theory to deduce that the spectrum of $H$ consists of at most $\nu$ bands (intervals), $\spec(H) =\spec_{\ac}(H)\cup\spec_{\fb}(H)$, where $\spec_{\ac}(H)$ are the bands of AC spectrum, while $\spec_{\fb}(H)$ is the set of flat bands, i.e., intervals degenerate to a single point, which can occur in general. For example, for the Laplacian there are no flat bands in the case of the triangular lattice, while there is a flat band in the case of the face-centered cubic lattice \cite{KorSa14}*{Proposition 8.3}.

%-------------------%
%:thm 3.3
%-------------------%
\begin{theorem}\label{thm:dicreper}
Let $H=\C{A}+Q$, with $Q$ periodic on the periodic discrete infinite graph $\Gamma$. Assume $\psi\in\ell^2(\Gamma)$ satisfies $\norm{x^m\psi}<\infty$. Then
\begin{equation}\label{eq:dicreper}
\lim_{t\to+\infty}\frac{\norm{x^m\,\eul^{-\ii tH}\psi}^2}{t^{2m}}=\int_{\D{T}_\ast^d}\sum_{n=1}^\nu\Bigl\lVert\Bigl(\frac{\nabla_{\theta_{\F{a}}}E_n(\theta_{\F{b}})}{2\pi}\Bigr)^mP_n(\theta_{\F{b}})(U\psi)_{\theta_{\F{b}}}\Bigr\rVert_{\D{C}^\nu}^2\dd\theta,
\end{equation}
where $\nabla_{\theta_{\F{a}}}\coloneqq\sum_{i=1}^d\F{a}_i\partial_{\theta_i}$. The limit is nonzero if $\psi$ has a nontrivial AC component. The limit is zero if $\psi$ is an eigenvector.
\end{theorem}
%-------------------%

The quantity $2\pi$ did not appear in Theorem~\ref{thm:contper} because we considered $\nabla_{\theta}=\sum_{i=1}^d \F{e}_i\partial_{\theta_i}$, while in that result $\F{a}_i=2\pi \F{e}_i$, so $\nabla_{\theta_{\F{a}}}=2\pi \nabla_{\theta}$. There, $\theta_{\F{b}}=\theta$.

It can indeed happen that $\psi$ is an eigenvector (corresponding to a flat band, as in face-centered cubic lattices). On the other hand, $H$ always has at least one band of AC spectrum, see \cite{KorSa18}*{Theorem 2.1}.

%-------------------%
\begin{proof}
We have $x^m(t)\psi\in\ell^2(\Gamma)$ by Theorem~\ref{thm:formel}. As in Theorem~\ref{thm:contper}, it suffices to show that
\begin{equation}\label{eq:dislim}
\lim_{t\to+\infty}\frac{Ux^m(t)\psi}{t^m}=\Bigl(\int_{\D{T}_\ast^d}^\oplus\sum_{n=1}^\nu\Bigl(\frac{\nabla_{\theta_{\F{a}}}E_n(\theta_{\F{b}})}{2\pi}\Bigr)^mP_n(\theta_{\F{b}})\dd\theta\Bigr)U\psi.
\end{equation}
Recall \eqref{eq:utile}. We note that if $v_n=\sum_{j=1}^d s_j \F{a}_j$, $s_j\in [0,1)$, then $\theta_{\F{b}}\cdot (k_{\F{a}}+v_n)=\sum_{i,j=1}^d \theta_i\F{b}_i \cdot(k_j+s_j)\F{a}_j=2\pi \theta \cdot (k+s)$. Hence, $\partial_{\theta_i} \eul^{-\ii\theta_{\F{b}}\cdot(k_{\F{a}}+v_n)}=-2\pi\ii(k_i+s_i)\eul^{-\ii\theta_{\F{b}}\cdot(k_{\F{a}}+v_n)}$ and $\nabla_{\theta_{\F{a}}} \eul^{-\ii\theta_{\F{b}}\cdot(k_{\F{a}}+v_n)}=-2\pi\ii(k_{\F{a}}+v_n)\eul^{-\ii\theta_{\F{b}}\cdot(k_{\F{a}}+v_n)}$. We thus get
\begin{equation}\label{eq:disfouper}
(Ux^m\phi)_{\theta_{\F{b}}}(v_n)=\sum_{k_{\F{a}}\in\D{Z}^d_{\F{a}}}\eul^{-\ii\theta_{\F{b}} \cdot (k_{\F{a}}+v_n)}(k_{\F{a}}+v_n)^m\phi(k_{\F{a}}+v_n)=\frac{\ii^m}{(2\pi)^m} \nabla_{\theta_{\F{a}}}^m(U\phi)_{\theta_{\F{b}}}(v_n)\,.
\end{equation}
Hence,
\[
(Ux^m(t)\psi)_{\theta_{\F{b}}}=\frac{\ii^m}{(2\pi)^m}\eul^{\ii tH(\theta_{\F{b}})}\nabla_{\theta_{\F{a}}}^m\eul^{-\ii tH(\theta_{\F{b}})}(U\psi)_{\theta_{\F{b}}}\,.
\]
On the other hand,
\begin{align}\label{eq:disna}
\nabla_{\theta_{\F{a}}}^m \eul^{-\ii tH(\theta_{\F{b}})}(U\psi)_{\theta_{\F{b}}} &= \sum_{n=1}^\nu\nabla_{\theta_{\F{a}}}^m[\eul^{-\ii tE_n(\theta_{\F{b}})}P_n(\theta_{\F{b}})(U\psi)_{\theta}] \notag\\
&= \sum_{n=1}^\nu \eul^{-\ii tE_n(\theta_{\F{b}})}(-\ii t \nabla_{\theta_{\F{a}}}E_n(\theta_{\F{b}}))^m P_n(\theta_{\F{b}})(U\psi)_{\theta}+\ord_{n,\theta}(t^{m-1})
\end{align}
for some error term $\ord_{n,\theta}(t^{m-1})$ which can be made explicit. 

From here we conclude the proof of \eqref{eq:dicreper} as in \eqref{eq:detail1}-\eqref{eq:detail2}. Here it is even easier because the technical Step 3 is no longer needed, as the sums are already finite.

It remains to prove the last two claims. 

Suppose first that $\psi$ has a nontrivial AC component. Since $\spec(H) = \cup_{n=1}^\nu\spec_n$, where $\spec_n=E_n(\C{C}_{\F{b}})= [E_n^-,E_n^+]$ are the spectral bands, then we have $\B{1}_B\psi\neq 0$ for some $B\subset\spec_n$ of positive measure, where $\spec_n$ is non-degenerate. This implies that $E_n$ is a non-constant piecewise analytic function \cite{KorSa14}, so $\nabla_{\theta_\F{a}}E_n(\theta_\F{b})\neq 0$ a.e. If $B$ lies in the intersection of several bands $\spec_k$, we ignore the degenerate ones as they are just finitely many points in $B$, as for the non-degenerate ones, they also have $\nabla_{\theta_\F{a}} E_k(\theta_\mathfrak{b})\neq 0$ a.e.

If the RHS of \eqref{eq:dicreper} is zero, then $\int_{\D{T}_\ast^d} \sum_{E_n(\theta_{\F{b}})\in B}\Big\|\Big(\frac{\nabla_{\theta_{\F{a}}}E_n(\theta_{\F{b}})}{2\pi}\Big)^mP_n(\theta_{\F{b}})(U\psi)_{\theta_{\F{b}}}\Big\|^2_{\D{C}^\nu}\,\dd\theta=0$, so $\sum_{E_n(\theta_{\F{b}})\in B}\Big\|\Big(\frac{\nabla_{\theta_{\F{a}}}E_n(\theta_{\F{b}})}{2\pi}\Big)^mP_n(\theta_{\F{b}})(U\psi)_{\theta_{\F{b}}}\Big\|^2_{\D{C}^\nu}=0$ a.e., so $\sum_{E_n(\theta_{\F{b}})\in B}\norm{P_n(\theta_{\F{b}})(U\psi)_{\theta_{\F{b}}}}^2_{\D{C}^\nu}=0$ a.e. by the previous paragraph, so $\norm{\B{1}_B(H(\theta_{\F{b}}))(U\psi)_{\theta_{\F{b}}}}^2=0$ a.e. But $\B{1}_B(H)\psi\neq 0$ by hypothesis. If $H'=\int_{\C{C}_{\F{b}}}^\oplus H(\theta_{\F{b}})\,\dd\rho_\star$, we have
\begin{align*}
\norm{\B{1}_B(H)\psi}^2 
&= \norm{U\B{1}_B(H)\psi}^2=\norm{\B{1}_B(H')U\psi}^2\\
&=\int_{\C{C}_{\F{b}}}\norm{(\B{1}_B(H')U\psi)(\theta_{\F{b}})}^2\,\dd\rho_\star\\
&=\int_{\C{C}_{\F{b}}}\norm{\B{1}_B(H(\theta_{\F{b}}))(U\psi)_{\theta_{\F{b}}}}^2\,\dd\rho_\star\,.
\end{align*}
We see that $\norm{\B{1}_B(H(\theta_{\F{b}}))(U\psi)_{\theta_{\F{b}}}}^2 =0$ a.e.\ would imply $\norm{\B{1}_B(H)\psi}^2=0$, a contradiction. Thus, the RHS of \eqref{eq:dicreper} is nonzero. This checks the claim for $\psi$ with a nontrivial AC component.

Now let $\psi$ be an eigenvector. Clearly then $\frac{\norm{x^m\eul^{-\ii tH}\psi}^2}{t^{2m}} =\frac{\norm{x^m\eul^{-\ii t\lambda}\psi}^2}{t^{2m}} = \frac{\norm{x^m\psi}^2}{t^{2m}}\to 0$ as $t\to+\infty$. In that case the RHS of \eqref{eq:dicreper} is also zero because such an eigenvector corresponds to some flat band $\spec_{\fb}(H) = E_r(\C{C}_{\F{b}})$, i.e. $E_r(\theta_{\F{b}})=\lambda$ is constant in $\theta$, so $\nabla_{\theta_{\F{a}}}E_r(\theta_{\F{b}})$ is identically zero. On the other hand $(U\psi)_{\theta_{\F{b}}} = (U\B{1}_{\{\lambda\}}(H)\psi)_{\theta_{\F{b}}} = [\B{1}_{\{\lambda\}}(H')(U\psi)]_{\theta_{\F{b}}} = \B{1}_{\{E_r(\theta_{\F{b}})\}}(H(\theta_{\F{b}}))(U\psi)_{\theta_{\F{b}}}=P_r(\theta_{\F{b}})(U\psi)_{\theta_{\F{b}}}$. Since $P_n(\theta_{\F{b}})P_r(\theta_{\F{b}})=\delta_{n,r}P_r(\theta_{\F{b}})$, the RHS of \eqref{eq:dicreper} reduces to 
\[
\int_{\D{T}_\ast^d}\biggl\lVert\biggl(\frac{\nabla_{\theta_{\F{a}}}E_r(\theta_{\F{b}})}{2\pi}\biggr)^{\!m}P_r(\theta_{\F{b}})(U\psi)_{\theta_{\F{b}}}\biggr\rVert^2_{\D{C}^\nu}\dd\theta=0.\qedhere
\]
\end{proof}
%-------------------%

%-------------------%
%:exa 3.4
%-------------------%
\begin{example}[the integer lattice] \label{exa:freezd1}
In the case of the adjacency matrix $\C{A}$ on $\D{Z}^d$, we have $V_f=\{0\}$, $H(\theta_{\F{b}})=\sum_{u\sim 0}\eul^{\ii\theta_{\F{b}}\cdot u}=2\sum_{j=1}^d \cos 2\pi\theta_j$, only one eigenvalue $E_1(\theta_{\F{b}})=2\sum_{j=1}^d\cos 2\pi\theta_j$, which traces the spectrum $[-2d,2d]$ as $\theta$ varies over $\D{T}_\ast^d$, and $P_1(\theta_{\F{b}})=\id$. So Theorem~\ref{thm:dicreper} asserts that
\begin{align*}
\lim_{t\to+\infty}\frac{\norm{n^m\eul^{-\ii tH}\psi}^2}{t^{2m}} 
&=\int_{\D{T}_\ast^d}\biggl\lvert\biggl(\frac{\nabla E_1(\theta_{\F{b}})}{2\pi}\biggr)^{\!m}\biggr\rvert^2\,\biggl\lvert\sum_{k\in\D{Z}^d}\eul^{-\ii k\cdot\theta_{\F{b}}}\psi(k)\biggr\rvert^2\dd\theta\\
&=4^m\int_{\D{T}^d}\biggl(\sum_{j=1}^d\sin^{2m}\theta_j\!\biggr)\abs{\hat\psi(\theta)}^2\dd\theta.
\end{align*}
This expression is the same as \eqref{eq:recall}, so we recover Theorem~\ref{thm:freediscrete}. 
\end{example}
%-------------------%

%-------------------%
%:exa 3.5
%-------------------%
\begin{example}[the triangular lattice] \label{exa:freezd2}
The triangular lattice has $V=\D{Z}^2$, but with two additional links:
\[
\C{A}\psi(n)=\psi(n-\F{e}_1) + \psi(n+\F{e_1}) + \psi(n-\F{e_2}) + \psi(n+\F{e}_2) + \psi(n-\F{e}_1-\F{e}_2) + \psi(n+\F{e}_1+\F{e}_2),
\]
where $(\F{e}_{1,2})$ is the standard basis,\footnote{Some authors replace the last two terms by $\psi(n-\F{e}_1+\F{e}_2) + \psi(n+\F{e}_1-\F{e}_2)$, this is just a different sheering convention and slightly changes the eigenvalue $E_1(\theta)$.}  see \cite{KorSa14}*{Fig.~3}. This is a $6$-regular graph. Here $V_f=\{0\}$, $H(\theta_{\F{b}})=\sum_{u\sim 0}\eul^{\ii\theta_{\F{b}}\cdot u}$, there is one eigenvalue $E_1(2\pi\theta)=2\cos2\pi\theta_1 + 2\cos 2\pi\theta_2 + 2\cos(2\pi(\theta_1+\theta_2))$ which traces the interval $[-3,6]=\spec(\C{A})=\spec_{\ac}(\C{A})$ as $\theta$ varies over $\D{T}^2_\ast$. Again $P_1(2\pi\theta)=\id$. So Theorem~\ref{thm:dicreper} asserts that
\[
\lim_{t\to+\infty}\frac{\norm{x^m\eul^{-\ii tH}\psi}^2}{t^{2m}}=4^m\int_{\D{T}^2}\bigl\lbrack(\sin\theta_1+\sin(\theta_1+\theta_2))^{2m}+(\sin\theta_2+\sin(\theta_1+\theta_2))^{2m}\bigr\rbrack\abs{\hat\psi(\theta)}^2\dd\theta.
\]
\end{example}
%-------------------%

%-------------------%
%:exa 3.6
%-------------------%
\begin{example}[the hexagonal lattice] \label{exa:freezd3}
The hexagonal (honeycomb/graphene) lattice is $3$-regular, see \cite{KorSa14}*{Fig.~7}. Here $\F{a}_1=a(1,0)$, $\F{a}_2=\frac{a}{2}(1,\sqrt{3})$, where $a$ is the distance between a vertex and its second nearest neighbor. We have $V_f=\accol{0,v}$, where $v=\frac{a}{2}(1,\frac{1}{\sqrt{3}})=\frac{\F{a}_1+\F{a}_2}{3}$, $V=\D{Z}_{\F{a}}^2 + V_f$, and
\begin{equation}\label{eq:ahex}
\begin{split}
\C{A}\psi(k_{\F{a}}) 
&=\psi(k_{\F{a}}+v)+\psi(k_{\F{a}}+v-\F{a}_1)+\psi(k_{\F{a}}+v-\F{a}_2),\\
\C{A}\psi(k_{\F{a}}+v) 
&=\psi(k_{\F{a}})+\psi(k_{\F{a}}+\F{a}_1)+\psi(k_{\F{a}}+\F{a}_2).
\end{split}
\end{equation}
By \eqref{eq:hfib}, we have $H(\theta_{\F{b}})f(0)=\sum_{u\sim 0} \eul^{\ii\theta_{\F{b}}\cdot u}f(\{u\}_{\F{a}})$. Using \eqref{eq:ahex}, this yields $H(\theta_{\F{b}})f(0)=\eul^{\ii\theta_{\F{b}}\cdot v}f(v) + \eul^{\ii\theta_{\F{b}}\cdot(v-\F{a}_1)}f(v) + \eul^{\ii\theta_{\F{b}}\cdot(v-\F{a}_2)}f(v)$. Similarly, we have 
\[
H(\theta_{\F{b}})f(v)=\sum_{u\sim v}\eul^{\ii\theta_{\F{b}}\cdot(u-v)}f(\{u\}_{\F{a}})=\eul^{-\ii\theta_{\F{b}}\cdot v}f(0) + \eul^{\ii\theta_{\F{b}}\cdot(\F{a}_1-v)}f(0)+\eul^{\ii\theta_{\F{b}}\cdot(\F{a}_2-v)}f(0).
\]
This shows that
\[
H(\theta_{\F{b}})=\begin{pmatrix}0&\eul^{\ii\theta_{\F{b}}\cdot v}\xi(\theta_{\F{b}})\\\eul^{-\ii\theta_{\F{b}}\cdot v}\overline{\xi(\theta_{\F{b}})}&0\end{pmatrix},\text{ where }\xi(\theta_{\F{b}})\coloneqq 1+\eul^{-\ii\theta_{\F{b}}\cdot\F{a}_1}+\eul^{-\ii\theta_{\F{b}}\cdot\F{a}_2}.
\]
Since $\theta_{\F{b}}\cdot k_{\F{a}}=2\pi \theta\cdot k$, this gives $\xi(\theta_{\F{b}})=1+\eul^{-2\pi \ii\theta_1} + \eul^{-2\pi\ii\theta_2}$. The eigenvalues are given by $\pm |\xi(\theta_{\F{b}})|=\pm \sqrt{3+2\cos 2\pi \theta_1+2\cos 2\pi \theta_2+2\cos 2\pi (\theta_1-\theta_2)}$. So 
\[
\frac{\nabla_{\theta_{\F{a}}} E_{\pm}(\theta_{\F{b}})}{2\pi}=\frac{\mp (\sin 2\pi\theta_1+\sin 2\pi(\theta_1-\theta_2))}{|\xi(\theta_{\F{b}})|} \F{a}_1 +\frac{\mp (\sin 2\pi\theta_2+\sin 2\pi(\theta_2-\theta_1))}{|\xi(\theta_{\F{b}})|} \F{a}_2.
\]
Since we have $|(\nabla_{\theta_{\F{a}}} E_+(\theta_{\F{b}}))^m|^2=|(\nabla_{\theta_{\F{a}}} E_-(\theta_{\F{b}}))^m|^2$, the definitions of $P_{\pm}(\theta_{\F{b}})$ do not matter as we use $\norm{P_1(\theta_{\F{b}})(U\psi)_{\theta_{\F{b}}}}^2+\norm{P_2(\theta_{\F{b}})(U\psi)_{\theta_{\F{b}}}}^2=\norm{(U\psi)_{\theta_{\F{b}}}}^2$. We conclude that
\begin{align*}
\lim_{t\to+\infty}\frac{\norm{x^m\eul^{-\ii tH}\psi}^2}{t^{2m}}
&=\int_{\D{T}^2}\frac{a^{2m}[(\sin\theta_1+\frac{\sin\theta_2+\sin(\theta_1-\theta_2)}{2})^{2m}+(\frac{\sqrt{3}}{2}(\sin\theta_2+\sin(\theta_2-\theta_1)))^{2m}]}{|1+\eul^{\ii\theta_1}+\eul^{\ii\theta_2}|^{2m}}\\
&\qquad\times\bigg[\Big|\sum_{k_{\F{a}}\in\D{Z}_{\F{a}}^2}\eul^{-\ii k_{\F{a}}\cdot\theta_{\F{b}}}\psi(k_{\F{a}})\Big|^2+\Big|\sum_{k_{\F{a}}\in \D{Z}_{\F{a}}^2}\eul^{-\ii(k_{\F{a}}+v)\cdot\theta_{\F{b}}}\psi\big(k_{\F{a}}+v)\big)\Big|^2\bigg]\,\frac{\dd\theta}{(2\pi)^2}\,.
\end{align*} 
\end{example}
%-------------------%

%-------------------%
\begin{remark}
While revising the paper, we discovered the recent work \cite{Fil} which considers ballistic transport for periodic Jacobi matrices. This corresponds to $\Gamma=\D{Z}^d$, $m=1$ and adjacency matrices carrying periodic weights (besides the periodic potential).
\end{remark}
%-------------------%

%---------------------------------------------------------%
%:s.4
%---------------------------------------------------------%
\section{Universal covers}  \label{sec:uc}
%---------------------------------------------------------%
%:s.4.1
%---------------------------------------------------------%
\subsection{Background}  \label{sec:ucbg}

Let $\C{T}$ be a tree, that is, a connected graph with no cycles. Consider a Schr\"odinger operator
\[
H=\C{A}+W
\]
on $\C{T}$, where $\C{A}$ is the adjacency matrix of the tree and $W$ a potential on $\C{T}$. We start by recalling some properties of the Green function of $H$.

Given $\gamma\in\D{C}^+\coloneqq\accol{z\in\D{C}:\Im z>0}$ and $(v,w)$ a directed edge in $\C{T}$, let 
\[
G^\gamma(v,w)\coloneqq\scal{\delta_v,(H-\gamma)^{-1}\delta_w}
\]
denote the Green function of $H$ and
\[
\zeta_v^\gamma(w)\coloneqq\frac{G^\gamma(v,w)}{G^\gamma(v,v)}.
\] 
The following identities are then classical (see, e.g., \cite{AnS19a}*{Section 2}):
\begin{alignat}{2}\label{eq:rec}
&\frac{1}{G^{\gamma}(v,v)}=W(v) + \sum_{u\sim v}\zeta_v^{\gamma}(u) - \gamma,&\qquad& 
\frac{-1}{\zeta_w^{\gamma}(v)}=W(v) + \sum_{u\in\C{N}_v\setminus \{w\}} \zeta_v^{\gamma}(u) - \gamma,\\
\label{eq:zetainv}
&\frac{1}{\zeta_w^{\gamma}(v)}-\zeta_v^{\gamma}(w)=\frac{-1}{G^{\gamma}(v,v)},&&\zeta_w^{\gamma}(v)=\frac{G^{\gamma}(v,v)}{G^{\gamma}(w,w)}\zeta_v^{\gamma}(w),\\
\label{eq:greenmulti}
&G^{\gamma}(v_0,v_k)=G^{\gamma}(v_0,v_0) \zeta_{v_0}^{\gamma}(v_1)\cdots\zeta_{v_{k-1}}^{\gamma}(v_k),&& G^{\gamma}(v,w)=G^{\gamma}(w,v)
\end{alignat}
for any non-backtracking path $(v_0;v_k)$ in a tree $\C{T}$.

The quantity $-\zeta_v^\gamma(w)$ can also be expressed as the Green function $G_{\C{T}^{(w|v)}}^\gamma(w,w)$ of a subtree of $\C{T}$. By the Herglotz property, we thus have $-\Im\zeta_v^\gamma(w)=|\Im\zeta_v^\gamma(w)|$ for any $(v,w)$. In particular, \eqref{eq:rec} implies the important relation
\begin{equation}\label{eq:redu}
\sum_{u\in\C{N}_v\setminus\{w\}} |\Im\zeta_v^\gamma(u)|=\frac{|\Im\zeta_w^\gamma(u)|}{|\zeta_w^\gamma(v)|^2}-\Im\gamma.
\end{equation}
%---------------------------------------------------------%
%:s.4.2
%---------------------------------------------------------%
\subsection{Universal covers}  \label{sec:ucper}

Suppose $(G,W)$ is any finite graph endowed with a potential $W$ and let $(\C{T},\C{W})=(\widetilde{G},\widetilde{W})$ be its universal cover, where the potential $W$ is lifted naturally by $\widetilde{W}(v)=W(\pi v)$, if $\pi\colon\C{T}\to G$ is the covering projection. For example, if $G$ is $d$-regular, then $\C{T}$ is the infinite $d$-regular tree. We assume that the minimal degree of $G$ is $\geq 2$.

The spectral theory of such trees has a rich history, see \cite{AnS19b} and references therein. It is known in particular that the spectrum consists of bands of absolutely continuous spectrum, possibly with infinitely degenerate eigenvalues between the bands. Moreover, in the interior of the bands, all Green functions limits $G^{\lambda}(v,w)\coloneqq\lim_{\eta\downarrow 0}G^{\lambda+\ii\eta}(v,w)$ and $\zeta^{\lambda}_v(w)\coloneqq\lim_{\eta\downarrow 0}\zeta_v^{\lambda+\ii\eta}(w)$ exist, $\Im G^{\lambda}(v,v)>0$ for any $v$, while $\Im\zeta_v^\lambda(w)$ is strictly negative, that is $\abs{\Im\zeta_v^\lambda(w)}=-\Im\zeta_v^\lambda(w)>0$ within the bands. 

We henceforth denote
\[
\gamma\coloneqq\lambda+\ii\eta\quad\text{with}\quad\eta>0.
\]
As in \cite{LeMS20}, we define the spectral quantities
\[
z_\gamma\coloneqq\inf_{(v,w)\in B(\C{T})}\abs{\Im\zeta_v^{\gamma}(w)}\quad\text{and}\quad
z_\lambda\coloneqq\inf_{(v,w)\in B(\C{T})}\abs{\Im\zeta_v^{\lambda}(w)},
\]
where the infimum is over the set $B(\C{T})$ of directed edges of $\C{T}$. Actually $\zeta_v^\gamma(w)=\zeta_{v'}^\gamma(w')$ if $\pi(v,w)=\pi(v',w')$, so $z_\lambda$ is just a minimum over the (lifts of) directed edges of $G$. In particular, $z_{\lambda}>0$. 

Back to the topic of ballistic transport, fix an origin $o\in\C{T}$, let $\abs{x}\coloneqq d(x,o)$ and $(x\psi)(x)\coloneqq\abs{x}\psi(x)$. In this section we shall consider averaged moments instead:
\[
\scal{x^\beta}_{\psi,\eta}\coloneqq 2\eta\int_0^\infty\eul^{-2\eta t}\norm{x^{\beta/2}\eul^{-\ii tH}\psi}^2\,\dd t\,.
\]
Let denote $G^z\coloneqq(H-z)^{-1}$. A well-known application of the Plancherel identity yields that
\begin{equation}\label{eq:planche}
\scal{x^\beta}_{\psi,\eta}=\frac{\eta}{\pi}\int_{-\infty}^{\infty}\norm{x^{\beta/2} G^{\lambda+\ii\eta}\psi}^2\,\dd t.
\end{equation}
Our aim is to find lower bounds for $\liminf\limits_{\eta\downarrow 0} \eta^\beta\scal{x^\beta}_{\psi,\eta}$. Using Theorem~\ref{thm:formel} and the Tauberian theorem \cite{Si79}*{Theorem 10.3}, such a lower bound implies 
\[
\liminf_{a\to+\infty}\frac{1}{a^{\beta+1}}\int_0^a\norm{x^{\beta/2}\eul^{-\ii tH}\psi}^2\,\dd t>0\quad\text{and}\quad\limsup_{t\to+\infty}\frac{\norm{x^{\beta/2}\eul^{-\ii tH}\psi}^2}{t^\beta}>0.
\]

%-------------------%
%:thm 4.1
%-------------------%
\begin{theorem}\label{thm:pertree}
We have 
\[
\liminf_{\eta\downarrow 0}\eta^{\beta}\scal{x^\beta}_{\delta_o,\eta}\geq\frac{1}{2^{\beta+1}\pi}\int_{\spec(H)}z_{\lambda}^{\beta}\Im G^{\lambda}(o,o)\dd\lambda>0.
\]
\end{theorem}
%-------------------%

%-------------------%
\begin{note*}
For $\beta=0$ it follows from \eqref{eq:planche} that
\[
\eta^0\scal{x^0}_{\psi,\eta}=\frac{\eta}{\pi}\int_{-\infty}^\infty\norm{G^{\lambda+\ii\eta}\psi}^2\,\dd t=\frac{1}{\pi}\int_{-\infty}^\infty\Im\langle\psi,G^{\lambda+\ii\eta}\psi\rangle\,\dd t.
\]
So the presence of $\Im G^{\lambda}(o,o)$ above is quite natural, while the moment effect is captured by $z_{\lambda}^\beta$. 
\end{note*}
%-------------------%

%-------------------%
\begin{proof}
We have
\begin{align}\label{eq:lb}
\eta^{\beta+1}\sum_v|v|^\beta|G^{\lambda+\ii\eta}(o,v)|^2&\geq\eta^{\beta+1}\sum_{|v|>b\eta^{-1}}|v|^\beta|G^{\lambda+\ii\eta}(o,v)|^2\notag\\
&\geq\eta^{\beta+1}b^\beta\eta^{-\beta}\sum_{|v|>b\eta^{-1}}|G^{\lambda+\ii\eta}(o,v)|^2 \notag\\
&= \eta b^{\beta}\Big[\norm{G^{\lambda+\ii\eta}\delta_o}^2-\sum_{|v|\leq b\eta^{-1}}|G^{\lambda+\ii\eta}(o,v)|^2\Big]\notag\\
&=b^\beta \Big[\Im G^{\lambda+\ii\eta}(o,o) - \eta\sum_{|v|\leq b\eta^{-1}}|G^{\lambda+\ii\eta}(o,v)|^2\Big]\,,
\end{align}
where we used the spectral theorem $G^z(o,o)=\int_{\spec(\C{A})}\frac{1}{x-z}\,\dd\mu_{o,o}(x)$ in the last equality.

Note that
\[
\eta\sum_{|v|\leq b\eta^{-1}}|G^{\lambda+\ii\eta}(o,v)|^2\leq\eta\norm{G^{\lambda+\ii\eta}\delta_o}^2=\Im G^{\lambda+\ii\eta}(o,o).
\]
This trivial estimate makes the lower bound \eqref{eq:lb} useless however, which is natural since it doesn't use any delocalization of $H$. The argument below can be summarized as showing that $b$ can be chosen so that
\[
\eta\sum_{|v|\leq b\eta^{-1}}|G^{\lambda+\ii\eta}(o,v)|^2 \leq\frac{1}{2}\Im G^{\lambda+\ii\eta}(o,o).
\]
Denoting $v_0\coloneqq o$, we have
\[
\sum_{|v|\leq M}|G^{\lambda+\ii\eta}(v_0,v)|^2=|G^{\lambda+\ii\eta}(v_0,v_0)|^2+\sum_{r=1}^M\sum_{v_1\sim v_0}\sum_{(v_2;v_r)}|G^{\lambda+\ii\eta}(v_0,v_{r})|^2
\]
where the last sum is over all nonbacktracking paths $(v_2;v_{M})$ of length $r-2$ ``outgoing'' from the directed edge $(v_0,v_1)$, i.e., $v_2\sim v_1$ and $v_2\neq v_0$. Now, for $\gamma=\lambda+\ii\eta$,
\begin{align}\label{eq:maxop}
\sum_{(v_2;v_{r})}|G^{\gamma}(v_0,v_{r})|^2 
&= |G^{\gamma}(o,o)|^2\sum_{(v_2;v_{r})}|\zeta_{v_0}^{\gamma}(v_1)\cdots\zeta_{v_{r-1}}^{\gamma}(v_r)|^2\notag\\
&\leq |G^{\gamma}(o,o)|^2\max_{(v,w)}\frac{|\zeta_v^\gamma(w)|^2}{|\Im\zeta_{v}^{\gamma}(w)|}\sum_{(v_2;v_{r})}|\zeta_{v_0}^{\gamma}(v_1)\cdots \zeta_{v_{r-2}}^{\gamma}(v_{r-1})|^2|\Im\zeta_{v_{r-1}}^{\gamma}(v_{r})|\notag\\
&\leq |G^{\gamma}(o,o)|^2|\Im\zeta_o^\gamma(v_1)|\max_{(v,w)}\frac{1}{|\Im\frac{1}{\zeta_{v}^{\gamma}(w)}|},
\end{align}
where we applied \eqref{eq:redu} $r-1$ times in the last step. 

By \eqref{eq:rec}, $\Im G^{\lambda+\ii\eta}(o,o)=(\sum_{u\sim o}|\Im\zeta_o^{\lambda+\ii\eta}(u)|+\eta)|G^{\lambda+\ii\eta}(o,o)|^2$. Thus,
\[
\sum_{(v_2;v_{r})}|G^{\lambda+\ii\eta}(v_0,v_{r})|^2 \leq\frac{\Im G^{\lambda+\ii\eta}(o,o)}{\sum_{u\sim o} |\Im\zeta_o^{\lambda+\ii\eta}(u)|+\eta} \frac{|\Im\zeta_o^{\lambda+\ii\eta}(v_1)|}{z_{\lambda+\ii\eta}}
\]
where we estimated $\Im \frac{1}{\zeta_v^\gamma(w)}$ using \eqref{eq:rec} and the definition of $z_{\lambda+\ii\eta}$. Thus,
\begin{equation}\label{eq:upshotcovs}
\sum_{|v|\leq M}|G^{\lambda+\ii\eta}(o,v)|^2\leq  |G^{\lambda+\ii\eta}(o,o)|^2 + M \frac{\Im G^{\lambda+\ii\eta}(o,o)}{z_{\lambda+\ii\eta}}\,.
\end{equation}
Take $M=b\eta^{-1}$. It follows from \eqref{eq:lb} that
\[
\eta^{\beta+1}\sum_v|v|^\beta|G^{\lambda+\ii\eta}(o,v)|^2\geq b^\beta\Big[\Im G^{\lambda+\ii\eta}(o,o)- b\frac{\Im G^{\lambda+\ii\eta}(o,o)}{z_{\lambda+\ii\eta}} - \eta |G^{\lambda+\ii\eta}(o,o)|^2\Big]\,.
\]
Choose $b=\frac{z_{\lambda+\ii\eta}}{2}$. Then
\begin{equation}\label{eq:aok}
\eta^{\beta+1}\sum_v|v|^\beta|G^{\lambda+\ii\eta}(o,v)|^2\geq\frac{z_{\lambda+\ii\eta}^{\beta}}{2^{\beta+1}}\Big[\Im G^{\lambda+\ii\eta}(o,o) - 2\eta |G^{\lambda+\ii\eta}(o,o)|^2\Big]\,.
\end{equation}
As $\eta\downarrow 0$, we get the lower bound $\frac{z_{\lambda}^{\beta}}{2^{\beta+1}}\Im G^{\lambda}(o,o)>0$.

This estimate is true for any $\lambda$ in the union of intervals $\cup_j\mathring{I}_j$ of AC spectrum, which is a set of positive measure. Using Fatou's lemma, we thus get
\[
\liminf_{\eta\downarrow 0}\eta^{\beta+1}\int_{-\infty}^{\infty}\sum_v|v|^\beta|G^{\lambda+\ii\eta}(o,v)|^2\,\dd\lambda \geq\int_{\cup_j\mathring{I}_j}\frac{z_{\lambda}^{\beta}}{2^{\beta+1}}\Im G^{\lambda}(o,o)\,\dd\lambda > 0\,.
\]
The claim follows by \eqref{eq:planche}.
\end{proof}
%-------------------%

%-------------------%
%:rem 4.2
%-------------------%
\begin{remark}
We may extend the result to functions $\psi$ of compact support, but the bound we get is not very good. Namely, if $\psi$ is supported in $\Lambda\subset \C{T}$, we obtain
\[
\liminf_{\eta\downarrow 0}\eta^{\beta}\scal{x^\beta}_{\psi,\eta}\geq \frac{1}{2^{2\beta+1}|\Lambda|^\beta}\int_{\spec(H)}\frac{z_{\lambda}^{\beta} \,(\Im\scal{\psi,G^{\lambda}\psi})^{\beta+1}}{\scal{\psi,\Im G^{\lambda}(\cdot,\cdot)\psi}^{\beta}}\,\dd\lambda,
\]
where $[\Im G^{\lambda}(\cdot,\cdot)\psi](w)=\Im G^{\lambda}(w,w)\psi(w)$. In particular, the bound becomes useless as the size of the support becomes infinite, which does not seem natural in view of our previous ballistic estimates. We thus omit the proof and think it a worthwhile question to obtain estimates for general $\psi$, without $\eta$-time averaging, and to find the exact limit.
\end{remark}
%-------------------%

%---------------------------------------------------------%
%:s.4.3
%---------------------------------------------------------%
\subsection{The Anderson model on universal covers}\label{sec:ucand}

We now study weak random perturbations of universal covering trees. The simplest example is the Anderson model on the $(q+1)$-regular tree, for which ballistic transport was previously established in the time-averaged sense by A.~Klein \cite{Kl96}, and by M.~Aizenman and S.~Warzel \cite{AiWa12}. Here we study the general case.

%-------------------%
\begin{assumptions*}
Compared to the previous subsection, we need to assume the tree $\C{T}=\widetilde{G}$ has a bit nicer geometry. If $b\in B(\C{T})$ we let $o(b)$ and $t(b)$ denote the origin and terminus of $b$, respectively.
\begin{itemize}
\item 
We now assume the minimal degree is $\geq 3$,
\item 
Since $G$ is finite, there are finitely many isomorphism classes of $(\C{T},b)$ as $b$ runs over the directed edges of $\C{T}$ and $(\C{T},b)$ is considered as a tree with ``root'' $b$. If $\C{N}_b\coloneqq\accol{b_+: o(b_+)=t(b)}$ is the set of edges outgoing from $b$, then we assume that for each $b$, there is at least one $b_+\in\C{N}_b$ such that $(\C{T},b)$ is isomorphic to $(\C{T},b_+)$.
\end{itemize}
The second condition is perhaps better visualized using the language of cone types, see \cite{KLWa12}*{condition (M1*)} in which it was introduced. Explicit examples of such trees can be found in \cite{AnS19b}. Concerning the random potential, we assume
\begin{itemize}
\item 
Each vertex $v\in\C{T}$ is endowed with a random variable $\C{W}(v)$, the $(\C{W}(v))_{v\in\C{T}}$ are i.i.d.\ with common distribution $\nu$ of compact support. We also impose some regularity as in \cite{AiWa12}*{assumption A2} for comfort, but perhaps this can be avoided.
\end{itemize}
\end{assumptions*}
%-------------------%

It is known \cite{KLWa12} that under these conditions the random Schr\"odinger operator $H=\C{A}+\epsilon\C{W}$ inherits the purely absolutely continuous spectrum of $\C{A}$ almost surely if $\epsilon$ is small enough -- one of the few results of Anderson delocalization. The result of \cite{KLWa12} can also be used to derive some inverse moments bounds on the imaginary parts of the Green functions, see \cite{AnISW21}*{Theorem 5.2} for details (in a different model). More precisely, it can be shown that within the stable intervals $I$ of pure AC spectrum, we also have
\begin{equation}\label{eq:(Green)}
\sup_{\lambda\in I}\sup_{\eta\in (0,1)}\sup_{b\in B(\C{T})}\expect(|\Im\zeta_{o(b)}^{\lambda+\ii\eta}(t_b)|^{-s}) < \infty \,,
\end{equation}
for $0\leq s\leq 5$, where $\sup_b$ runs over the set $B(\C{T})$ of directed edges of $\C{T}$ and it is in fact a maximum as it can be equivalently taken over the directed edges of $G$. 

%-------------------%
%:thm 4.3
%-------------------%
\begin{theorem}
Under the previous assumptions,
\[
\liminf_{\eta\downarrow 0}\eta^{\beta}\expect\bigl(\scal{x^\beta}_{\B{1}_I(H)\delta_o,\eta})>0.
\]
\end{theorem}
%-------------------%

The proof also works for states $f(H)\delta_o$ instead of $\B{1}_I(H)\delta_o$, if $f$ is piecewise continuous and supported in $I$. This result appeared before in \cites{AiWa12,Kl96} in the special case $\C{T}=\D{T}_q$, the $(q+1)$-regular tree.

%-------------------%
\begin{proof}
The main argument is the same as in Theorem~\ref{thm:pertree}, but there are two technical difficulties. First, we are dealing with the state $\B{1}_I(H)\delta_o$ instead of $\delta_o$. Second, we cannot estimate the expectation of a maximum as in \eqref{eq:maxop}.

Let us begin the proof. Let $f(H)=\B{1}_I(H)$. By \eqref{eq:planche}, we have 
\[
\eta^\beta \expect(\scal{x^\beta}_{f(H)\delta_o,\eta})=\frac{\eta^{\beta+1}}{\pi}\int_{-\infty}^\infty\sum_{v\in\C{T}} |v|^\beta \expect(\abs{(f(H) G^{\lambda+\ii\eta}\delta_o)(v)}^2)\,\dd\lambda.
\]
Essentially the same proof as \cite{AiWa12}*{Lemma 2.1} shows that $f(H)$ can be replaced by $f(\lambda)$. More precisely, we get
\begin{equation}\label{eq:nop}
\eta^\beta\expect(\scal{x^\beta}_{\B{1}_I(H)\delta_o,\eta})= \frac{\eta^{\beta+1}}{\pi}\int_I\sum_{v\in\C{T}} |v|^\beta\expect(\abs{G^{\lambda+\ii\eta}(o,v)}^2)\,\dd\lambda+\osmall(\eta).
\end{equation}
This settles the first difficulty mentioned above.

%-------------------%
\begin{note*}
We mention in passing that for the proof of \cite{AiWa12}*{Lemma 2.1}, one doesn't really need a Wegner bound. Inequality \cite{AiWa12}*{(2.7)} holds deterministically since $\frac{\eta}{\pi}\int_{-\infty}^\infty |f(\lambda)|^2\norm{G^{\lambda+\ii\eta}\varphi}^2\,\dd\lambda \leq\frac{\norm{f}_\infty^2}{\pi}\int_{-\infty}^\infty\int_{\spec(H)}\frac{\eta}{(x-\lambda)^2+\eta^2}\,\dd\mu_\varphi(x)\,\dd\lambda =\frac{\norm{f}_\infty^2}{\pi}\int_{\spec(H)}(\frac{\pi}{2}-\frac{-\pi}{2})\,\dd\mu_\varphi(x)=\norm{f}_\infty^2\norm{\varphi}^2$. Similarly, one can avoid the Wegner bound in \cite{AiWa12}*{(2.10)} by writing $Q(\eta)=\expect\int_{\spec(H)} F_\eta(x)\,\dd\mu_\varphi(x)$, where $F_\eta(x)=\frac{1}{\pi}\int_{-\infty}^\infty |f(x)-f(\lambda)|^2\frac{\eta}{(\lambda-x)^2+\eta^2}\,\dd\lambda$, then use dominated convergence. In fact $F_\eta(x)\to 0$ by a well-known calculation, cf.~\cite{Ka04}*{p.~6}, for continuous $f$. The extension to piecewise continuous $f$ is straighforward as long as $\expect(\mu_{\varphi}(\{a\}))=0$ for any $a$, which is weaker than a Wegner bound.
\end{note*}
%-------------------%

Using Fatou's lemma and \eqref{eq:lb}, we see as before that the proof is now reduced to showing that $\sum_{(v_2;v_r)}\expect(\abs{G^\gamma(v_0,v_r)}^2)$ stays bounded independently of $r$ and $\eta\in (0,1)$. Let $\gamma\coloneqq\lambda+\ii\eta$. We have
\[
\expect(\abs{G^\gamma(v_0;v_r)}^2)=\expect(|\zeta_{v_1}^\gamma(v_0)\cdots\zeta^\gamma_{v_r}(v_{r-1})|^2|G^\gamma(v_r,v_r)|^2)
\]
but
\[
|G^\gamma(v_r,v_r)|\leq\frac{1}{\sum_{u\sim v_r}|\Im\zeta^\gamma_{v_r}(u)|+\eta}\leq\frac{1}{|\Im\zeta^\gamma_{v_r}(v_{r+1})|},
\]
and $\zeta^\gamma_{v_r}(v_{r+1})$ is independent of $\zeta^\gamma_{v_1}(v_0)\cdots\zeta^\gamma_{v_r}(v_{r-1})$. Indeed, $\zeta_{v_r}^\gamma(v_{r+1})=- G_{\C{T}^{(v_{r+1}|v_{r})}}(v_{r+1},v_{r+1})$ is the Green function of the connected component $\C{T}^{(v_{r+1}|v_r)}$ containing $v_{r+1}$ obtained by removing the directed edge $(v_r,v_{r+1})$ from $\C{T}$, while $\zeta^\gamma_{v_1}(v_0)\cdots\zeta^\gamma_{v_r}(v_{r-1})=-G_{\C{T}^{(v_{r-1}|v_{r})}}(v_0,v_{r-1})$ lives on a disjoint subtree (see \cite{AnS19a}*{Eq.~(2.7)}). Thus,
\[
\expect(|G^\gamma(v_0;v_r)|^2)\leq\expect(|\zeta^\gamma_{v_1}(v_0)\cdots\zeta^\gamma_{v_r}(v_{r-1})|^2)\expect(|\Im\zeta^\gamma_{v_r}(v_{r+1})|^{-1}).
\]
Again by independence, we may insert $\frac{\expect(|\Im\zeta^\gamma_{v_r}(v_{r+1})|)}{\expect(|\Im\zeta^\gamma_{v_r}(v_{r+1})|)}$ and get
\begin{equation}\label{eq:1stred}
\expect(|G^\gamma(v_0;v_r)|^2)\leq\expect(|\zeta^\gamma_{v_1}(v_0)\cdots\zeta^\gamma_{v_r}(v_{r-1})|^2|\Im\zeta^\gamma_{v_r}(v_{r+1})|)\frac{\expect(|\Im\zeta^\gamma_{v_r}(v_{r+1})|^{-1})}{\expect(|\Im\zeta^\gamma_{v_r}(v_{r+1})|)}\,.
\end{equation}
The second fraction is $\frac{\expect(|\Im\zeta^\gamma_{v_0}(v_1)|^{-1})}{\expect(|\Im\zeta^\gamma_{v_0}(v_1)|)}\leq\expect(|\Im\zeta^\gamma_{v_0}(v_1)|^{-1})^2\leq c$ by \eqref{eq:(Green)}. For the first term, now that we gained the $|\Im\zeta^\gamma_{v_r}(v_{r+1})|$, we would like to put back the $|G^\gamma(v_r,v_r)|^2$ we removed. More precisely, we would like to find some $C$ independent of $r$ such that
\begin{align}\label{eq:awinsp}
&\expect(|\zeta^\gamma_{v_1}(v_0)\cdots\zeta^\gamma_{v_r}(v_{r-1})|^2|\Im\zeta^\gamma_{v_r}(v_{r+1})|)\notag\\
&\qquad\leq C\expect(|\zeta^\gamma_{v_1}(v_0)\cdots\zeta^\gamma_{v_r}(v_{r-1})|^2|\Im\zeta^\gamma_{v_r}(v_{r+1})||G^\gamma(v_r,v_r)|^2)\,.
\end{align}

Define $\expect_2(f)\coloneqq\frac{\expect(|\zeta^\gamma_{v_1}(v_0)\cdots\zeta^\gamma_{v_r}(v_{r-1})|^2|\Im\zeta^\gamma_{v_r}(v_{r+1})|f)}{\expect(|\zeta^\gamma_{v_1}(v_0)\cdots\zeta^\gamma_{v_r}(v_{r-1})|^2|\Im\zeta^\gamma_{v_r}(v_{r+1})|)}$. We should thus establish a lower bound $\expect_2(|G^\gamma(v_r,v_r)|^2)\geq C_-$ for some $C_-$ independent of $r$ and $\gamma$. This is nontrivial, but can be done exactly like \cite{AiWa12}*{pp.~8-10}, using \eqref{eq:(Green)}.

Since $|\zeta^\gamma_{v_1}(v_0)\cdots\zeta^\gamma_{v_r}(v_{r-1})G^\gamma(v_r,v_r)|^2=|G^\gamma(v_0,v_r)|^2$, we have finally shown that
\[
\expect(|G^\gamma(v_0,v_r)|^2)\leq cC\expect(|G^\gamma(v_0,v_r)|^2|\Im\zeta_{v_r}^\gamma(v_{r+1})|)\,.
\]
We may now apply \eqref{eq:redu} multiple times to get 
\begin{align*}
\sum_{(v_2;v_r)}\expect(|G^\gamma(v_0,v_r)|^2|\Im\zeta_{v_r}^\gamma(v_{r+1})|) 
&\leq\expect(|G^\gamma(v_0,v_0)|^2|\Im\zeta_{v_0}^\gamma(v_1)|)\\
&=\expect\biggl(\frac{\Im G^\gamma(o,o)|\Im\zeta_o^\gamma(v_1)|}{\sum_{u\sim o}|\Im\zeta_o^\gamma(u)|+\eta}\biggr)\leq\expect(\Im G^\gamma(o,o)).
\end{align*}
Taking $b=\frac{cC}{2\R{D}}$, where $\R{D}$ is the maximal degree, we deduce as in \eqref{eq:aok} that
\[
\eta^{\beta+1}\sum_v|v|^\beta\expect(|G^{\lambda+\ii\eta}(o,v)|^2)\geq\frac{1}{2}\Bigl(\frac{cC}{2\R{D}}\Bigr)^{\beta}\expect[\Im G^{\lambda+\ii\eta}(o,o)-2\eta|G^{\lambda+\ii\eta}(o,o)|^2]. 
\]
By \eqref{eq:nop}, we thus have
\begin{align*}
\liminf_{\eta\downarrow 0}\eta^\beta\expect(\scal{x^\beta}_{\B{1}_I(H)\delta_o,\eta})
&\geq\frac{1}{2}\Bigl(\frac{cC}{2\R{D}}\Bigr)^{\beta}\liminf_{\eta\downarrow 0}\int_I\expect[\Im G^{\lambda+\ii\eta}(o,o)]\,\dd\lambda \\
&= \frac{\pi}{2}\Bigl(\frac{cC}{2\R{D}}\Bigr)^{\beta} \expect[\mu_{\delta_o}(I)]=\frac{\pi}{2}\Bigl(\frac{cC}{2\R{D}}\Bigr)^{\beta}\expect(\norm{\B{1}_I(H)\delta_o}^2)>0.
\end{align*}
The last quantity is strictly positive by \eqref{eq:(Green)} and classical identities. Note indeed that
\[
\expect(\Im G^{\gamma}(o,o)) \geq\expect(|G^{\gamma}(o,o)|^2|\Im\zeta_o^\gamma(u)|)\text{ for any }u\sim o.
\] 
By Cauchy--Schwarz this is $\geq\frac{\expect(|\Im\zeta_o^\gamma(u)|^{1/2})}{\expect(|G^\gamma(o,o)|^{-2})}$. But, by \eqref{eq:rec}-\eqref{eq:zetainv},
\[
|G^\gamma(o,o)|^{-2} \leq 2|\zeta_o^\gamma(u)|^2 + \frac{2}{|\zeta_u^\gamma(o)|^2} \leq 2(|\Im\zeta_u^\gamma(u')|^{-2} + |\Im\zeta_u^\gamma(o)|^{-2}).
\]
Also, $\expect(|\Im\zeta_o^\gamma(u)|^{1/2})\geq\frac{1}{\expect(|\Im\zeta_o^\gamma(u)|^{-1/2})}$. Thus, using \eqref{eq:(Green)},
\[
\expect(\Im G^\gamma(o,o))\geq\frac{1}{4}\min_{u,u'\sim o}\expect(|\Im\zeta_o^\gamma(u)|^{-1/2})^{-1}\expect(|\Im\zeta_u^\gamma(u')|^{-2})^{-1}>0.
\]
The proof is completed.
\end{proof}
%-------------------%

%---------------------------------------------------------%
%:s.5
%---------------------------------------------------------%
\section{Epilogue: limiting distributions} \label{sec:epi}

As the paper was being finalized, we came upon some articles on the topic of \emph{quantum walks} and found interesting connections. The results we derive in this section are inspired by \cites{GJS04,SaSa20,Sa21}. Quantum walks are discrete in nature, traditionally consisting of simple shift operators on $\ell^2(\D{Z})$ with additional degrees of freedom. Schr\"odinger evolutions $\eul^{-\ii tH}$ are a lot more complex than shifts, however the quantum walks have been recently generalized in \cites{SaSa20,Sa21} and the theories now seem to share a good common ground.

Suppose in either continuous or discrete model that $\norm{\psi}=1$. Then $\norm{\eul^{-\ii tH}\psi}=1$ as the operator is unitary, so the entries $\abs{\eul^{-\ii tH}\psi(x)}^2$ define a probability density on $\D{R}^d$ (resp. $G$), and $\norm{x^m\eul^{-\ii tH}\psi}^2$ is just the $2m$-moment of this measure. Our results thus imply that for certain models, all the (normalized, even) moments of this probability measure converge. This naturally evokes the method of moments, see \cite{Bill95}*{Theorem 30.2}. 

We shall here give a fuller probabilistic/quantum mechanical picture of the transport theory by discussing the limiting behavior of this measure. If $H$ is a Schr\"odinger operator on a general graph and $X_t$ is a random vector with distribution $\abs{\eul^{-\ii tH}\psi(v)}^2$, we first show that $\frac{X_t}{t}$ is always asymptotically confined in a compact set. Afterwards we compute exactly the limiting distribution of $\frac{X_t}{t}$ for the periodic models of Section~\ref{sec:periodic}.
%---------------------------------------------------------%
%:s.5.1
%---------------------------------------------------------%
\subsection{General facts}
Our first task is to properly modify the probability measure to incorporate the division by $t^m$. This is easy. On $\D{R}^d$, if
\begin{equation} \label{eq:mutpsi}
\dd\mu_t^\psi(x)\coloneqq\abs{\eul^{-\ii tH}\psi(x)}^2\,\dd x,
\end{equation}
and if $Q^{(t)}\colon\D{R}^d\to\D{R}^d$ is defined by $Q^{(t)}(x)\coloneqq x/t$ then
\[
\frac{1}{t^m}\expect_{\mu_t^\psi}(x^m)=\int_{\D{R}^d}\Bigl(\frac{x}{t}\Bigr)^m\dd\mu_t^\psi(x)=\int_{\D{R}^d}\bigl\lvert Q^{(t)}(x)\bigr\rvert^m\,\dd\mu_t^\psi(x)=\int_{\D{R}^d}y^m\,\dd Q_\star^{(t)}\mu_t^\psi(y),
\]
where $Q_\star^{(t)}\mu_t^\psi$ is the image measure, see e.g. \cite{Te14}*{Appendix A.7}. Thus, $\frac{1}{t^m}\expect_{\mu_t^\psi}(x^m)=\expect_{Q_\star^{(t)}\mu_t^\psi}(x^m)$ and we have a statement about the convergence of moments of the measure
\begin{equation} \label{eq:nutpsi}
\nu_t^\psi\coloneqq Q_\star^{(t)}\mu_t^\psi\,.
\end{equation}

This measure has a very natural meaning. Suppose $X_t$ is a random vector in $\D{R}^d$ with distribution $\abs{\eul^{-\ii tH}\psi(x)}^2\,\dd x$. Then for $B\subset\D{R}^d$, we have $\nu_t^\psi(B)= \mu_t^\psi(tB)=\prob(X_t\in tB)=\prob(\frac{X_t}{t}\in B)$. In other words, $\nu_t^\psi$ is the distribution of $\frac{X_t}{t}$. Note that $\nu_t^\psi$ is also a probability measure since $\nu_t^\psi(\D{R}^d)=\mu_t^\psi(\D{R}^d)=1$.

Similarly, on a graph $G\subset\D{R}^d$, if $\psi$ is normalized, we define $\mu_t^\psi(v)\coloneqq\abs{\eul^{-\ii tH}\psi(v)}^2$, then $Q^{(t)}\colon G\to\D{R}^d$ by $Q^{(t)}(v)\coloneqq v/t$, and $\nu_t^\psi\coloneqq Q_\star^{(t)}\mu_t^\psi$. Note that $\nu_t^\psi$ is a measure on $\D{R}^d$ while $\mu_t^\psi$ is a measure on $G$.

The map $\psi\mapsto\nu_t^\psi$ enjoys a nice form of uniform continuity:

%-------------------%
%:lem 5.1
%-------------------%
\begin{lemma}\label{lem:nutcont}
For any measurable $B$ we have $\abs{\nu_t^\psi(B)-\nu_t^\varphi(B)} \leq(\norm{\psi}+\norm{\varphi})\norm{\psi-\varphi}$.
\end{lemma}
%-------------------%

%-------------------%
\begin{proof}
We write the proof on $\D{R}^d$, the same works on graphs. We have
\begin{align*}
\abs{\nu_t^\psi(B)-\nu_t^{\varphi}(B)} 
&=|\mu_t^{\psi}(tB)-\mu_t^{\varphi}(tB)|
=\Bigl\lvert\int_{tB}\left(\abs{\eul^{-\ii tH}\psi(x)}^2-\abs{\eul^{-\ii tH}\varphi(x)}^2\right)\dd x\Bigr\rvert\\
&\leq\int_{\D{R}^d}\left(\abs{\eul^{-\ii tH}\psi(x)}+\abs{\eul^{-\ii tH}\varphi(x)}\right)\abs{\eul^{-\ii tH}\psi(x)-\eul^{-\ii tH}\varphi(x)}\,\dd x\\
&\leq(\norm{\eul^{-\ii tH}\psi}_2+\norm{\eul^{-\ii tH}\varphi}_2)\norm{\eul^{-\ii tH}(\psi-\varphi)}_2=(\norm{\psi}_2+\norm{\varphi}_2)\norm{\psi-\varphi}_2.\qedhere
\end{align*}
\end{proof}
%-------------------%

To understand this measure further, we state a result which is always valid, regardless of the transport being ballistic or not.

%-------------------%
%:prop 5.2
%-------------------%
\begin{proposition}
Let $H=\C{A}+V$ be a Schr\"odinger operator on a countable graph $G\subset\D{R}^d$ with maximal degree $\leq\R{D}$. Assume $\abs{x-y}_2\leq\R{L}$ for any $x\sim y$ in $G$, where $\abs{\,\cdot\,}_2$ is the Euclidean norm. Assume the potential $V$ is bounded. Suppose $\psi\in\ell^2(G)$ satisfies $\norm{\psi}=1$. Let $\Lambda_{\R{LD}}\coloneqq[-\R{LD},\R{LD}]^d$. Then $\lim\limits_{t\to+\infty} \nu_t^\psi(\Lambda_{\R{LD}})=1$.
\end{proposition}
%-------------------%

Physically, this means that if $X_t$ is a random vector with distribution $\abs{\eul^{-\ii tH}\psi(v)}^2$, $v\in G$, then $X_t/t$ is asymptotically confined in the compact set $\Lambda_{\R{LD}}$, which is independent of $V$, regardless of the transport behavior. If $G=\D{Z}^d$ then $\R{L}=1$ and $\R{D}=2d$.

A quantum walk analog of this result previously appeared in \cite{Sa21}*{Theorem 3.10} for ``smooth'' $\psi$, with a different proof.

%-------------------%
\begin{proof}
This follows from the moment upper bounds, Theorem~\ref{thm:formel}. Fix some $o\in G$ and denote $\abs{x}\coloneqq d(x,o)$ where $d(\,\cdot\,,\,\cdot\,)$ is the graph distance. We first assume that $\norm{x^m\psi}<\infty$ for all $m$. We then know from Remark~\ref{rem:oddup} that 
\begin{equation}\label{eq:disup}
\limsup_{t\to+\infty}\frac{1}{t^{k}}\sum_{x\in G}\abs{x}^{k}|\eul^{-\ii tH}\psi(x)|^2\leq\R{D}^{k}
\end{equation}
for any $k$. To pass from the graph distance to the Euclidean distance $\abs{\,\cdot\,}_2$, we note that given $x\in G$, if $\abs{x}=n$, then there is a path $(v_0,\dots,v_n)$ such that $v_0=o$ and $v_n=x$. So $\abs{x-v_0}_2\leq\sum_{j=0}^{n-1}|v_{j+1}-v_j|_2\leq n\R{L}$ by our assumptions. Thus, $\abs{x-v_0}_2\leq\R{L}\abs{x}$. Since $\abs{x^m}_2^2=\sum_{j=1}^d x_j^{2m} \leq(\sum_{j=1}^d x_j^2)^m=\abs{x}_2^{2m}\leq(\abs{x-v_0}_2+|v_0|_2)^{2m}$, we get $\abs{x^m}_2^2\leq(\R{L}\abs{x}+|o|_2)^{2m}$. It now follows from \eqref{eq:disup} that
\[
\limsup_{t\to+\infty}\expect_{\nu_t^\psi}(\abs{x^m}_2^2)=\limsup_{t\to+\infty} \frac{1}{t^{2m}}\expect_{\mu_t^\psi}(\abs{x^m}_2^2)=\limsup_{t\to+\infty}\frac{1}{t^{2m}}\sum_{x\in G} \abs{x^m}_2^2|\eul^{-\ii tH}\psi(x)|^2 \leq(\R{LD})^{2m}.
\]

Now let $K=\Lambda_{\R{LD}}$, $K_\delta=\croch{-\R{LD}-\delta,\R{LD}+\delta}^d$ and $B\subset K_\delta^c$. Then $x\in B$ implies $|x_i|>\R{LD}+\delta$ for some $i$, which implies $\abs{x^m}_2^2=\sum_{j=1}^d x_j^{2m} > (\R{LD}+\delta)^{2m}$. Thus, $\expect_{\nu_t^\psi}(\abs{x^m}_2^2) \geq\int_B \abs{x^m}_2^2\,\dd\nu_t^\psi(x)\geq(\R{LD}+\delta)^{2m}\nu_t^\psi(B)$. It follows that $\limsup_{t\to+\infty}\nu_t^\psi(B) \leq(\frac{\R{LD}}{\R{LD}+\delta})^{2m}$ for all $m$. Taking $m\to \infty$ yields $\lim_{t\to+\infty} \nu_t^\psi(B)=0$. As $\delta$ is arbitrary and $\nu_t^\psi$ is a probability measure, this shows that $\lim_{t\to+\infty} \nu_t^\psi(K)=1$. This completes the proof when $\norm{x^m\psi}<\infty$ for all $m$. 

For general $\psi$, we argue by approximation. In fact, if $\psi_N\to \psi$, $\psi_N$ has compact support, then $|\nu_t^\psi(K)-1|\leq |\nu_t^{\psi_N}(K)-1| + |\nu_t^\psi(K) - \nu_t^{\psi_N}(K)|$. Using Lemma~\ref{lem:nutcont} then taking $t\to+\infty$, we get $\limsup_{t\to+\infty} |\nu_t^\psi(K)-1| \leq(\norm{\psi_N} + \norm{\psi})\norm{\psi-\psi_N}$. Since $\norm{\psi_N}\to\norm{\psi}$, the claim follows by taking $N\to\infty$.
\end{proof}
%-------------------%

%-------------------%
%:rem 5.3
%-------------------%
\begin{remark}
One may be tempted to deduce the same result in case of the continuous Schr\"odinger operator with smooth potential, evoking Theorem~\ref{thm:momconts}. This does not work however because the constant $C_m$ in Theorem~\ref{thm:momconts} depends on $m$, in contrast to the discrete case. This is not just an artefact, the result is wrong in general in the continuous case. In fact, for the simplest case $H=-\Delta$ on $\D{R}^d$, we have by \eqref{eq:semicts},
\[
\nu_t^\psi(B)=\mu_t^\psi(tB)=\int_{tB}|\eul^{\ii t\Delta}\psi(x)|^2\,\dd x=\frac{1}{(2t)^d}\int_{tB}\Big|\widehat{\phi}_t\Big(\frac{x}{2t}\Big)\Big|^2\,\dd x=\int_{B/2}|\widehat{\phi}_t(y)|^2\,\dd y
\]
for any measurable $B\subset\D{R}^d$, with $\phi_t(y)\coloneqq\eul^{\ii y^2/4t}\psi(y)$. As $t\to+\infty$ we have $\phi_t(y)\to \psi(y)$ pointwise, so $\phi_t\to \psi$ in $L^2$ by dominated convergence, and $\widehat{\phi}_t\to \hat\psi$ in $L^2$ by Parseval. In particular, $\norm{\B{1}_B(\widehat{\phi}_t-\hat\psi)}_2\to 0$. We thus get for any measurable $B\subset\D{R}^d$ and $\psi\in L^2(\D{R}^d)$,
\begin{equation}\label{eq:freectsden}
\nu_t^\psi(B) \to \int_{B/2} |\hat\psi(y)|^2\,\dd y.
\end{equation}
This has no reason to vanish for $B$ outside some compact region. For example, if $\psi(x)=\eul^{-x^2/2}$, then $\hat\psi(y)=\eul^{-y^2/2}$ and $\lim_{t\to+\infty}\nu_t^\psi(B)>0$ for any $B$ of positive measure.
\end{remark}
%-------------------%

On a different note, \eqref{eq:freectsden} also implies the following result.

%-------------------%
%:lem 5.4
%-------------------%
\begin{lemma}\label{lem:feectslim}
Consider $H=-\Delta$ on $\D{R}^d$. Then for any $\psi\in L^2(\D{R}^d)$, $\norm{\psi}_2=1$, $\nu_t^\psi$ converges weakly to the AC measure $\nu_\infty^\psi$ given by $\dd\nu_\infty^\psi(x)=2^{-d}|\hat\psi(\frac{x}{2})|^2\,\dd x$.
\end{lemma}
%-------------------%

%-------------------%
\begin{proof}
This follows immediately from the portemanteau theorem.
\end{proof}
%-------------------%

Lemma~\ref{lem:feectslim} says that if $X_t$ is a random vector with distribution $|\eul^{\ii t\Delta}\psi|^2\,\dd x$ then $\frac{X_t}{t}$ converges in distribution to a random vector $Y$ with distribution $2^{-d}|\hat\psi(\frac{x}{2})|^2\,\dd x$.

%---------------------------------------------------------%
%:s.5.2
%---------------------------------------------------------%
\subsection{Periodic models}
We start with the discrete case. Given a periodic graph $\Gamma$ as in Section~\ref{subsec:perdiscrete}, and a $\D{Z}_{\F{a}}^d$-periodic Schr\"odinger operator
\[
H=\C{A}+Q,
\]
let $\psi\in \ell^2(\Gamma)$ satisfy $\norm{\psi}=1$. Let $\Omega\coloneqq\D{T}_\ast^d\times\accol{1,\dots,\nu}$ and define $\mu^\psi$ on $\Omega$ by
\begin{equation}\label{eq:muinftyper}
\int_\Omega f(\theta,n)\,\dd\mu^\psi(\theta,n)=\int_{\D{T}_\ast^d}\sum_{n=1}^\nu f(\theta,n)\norm{P_n(\theta_{\F{b}})(U\psi)_{\theta_{\F{b}}}}^2 \dd\theta.
\end{equation}
This is a probability measure since 
\[
\mu^\psi(\Omega)=\int_{\D{T}_\ast^d}\sum_{n=1}^\nu\norm{P_n(\theta_{\F{b}})(U\psi)_{\theta_{\F{b}}}}^2 \dd\theta =\norm{U\psi}^2=\norm{\psi}^2=1.
\]
Next, consider the map $h\colon\Omega \to \D{R}^d$ given by $h(\theta,n)\coloneqq\frac{1}{2\pi}\nabla_{\theta_{\F{a}}} E_n(\theta_{\F{b}})$ and let $\nu_\infty^\psi$ be the image measure on $\D{R}^d$ given by
\begin{equation}\label{eq:nuinftyper}
\nu_\infty^\psi\coloneqq h_\star\mu^\psi.
\end{equation}

%-------------------%
%:thm 5.5
%-------------------%
\begin{theorem}\label{thm:perlim}
Let $\Gamma$ be a periodic discrete graph endowed with a periodic Schr\"odinger operator $H$. Let $X_t$ be a random vector with distribution $\abs{\eul^{-\ii tH}\psi(v)}^2$, $v\in \Gamma$. Then $\frac{X_t}{t}$ converges to a random vector $Y$ on $\D{R}^d$ with distribution $\nu_\infty^\psi$ given by \eqref{eq:muinftyper}-\eqref{eq:nuinftyper}.
\end{theorem}
%-------------------%

This agrees with Theorem~\ref{thm:dicreper}, as $\expect_{\nu_\infty^\psi}(\abs{x^m}_2^2)=\sum_{i=1}^d \expect_{\nu_\infty^\psi}(x_i^{2m})=\sum_{i=1}^d\expect_{\mu^\psi}(h_i^{2m}(\theta,n))$ is the RHS of \eqref{eq:dicreper}.

We prove this result using our analysis in Section~\ref{sec:periodic} and the idea from \cite{GJS04} to use the Cram\'er-Wold device. Perhaps the theory in \cite{Sa21} could also be adapted to prove a theorem of this kind.
 
%-------------------%
\begin{proof}
We may assume $\norm{x^m\psi}<\infty$ for all $m$, see Remark~\ref{rem:approxlim}.

We first note that $\nu_\infty^\psi$ is supported on a compact set since $\Omega$ is compact and $h$ is continuous. In one dimension this implies that $Y$ is characterized by its moments. In higher dimensions and lattices things are a bit more complicated. 

Denote $x=\sum_{i=1}^d \alpha_i(x)\F{a}_i$ and let $p_m^{(c)}(x)\coloneqq\bigl(\sum_{i=1}^d c_i\alpha_i(x)\bigr)^m$ where $c=(c_1,\dots,c_d)$. Suppose we showed that $\expect_{\nu_t^\psi}(p_m^{(c)}(x))\to\expect_{\nu_\infty^\psi}(p_m^{(c)}(x))$ for any $m=1,2,\dots$ and $c$. Denoting $x=(x_1,\dots,x_d)$ and $q_m^{(\ell)}(x)\coloneqq\bigl(\sum_{j=1}^d \ell_j x_j\bigr)^m$, since $x_j=x\cdot\F{e}_j=\sum_{i=1}^d\alpha_i(x)\F{a}_i\cdot\F{e}_j$, we get 
\[
q_m^{(\ell)}(x)=\sum_{i,j=1}^d \ell_j \F{a}_i\cdot\F{e}_j\alpha_i(x)=\sum_{i=1}^d c_i\alpha_i(x)=p_m^{(c)}(x)
\]
for $c_i\coloneqq\sum_{j=1}^d\ell_j \F{a}_i\cdot\F{e}_j$. Consequently, we have $\expect_{\nu_t^\psi}(q_m^{(\ell)}(x)) \to \expect_{\nu_\infty^\psi}(q_m^{(\ell)}(x))$ for any $m$ and $\ell$, so $\expect(q_m^{(\ell)}(\frac{X_t}{t}))\to\expect(q_m^{(\ell)}(Y))$ for $Y$ as in the statement. This implies that $\sum_{j=1}^d \ell_j \frac{X_{t,j}}{j}$ converges in distribution to $\sum_{j=1}^d \ell_j Y_j$ for any $\ell$ by \cite{Bill95}*{Theorem 30.2}, where $X_t=(X_{t,1},\dots,X_{t,d})$ and $Y=(Y_1,\dots,Y_d)$. By the Cram\'er-Wold device \cite{Bill95}*{Theorem 29.4}, this implies that $\frac{X_t}{t}$ converges in distribution to $Y$.

So consider $p_m^{(c)}(x)$. Since $h(\theta,n)=\frac{1}{2\pi}\nabla_{\theta_{\F{a}}}E_n(\theta_{\F{b}})=\frac{1}{2\pi}\sum_{i=1}^d\F{a}_i\partial_{\theta_i}E_n(\theta_{\F{b}})$, we have
\begin{align*}
\expect_{\nu_\infty^\psi}(p_m^{(c)}(x))=\expect_{\mu^\psi}(p_m^{(c)}(h)) 
&= \int_{\D{T}_\ast^d}\sum_{n=1}^\nu \Big(\frac{1}{2\pi}\sum_{j=1}^d c_j\partial_{\theta_j}E_n(\theta_{\F{b}})\Big)^m\norm{P_n(\theta_{\F{b}})(U\psi)_{\theta_{\F{b}}}}^2 \dd\theta\\
&= \Bigl\langle U\psi,\Bigl(\int_{\D{T}_\ast^d}^\oplus \sum_{n=1}^\nu \Bigl(\frac{1}{2\pi}\sum_{j=1}^d c_j\partial_{\theta_j}E_n(\theta_{\F{b}})\Bigr)^mP_n(\theta_{\F{b}})\dd\theta\Bigr)U\psi\Bigr\rangle.
\end{align*}
On the other hand, $\expect_{\nu_t^\psi}(p_m^{(c)}(x))=\expect_{\mu_t^\psi}(p_m^{(c)}(x/t))=\sum_{x\in\Gamma} (\sum_{j=1}^d c_j\frac{\alpha_j(x)}{t})^m|\eul^{-\ii tH}\psi(x)|^2=\frac{1}{t^m} \langle \eul^{-\ii tH}\psi, (\sum_{j=1}^d c_j \alpha_j(x))^m \eul^{-\ii tH}\psi\rangle=\frac{1}{t^m}\langle U\psi, U\eul^{\ii tH}(\sum_{j=1}^d c_j \alpha_j(x))^m \eul^{-\ii tH}\psi\rangle$. It thus suffices to prove that
\begin{equation}\label{eq:targetgen}
\lim_{t\to+\infty} \frac{U\eul^{\ii tH}(\sum_{j=1}^d c_j \alpha_j(x))^m \eul^{-\ii tH}\psi}{t^m}=\Big(\int_{\D{T}_\ast^d}^\oplus \sum_{n=1}^\nu \Big(\frac{1}{2\pi}\sum_{j=1}^d c_j\partial_{\theta_j}E_n(\theta_{\F{b}})\Big)^m P_n(\theta_{\F{b}}) \dd\theta\Big)U\psi.
\end{equation}

By the multinomial theorem we have 
\[
\biggl(\sum_{j=1}^dc_j\alpha_j(x)\biggr)^m =\sum_{l_1+\dots+l_r=m} \frac{m!}{l_1!\cdots l_r!} c_1^{l_1}\cdots c_r^{l_r} \alpha_1(x)^{l_1}\cdots \alpha_r(x)^{l_r}.
\]
We showed before \eqref{eq:disfouper} that $\partial_{\theta_i}\eul^{-\ii\theta_{\F{b}}\cdot(k_{\F{a}}+v_n)}=-2\pi\ii(k_i+s_i)\eul^{-\ii\theta_{\F{b}}\cdot(k_{\F{a}}+v_n)}$ if $k_{\F{a}}=\sum_{i=1}^d k_i\F{a}_i$ and $v_n=\sum_{i=1}^d s_i\F{a}_i$, i.e., if $\alpha_i(k_{\F{a}}+v_n)=k_i+s_i$. We thus have $\alpha_i(k_{\F{a}}+v_n)^{l_i}\eul^{-\ii\theta_{\F{b}}\cdot(k_{\F{a}}+v_n)}=(\frac{\ii}{2\pi})^{l_i}\partial_{\theta_i}^{l_i}\eul^{-\ii\theta_{\F{b}}\cdot(k_{\F{a}}+v_n)}$. This yields
\begin{align*}
(U\alpha_1(x)^{l_1}\cdots\alpha_r(x)^{l_r}\phi)_{\theta_{\F{b}}}(v_n) 
&= \sum_{k_{\F{a}}\in\D{Z}^d_{\F{a}}}\eul^{-\ii\theta_{\F{b}} \cdot (k_{\F{a}}+v_n)}\alpha_1(k_{\F{a}}+v_n)^{l_1}\cdots\alpha_r(k_{\F{a}}+v_n)^{l_r}\phi(k_{\F{a}}+v_n) \\
&=\frac{\ii^m}{(2\pi)^m}\partial_{\theta_1}^{l_1}\cdots \partial_{\theta_r}^{l_r}(U\phi)_{\theta_{\F{b}}}(v_n)
\end{align*}
and so,
\begin{align*}
&\biggl(U\eul^{\ii tH}\biggl(\sum_{j=1}^d c_j\alpha_j(x)\biggr)^m \eul^{-\ii tH}\psi\biggr)_{\theta_{\F{b}}}=\eul^{\ii tH(\theta_{\F{b}})}\biggl(U\biggl(\sum_{j=1}^d c_j \alpha_j(x)\biggr)^m\eul^{-\ii tH}\psi\biggr)_{\theta_{\F{b}}}\\
&\qquad\qquad\qquad=\sum_{l_1+\dots+l_r=m} \frac{m!}{l_1!\cdots l_r!} c_1^{l_1}\cdots c_r^{l_r} \frac{\ii^m}{(2\pi)^m} \eul^{\ii tH(\theta_{\F{b}})}\partial_{\theta_1}^{l_1}\cdots \partial_{\theta_r}^{l_r}\eul^{-\ii tH(\theta_{\F{b}})}(U\psi)_{\theta_{\F{b}}}\,.
\end{align*}
As in \eqref{eq:disna},
\begin{align*}
&\partial_{\theta_1}^{l_1}\cdots \partial_{\theta_r}^{l_r}\eul^{-\ii tH(\theta_{\F{b}})}(U\psi)_{\theta_{\F{b}}} \\
&\quad= \sum_{n=1}^\nu \eul^{-\ii tE_n(\theta_{\F{b}})}(-\ii t\partial_{\theta_1}E_n(\theta_{\F{b}}))^{l_1}\cdots(-\ii t\partial_{\theta_r} E_n(\theta_{\F{b}}))^{l_r}P_n(\theta_{\F{b}})(U\psi)_{\theta_{\F{b}}} + \ord_{n,\theta}(t^{l_1+\dots+l_r-1})\,,
\end{align*}
where the error term contains as usual derivatives of $E_n(\theta_{\F{b}})$, $P_n(\theta_{\F{b}})$, and $(U\psi)_{\theta_{\F{b}}}$. Since $l_1+\dots+l_r=m$, using the multinomial theorem again we conclude that \eqref{eq:targetgen} holds, the details are the same as \eqref{eq:detail1}-\eqref{eq:detail2}.
\end{proof}
%-------------------%

%-------------------%
%:exa 5.6
%-------------------%
\begin{example}[the integer lattice]
If $H=\C{A}$ on $\D{Z}^d$, this gives, for Borel $f\colon[-2,2]^d\to\D{C}$, 
\[
\expect_{\nu_\infty^\psi}(f)=\expect_{\mu^\psi}(f(h))=\int_{\D{T}^d}f(-2\sin\theta_1,\dots,-2\sin\theta_d)|\hat\psi(\theta)|^2\dd\theta.
\]
\end{example}
%-------------------%

%-------------------%
%:exa 5.7
%-------------------%
\begin{example}[the triangular lattice]
In this case, we similarly get
\[
\expect_{\nu_\infty^\psi}(f)=\int_{\D{T}^2} f(-2\sin\theta_1-2\sin(\theta_1+\theta_2),-2\sin\theta_2-2\sin(\theta_1+\theta_2))|\hat\psi(\theta)|^2\dd\theta.
\]
\end{example}
%-------------------%

%-------------------%
%:exa 5.8
%-------------------%
\begin{example}[the hexagonal lattice]
The hexagonal distribution can be calculated similarly using the information in Example~\ref{exa:freezd3}. Here we also need $P_{\pm}(\theta_{\F{b}})$. This is simple: if $\eul^{\ii\theta_{\F{b}}\cdot v}\xi(\theta_{\F{b}})=|\xi(\theta_{\F{b}})|\eul^{\ii\phi(\theta_{\F{b}})}$, then the normalized eigenvectors of $H(\theta_{\F{b}})$ are $w_{\pm}(\theta_{\F{b}})=\frac{1}{\sqrt{2}}(1,\pm \eul^{-\ii\phi(\theta_{\F{b}})})^T$, so $P_{\pm}(\theta_{\F{b}})=\scal{w_{\pm},\,\cdot\,}w_{\pm}$. In particular, $\norm{P_{\pm}(\theta_{\F{b}})(U\psi)_{\theta_{\F{b}}}}^2=\frac{|(U\psi)_{\theta_{\F{b}}}(0) \pm \eul^{\ii\phi(\theta_{\F{b}})} (U\psi)_{\theta_{\F{b}}}(v)|^2}{2}$.
\end{example}
%-------------------%

We may now tackle the continuous case.

%-------------------%
%:thm 5.9
%-------------------%
\begin{theorem}
Consider a Schr\"odinger operator $H$ on $\D{R}^d$ with smooth periodic potential $V$ having bounded derivatives and let $\psi\in L^2(\D{R}^d)$. Let $X_t$ be a random vector with distribution $|\eul^{-\ii tH}\psi(x)|^2\,\dd x$. Then $\frac{X_t}{t}$ converges to a random vector $Y$ on $\D{R}^d$ with distribution $\nu_\infty^\psi$ having the analogous expression \eqref{eq:muinftyper}-\eqref{eq:nuinftyper} (so $\sum_{n=1}^\nu$ becomes $\sum_{n=1}^\infty$).
\end{theorem}
%-------------------%

%-------------------%
\begin{proof}
Here $\Omega\coloneqq\D{T}_\ast^d\times\D{N}^\ast$ is no longer compact so we cannot use the method of moments directly. However we can argue by approximation. Namely, given $\psi\in L^2$, we know that $\psi_N\coloneqq U^{-1}\bigl(\int_{\D{T}_\ast^d}^\oplus\sum_{n=1}^NP_n(\theta)\,\dd\theta\bigr)U\psi$ satisfies $\norm{\psi_N-\psi}_2\to 0$ (in fact, in the proof of Theorem~\ref{thm:contper}, Step 3 we showed the stronger fact $\norm{\psi_N-\psi}_{H^{2m}}\to 0$ if $\psi\in H^{2m}(\D{R}^d)$).

For fixed $N$, we may take $\Omega_N\coloneqq\D{T}_\ast^d\times \accol{1,\dots,N}$. Then the proof of Theorem~\ref{thm:perlim} shows that $\nu_t^{\psi_N}$ converges weakly to $\nu_\infty^{\psi_N}$ as $t\to+\infty$. Given $B\subset\Omega$, it remains to control $|\nu_t^{\psi_N}(B)-\nu_t^{\psi}(B)|$ and $|\nu_\infty^{\psi_N}(B)-\nu_\infty^{\psi}(B)|$. The former vanishes uniformly in $t,B$ by Lemma~\ref{lem:nutcont}. Similarly,
\begin{align*}
|\nu_\infty^{\psi}(B)-\nu_\infty^{\psi_N}(B)| 
&=\Bigl\lvert\int_{\D{T}_\ast^d}\sum_{n=1}^\infty\B{1}_{h^{-1}B}(\theta,n)(\norm{P_n(\theta)(U\psi_N)_{\theta}}^2-\norm{P_n(\theta)(U\psi)_\theta}^2)\dd\theta\Bigr\rvert\\
&\leq\int_{\D{T}_\ast^d}\sum_{n=1}^\infty(\norm{P_n(\theta)(U\psi_N)_\theta} + \norm{P_n(\theta)(U\psi)_\theta})\norm{P_n(\theta)((U\psi)_\theta-(U\psi_N)_\theta)}\dd\theta\\
&\leq(\norm{\psi_N}+\norm{\psi})\norm{\psi-\psi_N}
\end{align*}
where we used Cauchy--Schwarz for $\int_{\Omega}$ in the last step and the fact that $\norm{\phi}^2=\norm{U\phi}^2=\int_{\D{T}_\ast^d}\sum_{n=1}^\infty\norm{P_n(\theta)(U\phi)_{\theta}}^2\,\dd\theta$. This completes the proof.
\end{proof}
%-------------------%

%-------------------%
%:rem 5.10
%-------------------%
\begin{remark}\label{rem:approxlim}
The same approximation trick works for any dense subspace. In fact we did not use the explicit form of $\psi_N$ in the previous proof.
\end{remark}
%-------------------%

%---------------------------------------------------------%
%:app.A
%---------------------------------------------------------%
\appendix
\section{Upper bounds and derivatives}\label{sec:app}

Here we first prove that the $m$-moments grow at most like $t^m$. This already appeared in various forms: for continuous Schr\"odinger operators see \cite{RaSi78} for $m=1,2$; for discrete Schr\"odinger operators, upper bounds can be deduced from \cite{AiWa12}*{Appendix B} for general moments but $\psi=\delta_x$. Using the upper bounds, we then give rigorous proofs of the moment derivative formulas \eqref{eq:deridis} and \eqref{eq:dercont}. 

We start with some general remarks on lower bounds.

%---------------------------------------------------------%
%:s.A.1
%---------------------------------------------------------%
\subsection{Lower bounds}\label{subsec:lowerbounds}

On $L^2(X)$, if $x_0\in X$ is fixed and $\abs{x}\coloneqq d(x,x_0)$, then we have
\begin{align}\label{eq:withmod}
\norm{\abs{x}^j\phi}^2=\scal{\phi,\abs{x}^{2j}\phi}=\int_X\abs{x}^{2j}|\phi|^2 
&\leq\Bigl(\int_X\abs{x}^{2m}\abs{\phi}^2\Bigr)^{j/m}\Bigl(\int_X|\phi|^2\Bigr)^{(m-j)/m}\notag \\
&=\norm{\abs{x}^m\phi}^{2j/m}\norm{\phi}^{2(m-j)/m},
\end{align}
where we used H\"older's inequality with $f=\abs{x}^{2j}|\phi|^{2j/m}$ and $g=|\phi|^{(2m-2j)/m}$. It follows that $\norm{\abs{x}^j\phi}^m\leq \norm{\abs{x}^m\phi}^j\norm{\phi}^{m-j}$. In particular,
\begin{equation}\label{eq:mcc}
\liminf_{t\to+\infty}\frac{\norm{\abs{x}^m\eul^{-\ii tH}\psi}^2}{t^{2m}} \geq\Big(\liminf_{t\to+\infty} \frac{\norm{\abs{x}\eul^{-\ii tH}\psi}^{2}}{t^2}\Big)^m\norm{\psi}^{2-2m}.
\end{equation}

On $\D{R}^d$ we usually consider $\norm{x^j\phi}$ instead of $\norm{\abs{x}^j\phi}$, where $x^j\phi\coloneqq(x_1^j\phi,\dots,x_d^j\phi)$. We have $\norm{x_k^j\phi}^2\leq\norm{x_k^m\phi}^{2j/m}\norm{\phi}^{2(m-j)/m}$ by the same argument. Using the Plancherel identity gives the Gagliardo--Nirenberg inequality $\norm{D_k^j\phi}\leq\norm{D_k^m\phi}^{j/m}\norm{\phi}^{(m-j)/m}$.

If $H$ is a Schr\"odinger operator, $x^j_k(t)\coloneqq\eul^{\ii tH}x^j_k\eul^{-\ii tH}$ and $D^j_k(t)\coloneqq\eul^{\ii tH}(-\ii\partial_{x_k})^j\eul^{-\ii tH}$, then this implies
\begin{equation}\label{eq:holnor}
\begin{split}
\norm{x_k^j(t)\psi}&\leq\norm{x_k^m(t)\psi}^{j/m}\norm{\psi}^{(m-j)/m},\\
\norm{D_k^j(t)\psi}&\leq\norm{D_k^m(t)\psi}^{j/m}\norm{\psi}^{(m-j)/m}.
\end{split}
\end{equation}

Lastly in this connection, recall the uncertainty principle $\norm{\phi}^2\leq 2\norm{x_k\phi}\norm{\partial_{x_k}\phi}$. The above yields the generalization $\norm{\phi}^2\leq 2\norm{x_k^m\phi}^{1/m}\norm{\phi}^{(m-1)/m}\norm{\partial_{x_k}^m\phi}^{1/m}\norm{\phi}^{(m-1)/m}$, i.e.
\begin{equation}\label{eq:uncert}
\norm{\phi}^2 \leq 2^m \norm{x_k^m\phi}\cdot\norm{D_k^m\phi}.
\end{equation}
Applying this to $\phi=\eul^{-\ii tH}\psi$, we thus get
\begin{equation}\label{eq:produp}
\norm{x_k^j(t)\psi}\cdot\norm{D_k^j(t)\psi}\leq 2^{m-j}\norm{x_k^m(t)\psi}\cdot\norm{D_k^m(t)\psi}.
\end{equation}
More generally, for $j,n\leq m$,
\begin{equation}\label{eq:produp2}
\norm{x_k^j(t)\psi}\cdot\norm{D_k^n(t)\psi}\leq 2^{\frac{2m-j-n}{2}}\norm{x_k^m(t)\psi}^{\frac{2m+j-n}{2m}}\norm{D_k^m(t)\psi}^{\frac{2m-j+n}{2m}}.
\end{equation}
The preceding estimates provide useful lower bounds for $x^m(t)\psi$ and $D^m(t)\psi$ in terms of lower moments. 

We now consider upper bounds.

%---------------------------------------------------------%
%:s.A.2
%---------------------------------------------------------%
\subsection{Discrete case}  \label{subsec:discret}

In the following, given a countable graph $G$, we fix some vertex $o\in G$ regarded as an origin and denote $\abs{x}\coloneqq d(x,o)$ and $x^m\psi(x)\coloneqq\abs{x}^m\psi(x)$. 

%-------------------%
%:thm A.1
%-------------------%
\begin{theorem}                    \label{thm:formel}
Let $H=\C{A}+V$ be a Schr\"odinger operator on a countable graph $G$ with maximal degree $\leq\R{D}$. We assume the potential $V$ is bounded. Then for any $t\geq 0$ and $m\in\D{N}$, if $\norm{x^m\psi}<\infty$, then \begin{equation}  \label{eq:discretpolynom}
\norm{x^m\eul^{-\ii tH}\psi}\leq\sum_{r=0}^m p_{m-r}(t)\norm{x^r\psi},
\end{equation}
where $p_k(t)$ is a polynomial in $t$ of degree $k$ with $p_0(t)=1$, and the leading term of the top polynomial $p_m(t)$ is $\R{D}^mt^m$. In particular,
\begin{equation}  \label{eq:discretbound}
\limsup_{t\to+\infty}\frac{\norm{x^m\eul^{-\ii tH}\psi}}{t^m}\leq\R{D}^m \norm{\psi}.
\end{equation}
\end{theorem}
%-------------------%

Note that each $p_k(t)$ also depends on $m$, that is, for each fixed $m$ there is a set of polynomials $p_{0,m}(t),\dots,p_{m,m}(t)$ with $p_{k,m}$ of degree $k$ such that \eqref{eq:discretpolynom} is satisfied with $p_k\equiv p_{k,m}$. See Remark~\ref{rem:oddup} for a further comment.

%-------------------%
\begin{proof}
By induction on $m$. The statement is trivial for $m=0$ since $\eul^{-\ii tH}$ is unitary.

Let $m=1$. For an operator $O$, recall we denote $O(t)\coloneqq\eul^{\ii tH}O\eul^{-\ii tH}$. In particular, $x(t)\coloneqq\eul^{\ii tH}x\eul^{-\ii tH}$ when $O$ is the operator of multiplication by $x$. We formally have
\begin{align}\label{eq:diffform}
\frac{\dd}{\dd t} x(t)\psi 
&=\ii H\eul^{\ii tH}x\eul^{-\ii tH}\psi-\ii\eul^{\ii tH}x H\eul^{-\ii tH}\psi\notag\\
&=\ii\eul^{\ii tH}[H,x]\eul^{-\ii tH}\psi\notag\\
&=\ii\eul^{\ii tH}[\C{A},x]\eul^{-\ii tH}\psi\,.
\end{align}
This calculation is formal because the first derivative $\ii H\eul^{\ii tH}x\eul^{-\ii tH}\psi$ requires $x\eul^{-\ii tH}\psi\in\ell^2(G)$, while the second $-\ii\eul^{\ii tH}x H\eul^{-\ii tH}\psi$ requires $\lim_{\delta\to 0}x\frac{\eul^{-\ii(t+\delta)H}-\eul^{-\ii tH}}{\delta}\psi=-\ii xH\eul^{-\ii tH}\psi$. See, e.g., \cite{Ha13}*{Lemma 10.17}. None of these facts is a priori clear (in fact the first point is partly what the theorem tries to prove, we only know that $x\psi\in\ell^2(G)$, a priori). Note that this formal calculation is justified however if instead of $x$ we multiply by a bounded function.

So, similar to \cite{RaSi78}, given $\epsilon>0$, we consider $f_{\epsilon}(\lambda)\coloneqq\frac{\lambda}{1+\epsilon\lambda}$ for $\lambda\geq 0$. Then multiplication by $f_{\epsilon}(\abs{x})$ is a bounded operator and we have $\frac{\dd}{\dd t} [f_{\epsilon}(\abs{x})](t)=\ii\eul^{\ii t H}[\C{A},f_{\epsilon}(\abs{x})]\eul^{-\ii tH}\psi$. But
\[
[\C{A},f_{\epsilon}(\abs{x})]\phi(x)=\sum_{y\sim x}[f_\epsilon(|y|) - f_\epsilon(\abs{x})]\phi(y) =:\sum_{y\sim x}\alpha_{x,y}\phi(y).
\]
Here $|\alpha_{x,y}|=|f_{\epsilon}'(\lambda)|$ for some $\lambda \in [\abs{x}-1,\abs{x}+1]$. Hence, $|\alpha_{x,y}|\leq \frac{1}{(1+\epsilon\lambda)^2}\leq 1$. We thus get
\[
\norm{[\C{A},f_{\epsilon}(\abs{x})]\phi}^2\leq\R{D}\sum_x\sum_{y\sim x}|\phi(y)|^2\leq\R{D}^2\norm{\phi}^2.
\]
Applying this to $\phi=\eul^{-\ii tH}\psi$, we get $\norm{[\C{A},f_{\epsilon}(\abs{x})]\eul^{-\ii tH}\psi}\leq \R{D} \norm{\psi}$. So using \cite{Ca67}*{Theorem~5.6.1},
\begin{align*}
\norm{(f_\epsilon(\abs{x}))(t)\psi}  
&= \Bigl\lVert(f_\epsilon(\abs{x}))(0)\psi + \int_0^t\frac{\dd}{\dd s}(f_\epsilon(\abs{x}))(s)\psi\,\dd s\Bigr\rVert\\
&\leq  \norm{f_\epsilon(\abs{x})\psi} + \int_0^t \norm{[\C{A},f_\epsilon(\abs{x})]\eul^{-\ii sH}\psi}\,\dd s\\
&\leq \norm{x\psi} + t\R{D}\norm{\psi}\,.
\end{align*}
Since $\epsilon>0$ is arbitrary, taking $\epsilon\downarrow 0$ and using Fatou's lemma, we get
\begin{align}\label{eq:fatou}
\norm{x\eul^{-\ii tH}\psi}^2=\sum_x\abs{x}^2|(\eul^{-\ii tH}\psi)(x)|^2
&\leq\liminf_{\epsilon\downarrow 0}\sum_x\frac{\abs{x}^2}{(1+\epsilon \abs{x})^2}|(\eul^{-\ii tH}\psi)(x)|^2\notag\\
&= \liminf_{\epsilon\downarrow 0} \norm{f_\epsilon(\abs{x})\eul^{-\ii tH}\psi}^2\notag\\
&\leq (\norm{x\psi} + t\R{D}\norm{\psi})^2,
\end{align}
where we used $\norm{f_\epsilon(\abs{x})\eul^{-\ii tH}\psi}=\norm{(f_\epsilon(\abs{x}))(t)\psi}$ since $\eul^{\ii tH}$ is unitary.
This settles $m=1$.

Now assume the statement holds for all $k<m$. Let $f_\epsilon(\lambda)=\frac{\lambda^m}{1+\epsilon\lambda^m}$. Here $|f_\epsilon'(\lambda)|\leq m|\lambda|^{m-1}$. Arguing as before, we get
\begin{align*}
\norm{[\C{A},f_{\epsilon}(\abs{x})]\phi}^2
&\leq \R{D}m^2\sum_x(\abs{x}+1)^{2(m-1)}\sum_{y\sim x}|\phi(y)|^2\\
&\leq \R{D}^2m^2 \norm{(\abs{x}+2)^{m-1}\phi}^2\,.
\end{align*}
Hence,
\[
\norm{(f_\epsilon(\abs{x}))(t)\psi}\leq \norm{f_\epsilon(\abs{x})\psi} + \R{D}m\sum_{q=0}^{m-1}\binom{m-1}{q}2^{m-1-q}\int_0^t \norm{x^q\eul^{-\ii sH}\psi}\,\dd s.
\]
Since $f_\epsilon(\abs{x})\leq\abs{x}^m$, using the induction hypothesis we get
\[
\norm{(f_\epsilon(\abs{x}))(t)\psi}\leq \norm{x^m\psi}+\R{D}m\sum_{q=0}^{m-1}\binom{m-1}{q}2^{m-1-q}\sum_{r=0}^q\tilde{p}_{q-r+1}(t)\norm{x^r\psi}\,.
\]
As the RHS is independent of $\epsilon$, taking $\epsilon \downarrow 0$ and using Fatou's lemma again yields $\norm{x^m(t)\psi}\leq \sum_{s=0}^m p_{m-s}(t)\norm{x^s\psi}$. The above also shows the coefficient of $\norm{x^m\psi}$ is $p_0(t)=1$. The top polynomial $p_m(t)$ is found by taking $q=m-1$ and $r=0$ and equals $\R{D}m\tilde{p}_m(t)$, where $\tilde{p}_m(t)\coloneqq\int_0^t p_{m-1}(s)\,\dd s$. As the leading term of $p_{m-1}(s)$ is $\R{D}^{m-1}s^{m-1}$ by hypothesis, the leading term of $\R{D}m\int_0^t p_{m-1}(s)\,\dd s$ is $\R{D}^m t^m$.  
\end{proof}
%-------------------%

A posteriori, the formal differentiation \eqref{eq:diffform} is actually valid. Recall $x(t)\coloneqq\eul^{\ii tH}x\eul^{-\ii tH}$.

%-------------------%
%:cor A.2
%-------------------%
\begin{corollary}\label{cor:deri}
Under the same assumptions, $\lim\limits_{s\to t}x^m(s)\psi=x^m(t)\psi$ and 
\begin{equation}\label{eq:deridis}
\frac{\dd}{\dd t}x^m(t)\psi=\ii\eul^{\ii tH}[\C{A},x^m]\eul^{-\ii tH}\psi\,.
\end{equation}
\end{corollary}
%-------------------%

%-------------------%
\begin{proof}
We have
\begin{align}\label{eq:contmoms}
\norm{x^m(s)\psi - x^m(t)\psi}
&= \norm{\eul^{\ii sH}x^m\eul^{-\ii sH}\psi - \eul^{\ii tH}x^m\eul^{-\ii tH}\psi}\notag \\
&\leq \norm{\eul^{\ii sH} x^m \eul^{-\ii sH}\psi - \eul^{\ii sH} x^m \eul^{-\ii tH}\psi} + \norm{\eul^{\ii sH}x^m\eul^{-\ii tH}\psi - \eul^{\ii tH}x^m\eul^{-\ii tH}\psi}\notag\\
&= \norm{ x^m \eul^{-\ii sH}\psi - x^m \eul^{-\ii tH}\psi} + \norm{\eul^{\ii sH}x^m\eul^{-\ii tH}\psi - \eul^{\ii tH}x^m\eul^{-\ii tH}\psi}\,.
\end{align}

We know from Theorem~\ref{thm:formel} that $\phi=x^m\eul^{-\ii tH}\psi\in D(H)=\ell^2(G)$ for any $t\geq 0$, so $\lim_{s\to t}\eul^{\ii sH}x^m\eul^{-\ii tH}\psi=\eul^{\ii tH}x^m\eul^{-\ii tH}\psi$. This settles the second term in the RHS. 

For the first term, we use Fatou's lemma as in \eqref{eq:fatou}. Let $f_\epsilon(\lambda)=\frac{\lambda^m}{1+\epsilon\lambda^m}$. We have $f_{\epsilon}(\abs{x})\eul^{-\ii tH}\psi=f_\epsilon(\abs{x})\eul^{-\ii sH}\psi+\int_s^t\frac{\dd}{\dd\alpha}f_\epsilon(\abs{x})\eul^{-\ii\alpha H}\psi\,\dd\alpha$. Now
\[
\Bigl\lVert\frac{\dd}{\dd \alpha}f_\epsilon(\abs{x})\eul^{-\ii\alpha H}\psi\Bigr\rVert= \norm{f_\epsilon(\abs{x})H\eul^{-\ii\alpha H}\psi} \leq \norm{x^m \eul^{-\ii\alpha H}H\psi} \leq \sum_{r=0}^{m} p_{m-r}(\alpha)\norm{x^r H\psi}
\]
for some polynomials $p_k$, by Theorem~\ref{thm:formel}. These are uniformly bounded by some $M(t)$ for all $\alpha\in [t-1,t+1]$ and we get $\norm{f_\epsilon(\abs{x}) \eul^{-\ii sH}\psi - f_\epsilon(\abs{x}) \eul^{-\ii tH}\psi} \leq \left|t-s\right|M(t)\sum_{r=0}^{m}\norm{x^rH\psi}$, with $M(t)$ independent of $\epsilon$. 

By Fatou's lemma, $\norm{x^m \eul^{-\ii sH}\psi - x^m \eul^{-\ii tH}\psi}^2 \leq \liminf\limits_{\epsilon\to 0}\norm{f_\epsilon(\abs{x}) \eul^{-\ii sH}\psi - f_\epsilon(\abs{x}) \eul^{-\ii tH}\psi}^2$. Thus, $\norm{x^m \eul^{-\ii sH}\psi - x^m \eul^{-\ii tH}\psi}\leq \left|t-s\right|M(t)\sum_{r=0}^{m-1}\norm{x^rH\psi}\to 0$ as $s\to t$. Recalling \eqref{eq:contmoms}, this completes the proof of the first claim.

For the derivative we first make some simplifications. Given $\delta>0$,
\begin{align}\label{eq:derfo}
&\Bigl\lVert\frac{x^m(t+\delta)\psi - x^m(t)\psi}{\delta} - \ii\eul^{\ii tH}[H,x^m]\eul^{-\ii tH}\psi\Bigr\rVert \notag\\
&\quad =\Bigl\lVert\frac{\eul^{\ii\delta H} x^m\eul^{-\ii(t+\delta)H}\psi - x^m\eul^{-\ii tH}\psi}{\delta} - \ii [H,x^m]\eul^{-\ii tH}\psi\Bigr\rVert \notag\\
&\quad \leq  \Bigl\lVert\frac{\eul^{\ii\delta H}  - I}{\delta}x^m\eul^{-\ii(t+\delta)H}\psi - \ii H x^m\eul^{-\ii(t+\delta)H}\psi\Bigr\rVert \notag\\
&\qquad + \norm{Hx^m(\eul^{-\ii tH}\psi-\eul^{-\ii(t+\delta)H})\psi} + \Bigl\lVert x^m\Bigl(\frac{\eul^{-\ii(t+\delta)H} - \eul^{-\ii tH}}{\delta}\psi +\ii H\eul^{-\ii tH}\psi\Bigr)\Bigr\rVert\notag\\
&\quad \leq \Bigl\lVert\Bigl(\frac{\eul^{\ii\delta H}  - I}{\delta}-\ii H\Bigr)x^m\eul^{-\ii tH}\psi\Bigr\rVert+\Bigl\lVert\Bigl(\frac{\eul^{\ii\delta H}  - I}{\delta}-\ii H\Bigr)x^m(\eul^{-\ii(t+\delta)H}\psi - \eul^{-\ii tH}\psi)\Bigr\rVert\notag\\
&\qquad + \norm{Hx^m(\eul^{-\ii tH}\psi-\eul^{-\ii(t+\delta)H})\psi} + \Bigl\lVert x^m\Big(\frac{\eul^{-\ii(t+\delta)H} - \eul^{-\ii tH}}{\delta}\psi +\ii H\eul^{-\ii tH}\psi\Big)\Bigr\rVert\,.
\end{align}

For the first term, we know from Theorem~\ref{thm:formel} that $x^m\eul^{-\ii tH}\psi\in D(H)=\ell^2(G)$, so this term vanishes as $\delta\to 0$. For the second term, we use the spectral theorem: $\norm{\frac{\eul^{\ii\delta H}-I}{\delta}\phi}^2=\int|\frac{\eul^{\ii\delta\lambda}-1}{\delta}|^2\,\dd\mu_\phi(\lambda)\leq \int\lambda^2\,\dd\mu_\phi(\lambda)=\norm{H\phi}^2$. With this bound, we see the second and third terms vanish as $\delta\to 0$ by the argument of \eqref{eq:contmoms} (note that $H$ is bounded; see also Corollary~\ref{cor:deriv} for unbounded operators). So it remains to control the last term. For this we first use Fatou's lemma to replace $x^m$ by $f_\epsilon(\abs{x})$ as follows.

We know that 
\[
\frac{\dd}{\dd s} f_\epsilon(\abs{x})\eul^{-\ii sH}\psi=f_\epsilon(\abs{x})(-\ii H\eul^{-\ii sH})\psi\quad\text{and}\quad\frac{\dd}{\dd s} f_\epsilon(\abs{x})(-\ii H\eul^{-\ii sH}\psi)=-f_\epsilon(\abs{x}) H^2\eul^{-\ii sH}\psi. 
\]
It follows from \cite{Ca67}*{Theorem~5.6.2} that for small $\delta$,
\[
\Bigl\lVert f_\epsilon(\abs{x})\Big(\frac{\eul^{-\ii(t+\delta)H} - \eul^{-\ii tH}}{\delta}\psi +\ii H\eul^{-\ii tH}\psi\Big)\Bigr\rVert \leq \frac{\abs{\delta}}{2}\sup_{s\in [t-1,t+1]} \norm{x^mH^2\eul^{-\ii sH}\psi},
\]
where we used that $f_\epsilon(\abs{x})\leq\abs{x}^m$. Using Theorem~\ref{thm:formel} again, we may bound the RHS by $\frac{\abs{\delta}}{2}M(t)\sum_{r=0}^m\norm{x^r H^2\psi}$, with $M(t)$ independent of $\epsilon$. Fatou's lemma implies as before that the last term in \eqref{eq:derfo} is now bounded by $\frac{\abs{\delta}}{2}M(t)\sum_{r=0}^m\norm{x^rH^2\psi}$. Taking $\delta\to 0$ finally completes the proof.
\end{proof}
%-------------------%

%-------------------%
%:rem A.3
%-------------------%
\begin{remark}\label{rem:oddup}
Theorem~\ref{thm:formel} implies that 
\[
\limsup_{t\to+\infty}\frac{1}{t^{2m}}\sum_{x\in G}\abs{x}^{2m}|\eul^{-\ii tH}\psi(x)|^2 \leq \R{D}^{2m}\norm{\psi}^2.
\]
This also implies a control for odd powers. Namely, if $\norm{x^m\psi}<\infty$, letting $\psi_t\coloneqq\eul^{-\ii tH}\psi$, we have by Cauchy--Schwarz that $\sum \abs{x}^m |\psi_t(x)|^2 \leq (\sum \abs{x}^{2m}|\psi_t(x)|^2)^{1/2}(\sum|\psi_t(x)|^2)^{1/2}$. So $\frac{1}{t^m}\sum \abs{x}^m |\psi_t(x)|^2 \leq (\frac{1}{t^{2m}}\sum \abs{x}^{2m}|\psi_t(x)|^2)^{1/2}\norm{\psi}$ and thus 
\[
\limsup_{t\to+\infty}\frac{1}{t^{m}}\sum_{x\in G}\abs{x}^{m}|\psi_t(x)|^2 \leq\R{D}^{m}\norm{\psi}^2.
\]
\end{remark}
%-------------------%

%---------------------------------------------------------%
%:s.A.3
%---------------------------------------------------------%
\subsection{Continuous case}  \label{subsec:continu}

Assume now that on $\D{R}^d$, we have a potential $V\in C^{m-1}$ such that $V$ and its partial derivatives of order $<m$ are bounded. Let $H=H_0+V=-\Delta+V$. Then we claim that
\begin{equation}\label{eq:pre}
\norm{D^m\phi}^2\leq C_{m,V}\sum_{k=0}^m\norm{H^k\phi}.
\end{equation}
Indeed, using $X^m-Y^m=\sum_{p=0}^{m-1}X^p(X-Y)Y^{m-1-p}$, we have
\[
\norm{D^m\phi}^2=\scal{\phi,D^{2m}\phi}\leq\scal{\phi,H_0^m\phi}= \scal{\phi,H^m\phi}-\sum_{p=0}^{m-1}\scal{\phi,H_0^pVH^{m-1-p}\phi}.
\]
with the convention $\sum_{p=0}^{-1}\coloneqq 0$.

Now \eqref{eq:pre} is clear for $m=0$. If \eqref{eq:pre} holds for all $p<m$, then using Cauchy--Schwarz, Leibniz formula, and our assumption on $V$, we get
\[
|\scal{\phi, H_0^p VH^{m-1-p}\phi}|\leq c_{d,p}\norm{D^p\phi}\cdot\norm{D^p VH^{m-1-p}\phi}\leq c_{V,m} \sum_{q=0}^p \norm{D^p\phi}\cdot\norm{D^qH^{m-1-p}\phi}.
\]
So by induction hypothesis,
\[
\norm{D^m\phi}^2 \leq \norm{\phi}\cdot\norm{H^m\phi} + c_{V,m}'\sum_{p=0}^{m-1}\sum_{r=0}^p\sum_{q=0}^p\sum_{s=0}^q \norm{H^r\phi}\cdot\norm{H^{m-1-p+s}\phi}.
\]

Using $ab\leq\frac{1}{2}(a^2+b^2)$, we thus get $\norm{D^m\phi}^2 \leq c_{V,m}''\sum_{k=0}^m\norm{H^k\phi}^2$, implying \eqref{eq:pre}. This implies that for $D(t)\coloneqq\eul^{\ii tH}D\eul^{-\ii tH}$ and $\psi_t\coloneqq\eul^{-\ii tH}\psi$, we have
\begin{equation}\label{eq:ptr}
\norm{D(t)^r\psi}=\norm{D^r\psi_t}\leq C_{r,V}\sum_{k=0}^r\norm{H^k\psi} \leq C_{r,V}'\norm{\psi}_{H^{2r}},
\end{equation}
independently of $t$, where we used that $H^k\eul^{-\ii tH}=\eul^{-\ii tH}H^k$. 

%-------------------%
%:thm A.4
%-------------------%
\begin{theorem}\label{thm:momconts}
If $V\in C^{m-1}(\D{R}^d)$, if $V$ and its partial derivatives of order $<m$ are bounded, and if $\psi\in H^{2m}(\D{R}^d)$, then $\norm{x^m\eul^{-\ii tH}\psi}\leq 2^{m-1}(\norm{x^m\psi}+C_mt^m \sum_{k=0}^m\norm{H^k\psi})$.
\end{theorem}
%-------------------%

A sketch of an earlier result can also be found in \cite{RaSi78}*{Theorem~4.1}. We first give a formal proof, then indicate how to make it rigorous.

%-------------------%
\begin{proof}[Proof (formal)]
Recall that $x^m=(x_1^m,\dots,x_d^m)$. In this proof we denote $x^{2m}= x_1^{2m}+\dots+x_d^{2m}$ instead of $\abs{x^m}_2^2$ to avoid too cumbersome formulas.

Formally, $\frac{\dd}{\dd t} x^{2m}(t)\psi=\ii\eul^{\ii tH}[-\Delta, x^{2m}]\eul^{-\ii tH}\psi$ for $\psi\in D(H)$. But, for $F$ smooth on $\D{R}^d$ we have $[-\Delta,F]\phi=-(\Delta F)\phi - 2\nabla F\cdot\nabla\phi=-\nabla\cdot[(\nabla F)\phi]-\nabla F\cdot\nabla\phi$. In particular, for $F(x)=x^{2m}$, since $\nabla x^{2m}=2m(x_1^{2m-1},\dots,x_d^{2m-1})$, we get
\begin{align*}
\frac{\dd}{\dd t}\norm{x^m(t)\psi}^2=\frac{\dd}{\dd t}\scal{\psi,x^{2m}(t)\psi} 
&= \ii\scal{\psi_t,[-\Delta,x^{2m}] \psi_t}  \\
&= \scal{\psi, D(t)\cdot[(\nabla x^{2m})(t)\psi] + (\nabla x^{2m})(t)\cdot D(t)\psi}\\
&\leq 4m \sum_{j=1}^d\norm{x_j^m(t)\psi}\cdot\norm{x_j^{m-1}(t)D_j(t)\psi},
\end{align*}
where $D_j\coloneqq-\ii\partial_{x_j}$ and $D_j(t)\coloneqq\eul^{\ii tH}D_j\eul^{-\ii tH}$. We have in general
\begin{align}
\norm{x_j^n(t)D_j^k(t)\psi}&=\scal{\psi,D_j^k(t)x_j^{2n}(t)D_j^k(t)\psi}^{1/2}\notag\\
\label{eq:toapply}
&\leq\sum_{p=0}^k\binom{k}{p}(2n)\cdots(2n-p)\abs{\scal{\psi,x_j^{2n-p}(t)D_j^{2k-p}(t)\psi}}^{1/2}\\
\label{eq:toapply2}
&\leq\sum_{p=0}^k c_{p,k,n}\norm{x_j^m(t)\psi}^{1/2}\norm{x_j^{2n-p-m}(t)D_j^{2k-p}(t)\psi}^{1/2}.
\end{align}

We may apply the same inequality to $\norm{x_j^{2n-p-m}(t)D_j^{2k-p}(t)\psi}$. Doing this $\ell-1$ times we see that the term with highest power is
\[
C_{k,n}\norm{x_j^m(t)\psi}^{\frac{1}{2}+\frac{1}{4}+\dots+\frac{1}{2^{\ell-1}}}\norm{x_j^{2^{\ell-1}(n-m)+m}(t)D_j^{2^{\ell-1}k}(t)\psi}^{\frac{1}{2^{\ell-1}}}.
\]
%-------------------%
%:m=2^\ell
%-------------------%
\begin{mpoweroftwo*}
Let $n=m-1=2^\ell-1$. Then by applying \eqref{eq:toapply2} $\ell-1$ times, we get
\begin{equation}\label{eq:compactform}
\frac{\dd}{\dd t}\norm{x(t)^m\psi}^2\leq C_m\sum_{j=1}^d\norm{x_j^m(t)\psi}\cdot\norm{x_j^m(t)\psi}^{1-\frac{1}{2^{\ell-1}}}\sum_{r=0}^{2^{\ell-1}}c_{r,m}\norm{x_j^r(t)D_j^r(t)\psi}^{\frac{1}{2^{\ell-1}}}.
\end{equation}
In fact, the terms have the general form
\[
x_j^{2^{\ell-1}(n-m)+m-2^{\ell-2}p_1-2^{\ell-3}p_2-\dots-p_{\ell-1}}(t)D_j^{2^{\ell-1}k-2^{\ell-2}p_1-\dots-p_{\ell-1}}(t)\psi.
\]
For $m=2^\ell$, $n=m-1$, $k=1$, we see the powers of $x_j(t)$ and $D_j(t)$ match indeed.

We next apply \eqref{eq:toapply} plus Cauchy--Schwarz to get
\begin{align}\label{eq:compa}
&\frac{\dd}{\dd t}\norm{x^m(t)\psi}^2\notag\\
&\qquad\leq C_m\sum_{j=1}^d\norm{x_j^m(t)\psi}^{2-\frac{1}{2^{\ell-1}}}\sum_{r=0}^{2^{\ell-1}} c_{r,m}\sum_{p=0}^rc_{p,r}\norm{x_j^{2r-p}(t)\psi}^{\frac{1}{2^\ell}}\norm{D_j^{2r-p}(t)\psi}^{\frac{1}{2^{\ell}}}\,.
\end{align}
Recalling \eqref{eq:produp} and \eqref{eq:ptr}, we conclude that for $m=2^\ell$,
\begin{align}\label{eq:goon}
\frac{\dd}{\dd t}\norm{x^m(t)\psi}^2 
&\leq C_{m,d}\sum_{j=1}^d\norm{x_j^m(t)\psi}^{2-\frac{1}{m}}\norm{D_j^m(t)\psi}^{1/m}\notag\\
&\leq C_{m,d,V}\norm{x^m(t)\psi}^{2-\frac{1}{m}}\sum_{k=0}^m\norm{H^k\psi}^{1/m}.
\end{align}
Thus, 
\[
\frac{\dd}{\dd t}\scal{\psi,x^{2m}(t)\psi}^{\frac{1}{2m}}=\frac{1}{2m}\scal{\psi,x^{2m}(t)\psi}^{\frac{1}{2m}-1}\frac{\dd}{\dd t}\scal{\psi,x^{2m}(t)\psi}\leq c\sum_{k=0}^m\norm{H^k\psi}^{1/m}.
\]
We thus have $\norm{x^m(t)\psi}^{1/m}\leq\norm{x^m\psi}^{1/m}+ Ct\sum_{k=0}^m\norm{H^k\psi}^{1/m}$. The result follows in this case.
\end{mpoweroftwo*}
%-------------------%
%:m arbitrary
%-------------------%
\begin{marbitrary*}
For general $m$ we let $\ell$ such that $2^{\ell} \leq m < 2^{\ell+1}$. Say $m=2^\ell+q$ with $0\leq q<2^\ell$. Then following the scheme, we apply \eqref{eq:toapply2} $\ell-1$ times. Then $\norm{x_j^r(t)D_j^r(t)\psi}^{\frac{1}{2^{\ell-1}}}$ in \eqref{eq:compactform} is replaced by $\norm{x_j^{r+q}(t)D_j^r(t)\psi}^{\frac{1}{2^{\ell-1}}}$. The proof must be slightly modified as now $2r+2q\leq 2^\ell+2q=m+q$, i.e., the powers of $x_j(t)$ in \eqref{eq:compa} can exceed $m$. So to the $q$ highest terms $r=2^{\ell-1}-q+1,\dots,2^{\ell-1}$, we apply \eqref{eq:toapply2} once more to get $\sum_{p=0}^r c_{p,q,r}\norm{x_j^m(t)\psi}^{\frac{1}{2^\ell}}\norm{x_j^{2r+2q-p-m}(t)D_j^{2r-p}(t)\psi}^{\frac{1}{2^\ell}}$. We can now apply \eqref{eq:toapply} plus Cauchy--Schwarz to this and the lower terms as before. Then \eqref{eq:compa} is replaced by
\begin{align} \label{eq:e}
\frac{\dd}{\dd t}\norm{x(t)^m\psi}^2 
&\leq C_m\sum_{j=1}^d\norm{x_j^m(t)\psi}^{2-\frac{1}{2^{\ell-1}}}\Big(\sum_{r=2^{\ell-1}-q+1}^{2^{\ell-1}}\sum_{p=0}^r c_{p,q,r}\norm{x_j^m(t)\psi}^{\frac{1}{2^\ell}}\notag\\
&\qquad\times\sum_{p'=0}^{2r-p}c_{p',r,p,q,m}\norm{x_j^{4r+4q-2p-2m-p'}(t)\psi}^{\frac{1}{2^{\ell+1}}}\norm{D_j^{4r-2p-p'}(t)\psi}^{\frac{1}{2^{\ell+1}}}\notag\\
&\qquad+\sum_{r=0}^{2^{\ell-1}-q} c_{r,m}\sum_{p=0}^rc_{p,r}\norm{x_j^{2r+2q-p}(t)}^{\frac{1}{2^\ell}}\norm{D_j^{2r-p}(t)\psi}^{\frac{1}{2^{\ell}}}\Big)\,.
\end{align}
We may now apply \eqref{eq:produp2} and \eqref{eq:ptr} to get
\[
\norm{x_j^{4r+4q-2p-2m-p'}(t)\psi}\cdot\norm{D_j^{4r-2p-p'}(t)\psi}\leq 2^{2m-r-2q}\norm{x_j^m(t)\psi}^{\frac{2q}{m}}\norm{D_j^m(t)\psi}^{\frac{2m-2q}{m}}.
\]
Recall $q=m-2^{\ell}$, so $\frac{2q}{2^{\ell+1}m}=\frac{1}{2^\ell}-\frac{1}{m}$. In the first sum of \eqref{eq:e}, $\norm{x^m_j(t)\psi}$ thus gets elevated to the power $2-\frac{1}{2^{\ell-1}}+\frac{1}{2^\ell}+\frac{1}{2^\ell}-\frac{1}{m}=2-\frac{1}{m}$ as required. For the lower terms, the power is similarly $2-\frac{1}{2^{\ell-1}}+\frac{2m+2q}{2^{\ell+1}m}=2-\frac{1}{m}$. This completes the formal proof.
\end{marbitrary*}
%-------------------%
\renewcommand{\qed}{}
\end{proof}
%-------------------%

%-------------------%
\begin{proof}[Proof (completed)]
To make the formal proof rigorous we consider the operator of multplication by
\[
F_\epsilon(x)\coloneqq\frac{x^{2m}}{1+\epsilon x^{2m}}=\frac{x_1^{2m}+\dots+x_d^{2m}}{1+\epsilon(x_1^{2m}+\dots+x_d^{2m})},\quad\epsilon>0.
\]
This is a bounded operator. For $F_\epsilon(x)(t)\coloneqq\eul^{\ii tH}F_\epsilon(x)\eul^{-\ii tH}$ we get
\[
\frac{\dd}{\dd t}F_\epsilon(x)(t)\psi=\ii\eul^{\ii tH}[-\Delta, F_\epsilon(x)]\eul^{-\ii tH}\psi\text{ for }\psi\in D(H). 
\]
Again $[-\Delta,F_\epsilon]\phi=-(\Delta F_\epsilon)\phi - 2\nabla F_\epsilon\cdot\nabla\phi=-\nabla\cdot[(\nabla F_\epsilon)\phi]-\nabla F_\epsilon\cdot\nabla\phi$. On the other hand $\nabla F_\epsilon=2m(\frac{x_1^{2m-1}}{(1+\epsilon x^{2m})^2},\dots,\frac{x_d^{2m-1}}{(1+\epsilon x^{2m})^2})$. So 
\begin{align*}
&\frac{\dd}{\dd t}\scal{\psi,F_\epsilon(x)(t)\psi}=\scal{\psi,D(t)\cdot [(\nabla F_\epsilon)(x)(t))\psi]+(\nabla F_\epsilon)(x)(t)\cdot D(t)\psi}\\
&\qquad\quad=2m\sum_{j=1}^d\Bigl\langle\frac{x_j^{m-1}}{1+\epsilon x^{2m}}(t)D_j(t)\psi,\frac{x_j^{m}}{1+\epsilon x^{2m}}(t)\psi\Bigr\rangle+ \Bigl\langle\frac{x_j^{m}}{1+\epsilon x^{2m}}(t)\psi,\frac{x_j^{m-1}}{1+\epsilon x^{2m}}(t)D_j(t)\psi\Bigr\rangle\\
&\qquad\quad\leq 4m\sum_{j=0}^d \Bigl\lVert\frac{x_j^{m}}{1+\epsilon x^{2m}}(t)\psi\Bigr\rVert\Bigl\lVert\frac{x_j^{m-1}}{1+\epsilon x^{2m}}(t)D_j(t)\psi\Bigr\rVert.
\end{align*}
We have $\bigl\lVert\frac{x_j^{m}}{1+\epsilon x^{2m}}(t)\psi\bigr\rVert=\bigl\langle\psi_t,\frac{x_j^{2m}}{(1+\epsilon x^{2m})^2}\psi_t\bigr\rangle^{1/2}\leq\scal{\psi_t, F_\epsilon(x)\psi_t}^{1/2}=\scal{\psi, F_\epsilon(x)(t)\psi}^{1/2}$ where $\psi_t\coloneqq\eul^{-\ii tH}\psi$. On the other hand,
\begin{align*}
\Bigl\lVert\frac{x_j^r}{1+\epsilon x^{2m}}(t)D_j^k(t)\psi\Bigl\lVert 
&=\Bigl\langle\psi,D_j^k(t)\frac{x_j^{2r}}{(1+\epsilon x^{2m})^2}(t)D_j^k(t)\psi\Bigr\rangle^{1/2}\\
&\leq\sum_{p=0}^k \binom{k}{p}\Bigl\lvert\Bigl\langle\psi_t,D_j^p\frac{x_j^{2r}}{(1+\epsilon x^{2m})^2}D_j^{2k-p}\psi_t\Bigr\rangle\Bigr\rvert^{1/2}
\end{align*}
and
\[
\partial_{x_j}^p\frac{x_j^{2r}}{(1+\epsilon x^{2m})^2}=\sum_{\ell=0}^p\binom{p}{\ell}(2r)\cdots(2r-\ell+1) x_j^{2r-\ell}\partial_{x_j}^{p-\ell}\frac{1}{(1+\epsilon x^{2m})^2}.
\]
If $f_\epsilon(u)\coloneqq\frac{1}{(1+\epsilon u)^2}$ and $g(x)\coloneqq x^{2m}$ then by the Fa\`a di Bruno formula,
\[
\partial_{x_j}^n\frac{1}{(1+\epsilon x^{2m})^2}=\partial_{x_j}^nf_\epsilon(g(x))=\sum_{\substack{(m_1,\dots,m_n)\\\sum_{i=1}^nim_i=n}}c_{n,m_i}f_\epsilon^{(m_1+\dots+m_n)}(g(x))\prod_{i=1}^n(\partial_{x_j}^ig(x))^{m_i}\,.
\]
But $f_\epsilon^{(q)}(u)=(-\epsilon)^q(q+1)!(1+\epsilon u)^{-2-q}$ and $\partial_{x_j}^ig(x)=(2m)\cdots(2m-i+1)x_j^{(2m-i)}$. Thus,
\begin{align*}
\partial_{x_j}^n\frac{1}{(1+\epsilon x^{2m})^2} &=\sum_{\substack{(m_1,\dots,m_n)\\\sum_{i=1}^nim_i=n}}\tilde{c}_{n,m_i} \frac{\epsilon^{m_1+\dots+m_n}}{(1+\epsilon x^{2m})^{2+m_1+\dots+m_n}}(2m)x_j^{(2m-1)m_1}\\
&\qquad\quad\times (2m)(2m-1)x_j^{(2m-2)m_2}\cdots(2m)\cdots(2m-n+1)x_j^{(2m-n)m_n}\\
&=\sum_{\substack{(m_1,\dots,m_n)\\\sum_{i=1}^nim_i=n}}C_{n,m_i} \frac{\epsilon^{m_1+\dots+m_n}}{(1+\epsilon x^{2m})^{2+m_1+\dots+m_n}}x_j^{2m(m_1+\dots+m_n)-n}\,,
\end{align*}
where we used $\sum im_i=n$ in the last step. But 
\[
\epsilon ^{m_1+\dots+m_n}x_j^{2m(m_1+\dots+m_n)}\leq (1+\epsilon x^{2m})^{m_1+\dots+m_n}.
\]
Thus,
\begin{align*}
\partial_{x_j}^p\frac{x_j^{2r}}{(1+\epsilon x^{2m})^2}
&\leq\sum_{\ell=0}^p \binom{p}{\ell}(2r)\cdots(2r-\ell+1)x_j^{2r-\ell}\sum_{(m_1,\dots,m_{p-\ell})} C_{p-\ell,m_i}\frac{x_j^{\ell-p}}{(1+\epsilon x^{2m})^2}\\
&= \sum_{\ell=0}^p c_{p,\ell,r}\frac{x_j^{2r-p}}{(1+\epsilon x^{2m})^2}.
\end{align*}
It follows that
\begin{align*}
\Bigl\lVert\frac{x_j^r}{1+\epsilon x^{2m}}(t)D_j^k(t)\psi\Bigr\rVert
&\leq\sum_{p=0}^k c_{p,k,r}\Bigl\lvert\Bigl\langle\psi,\frac{x_j^{2r-p}}{(1+\epsilon x^{2m})^2}(t)D_j^{2k-p}(t)\psi\Bigr\rangle\Bigr\rvert^{1/2}\\
&\leq\sum_{p=0}^k c_{p,k,r}\Bigl\lVert\frac{x_j^m}{1+\epsilon x^{2m}}(t)\psi\Bigl\lVert^{1/2}\Bigl\lVert\frac{x_j^{2r-p-m}}{1+\epsilon x^{2m}}(t)D_j^{2k-p}(t)\psi\Bigr\rVert^{1/2}.
\end{align*}
This proves the analog of \eqref{eq:toapply}-\eqref{eq:toapply2}. From here, the proof goes as before and we get 
\[
\Bigl\lVert\frac{x_j^m}{1+\epsilon x^{2m}}\psi\Bigr\rVert^{1/m}\leq\norm{x^m\psi}^{1/m}+Ct\sum_{k=0}^m\norm{H^k\psi}^{1/m},
\]
independently of $\epsilon$. The result follows by taking $\epsilon\downarrow 0$, using Fatou's lemma.
\end{proof}
%-------------------%

%-------------------%
%:cor A.5
%-------------------%
\begin{corollary}\label{cor:deriv}
Under the same assumptions on $V$, if $\psi\in H^{2m+4}(\D{R}^d)$ and $x^m\psi\in L^2(\D{R}^d)$, then
\begin{equation}\label{eq:dercont}
\frac{\dd}{\dd t}x^m(t)\psi=\ii\eul^{\ii tH}[-\Delta,x^m]\eul^{-\ii tH}\psi.
\end{equation}
\end{corollary}
%-------------------%

%-------------------%
\begin{proof}
The proof is the same as that of Corollary~\ref{cor:deri}, using Theorem~\ref{thm:momconts}. In more details, the fact that $x^m(s)\psi\to x^m(t)\psi$ as $s\to t$ for $\psi\in H^{2m+2}(\D{R}^d)$ is proved the same way by considering $f_\epsilon(x)\coloneqq\frac{x^m}{\sqrt{1+\epsilon x^{2m}}}$ instead. For the derivative, to deal with the second and third terms at the end of \eqref{eq:derfo}, we use that $Hx^m=x^m H + [-\Delta,x^m]$ instead. The term $x^mH$ is dealt with as before. For the second term $[-\Delta,x^m]$, let $\phi_t^\delta\coloneqq\eul^{-\ii tH}\psi-\eul^{-\ii(t+\delta)H}\psi$. We have $\norm{[-\Delta,x^m]\phi_t^\delta}\leq\sum_{i=1}^d [m(m-1)\norm{x_i^{m-2}\phi_t^\delta}+ 2m\norm{x_i^{m-1}\partial_{x_i}\phi_t^\delta}]$. The calculations \eqref{eq:toapply}--\eqref{eq:compa} and their later generalization to all $m$ imply that we may bound the second term by $C\norm{x^m\phi_t^{\delta}}^{1-\frac{1}{m}}\norm{D^m\phi_t^{\delta}}^{\frac{1}{m}}$. The norm $\norm{x^p\phi_t^{\delta}}\to 0$ as $\delta\to 0$ for $p=m-2,m$ using the analog of \eqref{eq:contmoms}, while $\norm{D^m\phi_t^\delta}$ is controlled using \eqref{eq:ptr}. Finally the last term in \eqref{eq:derfo} is controlled using the same Fatou argument and we get \eqref{eq:dercont}.
\end{proof}
%-------------------%

%---------------------------------------------------------%
%:bib
%---------------------------------------------------------%

%---------------------------------------------------------%
\end{document}